\documentclass[12pt]{amsart}
\usepackage{amsmath,amsfonts,amssymb} \usepackage{color}
\usepackage[all,cmtip]{xy}
\usepackage{graphicx}

\newcommand{\corr}[1]{\langle {#1} \rangle}
\newcommand{\corrr}[1]{\langle\langle {#1} \rangle\rangle}

 \newcommand{\bR}{\mathbb{R}} \newcommand{\cF}{\mathcal{F}} 
 
\newcommand{\bt}{{\bf t}}

\newcommand{\bE}{\mathbb{E}}
\newcommand{\bP}{\mathbb{P}}
 \newcommand{\cO}{\mathcal{O}} 
   
  \newcommand{\cM}{\mathcal{M}} \newcommand{\cR}{\mathcal{R}}
 \newcommand{\cU}{\mathcal{U}}
 \newcommand{\cV}{\mathcal{V}}
   \newcommand{\cX}{\mathcal{X}}
 \newcommand{\bZ}{\mathbb{Z}}
 \newcommand{\bC}{\mathbb{C}}
 \newcommand{\pd}{\partial}
\newcommand{\Mbar}{\overline{\mathcal M}}

\DeclareMathOperator{\End}{End}
  \DeclareMathOperator{\res}{res}
 
 \DeclareMathOperator{\diag}{diag}

\newcommand{\be}{\begin{equation}}
\newcommand{\ee}{\end{equation}}
\newcommand{\bea}{\begin{eqnarray}}
\newcommand{\eea}{\end{eqnarray}}
\newcommand{\ben}{\begin{eqnarray*}}
\newcommand{\een}{\end{eqnarray*}}

\newcommand{\half}{\frac{1}{2}}

\newtheorem{cor}{Corollary}[section]

 \newtheorem{prop}[cor]{Proposition}
 \newtheorem{thm}[cor]{Theorem}
\theoremstyle{remark}

 \newtheorem{rmk}[cor]{Remark}

\definecolor{A}{rgb}{.75,1,.75}

\definecolor{green}{rgb}{0,1,0}
\definecolor{yellow}{rgb}{1,1,0}
\definecolor{orange}{rgb}{1,.7,0}
\definecolor{red}{rgb}{1,0,0}
\definecolor{white}{rgb}{1,1,1}

 \makeindex
\begin{document}
\title
[Emergent Geometry of the Quintic]
{On Emergent Geometry of the Gromov-Witten Theory of Quintic Calabi-Yau Threefold}
%\author{ }
%\thanks{ }

\author{Jian Zhou}
\address{Department of Mathematical Sciences\\Tsinghua University\\Beijng, 100084, China}
\email{jianzhou@mail.tsinghua.edu.cn}

\begin{abstract}
We carry out the explicit computations that are used
to write down the integrable hierarchy associated with
the quintic Calabi-Yau threefold.
We also do the calculations for the geometric structures
emerging in the Gromov-Witten theory of the quintic,
such as the Frobenius manifold structure
and the special K\"ahler structure.
\end{abstract}

\dedicatory{Dedicated to the memory of Professor Boris Dubrovin}

\maketitle

\section{Introduction}

Gromov-Witten theory is the mathematical theory
for the physical theory of
topological sigma models coupled to two-dimensional topological
gravity.
The observables in a  topological sigma models,
called the {\em primary obsrvable},
correspond to the cohomological classes of the underlying
symplectic manifold, and they form a finite dimensional
space called the {\em small phase spaces}.
The coupling with the two-dimensional topological gravity
associates to each primary observable a sequence
of observables, called their {\em gravitational descendants},
which form an infinite-dimensional space called the
{\em big phase space}.
One assigns a coupling constant to each observable,
so one gets infinitely many formal variables
as linear coordinates on the big phase space.
The goal of the Gromov-Witten theory is to study
the partition functions and the free energy
of the topological sigma model coupled to the two-dimensional
topological gravity
as formal power series in the infinitely many coupling constants.

There are two major paradigms in the study of Gromov-Witten theory.
The first is the Witten Conjecture/Kontsevich Theorem
\cite{Wit, Kon} which relates
the two-dimensional topological gravity
to the KdV hierarachy.
Since it involves integrable hierarchy,
I will refer to it as the {\em integrable hierarchy paradigm}.
The second is the mirror symmetry   \cite{CDGP} which identifies the Gromov-Witten
theory of the quintic Calabi-Yau threefold with
another more computable theory.
I will refer to it as the {\em mirror symmetry paradigm}.
For the past thirty years,
Gromov-Witten theory has been developed to provide
the justifications and generalizations of these discoveries
by physicists by  mathematically rigorous methods.

More recently,
the author has worked on a program to unify these two paradigms.
This of course has to go in two directions:
(a) Understand integrable hierarchies from
the point of view of mirror symmetry;
(b) Understand mirror symmetry from the point of view
of integrable hierarchy.
At first sight,
there are some manifest differences between the two paradigms.
The integrable hierarchy paradigm predicts that the partition
function of the Gromov-Witten invariants
of each symplectic manifold
is the tau-function of a suitable integrable
hierarchy,
specified by suitable constraints.
It naturally involves infinitely many variables
and infinite-dimensional Lie algebras that describes
the symmetries of the relevant integrable hierarchy.
On the contrary,
the hallmark example of the mirror symmetry
of the quintic Calabi-Yau three-folds
focuses on the computation of the free energy function
of the Gromov-Witten invariants
restricted to a one-dimensional subspace
of the small phase space.
But the key feature of mirror symmetry is that it predicts
the equivalence with a different theory,
called its {\em mirror theory}.
Furthermore,
in the physics literature the differential geometry of moduli spaces play an important role.

Therefore,
a first step towards unification consists of three problems:
Problem 1. Provide
a construction of a mirror theory for each integrable hierarchy;
Problem 2. Find the integrable hierarchy associated with
the quintic Calabi-Yau threefold.
Problem 3. Discover the geometric objects hidden behind the integrable hierarchies.

One of the purpose of this paper is to report our results on
Problem 2 and Problem 3.
However,
our approach  used in this work
are based on our understandings
obtained in our work on Problem 1.
We will postpone a detailed review of some progress there in Section \ref{sec:Conclusions},
here we will just mention that in our earlier work on Problem 1 we have borrowed
some ideas from statistical physics and introduced the notion of emergent geometry.

It turns out that the idea of using statistical physics
becomes the crucial point of departure for our work on Problem 2 and Problem 3 in this paper.
More generally,
one can consider the problem of finding the integrable hierarchy
associated to the Gromov-Witten theory of any compact symplectic manifold.

Like ordinary mean field theory,
the mean field theory of Gromov-Witten theory
discovered by Dijkgraaf and Witten \cite{Dij-Wit, Wit}
reduces a problem with an infinite degree of freedom
to a problem with only finite degree of freedom.
More precisely,
they introduced finitely many order parameters that encode all the information
in genus zero.
The dependence of these order parameters on the infinitely many coupling constants
can be determined in two different ways.
First they are the critical points of a Landau-Ginzburg potential functions.
Secondly, they satisfy a sequence of evolution equations in the coupling constants.
Furthermore,
the integrable hierarchy is expected to be generalized to arbitrary genus.

We reinterpret Dubrovin's theory of Frobenius manifolds as part of emergent geometry
of GW theory. The original goal of this theory is to reconstruct the whole theory
from the genus zero part of the theory restricted to the small phase space.
We reverse the direction.
We start with the GW theory of a symplectic manifold and search
for geometric structures that naturally emerge.
Frobenius manifold structure is then an unavoidable candidate.

The rest of this paper consist of sections of two different natures.
In some sections we review general theory reinterpreted from an emergent point of view,
and in some sections we focus on the detailed computations for the quintic.
More precisely,
in Section \ref{sec:MFT} we review the mean field theory as developed
by Dijkgraaf-Witten \cite{Dij-Wit} and Witten \cite{Wit}.
The results for higher genera will be reviewed in Section \ref{sec:Higher}
where we also present an operator formalism to compute the $n$-point functions
in arbitrary genera.
The corresponding  concrete computations for the quintic
are presented in Section \ref{sec:Quintic-0} for genus zero and in Section \ref{sec:Quintic-Higher}
respectively.
As our key results for these two Sections,
we write down the Landau-Gingzburg potential function and the integrable hierarchy
for the quintic.
We briefly review some aspects of the Frobenius manifold theory useful for our purpose
in Section \ref{sec:Frobenius},
then we do the concrete computations for the quintic in Section \ref{sec:Quintic-Frob}.
In Section \ref{sec:Quintic-Spec} we discuss the emergent special geometry for the quintic.
We make some conclusions and present some speculations for further investigations in
the final Section \ref{sec:Conclusions}.

\section{Mean Field Theory of Gromov-Witten Theory}

\label{sec:MFT}

In this Section we summarize
the mean field theory for  Gromov-Witten theory
developed by Dijkgraaf-Witten \cite{Dij-Wit}
and  Witten \cite{Wit}.
We emphasize that we want to understand their results
as examples of emergent phenomena
from the point of view of statistical physics.
In the beginning we are working on an infinite-dimensional
phase space,
dealing with a formal power series with infinitely many
formal variables called the coupling constants.
In the end we know that we need to only work with
finitely many functions called the order parameters.

The dependence of the order parameters
on the coupling constants can be determined in two different ways.
First,
the order parameters are the critical points
of  a Landau-Ginzburg potential function.
Secondly, the order parameters are governed by
a system of infinitely many evolutive equations,
one for each coupling constant.

\subsection{The Gromov-Witten invaraints as correlators}

Let $M$ be a compact symplectic manifold of dimension $2m$.
For simplicity,
we will only consider its Gromov-Witten invariants that involve
cohomology classes of even degrees.
So the small phase space is
\be
H^{ev}(M):=\bigoplus_{k=0}^m H^{2k}(M).
\ee
For $\alpha \in H^{2k}(M)$, set $\deg \alpha: = k$.
Fix a basis $\{\cO_j\}_{j=0}^r$ of $H^{ev}(M)$ such that $\cO_0=1$.

The big phase space
is the space $H^{ev}(M)[[z]]$
of formal power series with coefficients in $H^{ev}(M)$.
For $\omega\in H^{ev}(M)$,
rewrite $\omega z^n$ as $\tau_n(\alpha)$.
It is called the $n$-th gravitational descendant of $\omega$.

For classes $\omega_1, \dots, \omega_n\in H^{ev}(M)$,
and integers $a_1, \dots, a_n \geq 0$,
the Gromov-Witten invariants are defined by
\be
\corr{\tau_{a_1}(\omega_1),\dots, \tau_{a_n}(\omega_n)}_{g,n; \beta}
:= \int_{[\Mbar_{g,n}(M;\beta)]^{virt}}
\omega_1\psi_1^{a_1} \cdots \omega_n\psi_n^{a_n},
\ee
where
$[\Mbar_{g,n}(M;\beta)]^{virt}$ is the virtual fundamental class
of stable maps of genus $g$ with $n$ marked points
in the homology class  $\beta \in H_2(M;\bZ)$.

The degree of $[\Mbar_{g,n}(M;\beta)]^{virt}$ is the virtual
dimension of $\Mbar_{g,n}(M;\beta)$, given by the formula:
\be
(m-3)(1-g) + n + \int_\beta c_1(M).
\ee
So the correlator
$\corr{\tau_{a_1}(\omega_1),\dots, \tau_{a_n}(\omega_n)}_{g,n; \beta}$
vanishes unless the following {\em selection rule}
is satisfied:
\be \label{eqn:Selection-Rules}
\sum_{i=1}^n (\deg \omega_i + a_i)
= (m-3)(1-g) + n + \int_\beta c_1(M).
\ee

\subsection{The free energy, the partition function,
and the $n$-point correlators}

Denote by $t^\alpha_n$ the coupling constant associated with
$\tau_n(\cO_\alpha)$.
The genus $g$ free energy of the Gromov-Witten theory
of $M$ is defined by:
\be
F_g = \sum_{n \geq 1} \sum_{\alpha_1, \dots, \alpha_n=0}^r
\sum_{\beta \in H_2(M;\bZ)}
\frac{\prod\limits_{j=1}^n t_{a_j}^{\alpha_j}}{n!}
q^\beta \corr{\tau_{a_1}(\cO_{\alpha_1}) \cdots
\tau_{a_n}(\cO_{\alpha_n})}_{g,n;\beta}.
\ee
For $g=0$ and $\beta = 0$, the summation starts at $n=3$.
The total free energy is defined by:
\be
F_g : = \sum_{g=0}^\infty \lambda^{2g-2} F_g.
\ee

The free energy contains all the information in the theory.
From it one can define the $n$-point correlators
in genus $g$:
\be
\corrr{\tau_{a_1}(\cO_{\alpha_1}) \cdots
\tau_{a_n}(\cO_{\alpha_n})}_g:
= \frac{\pd^nF_g}{\pd t^{\alpha_1}_{a_1} \cdots
\pd t^{\alpha_n}_{a_n} }.
\ee
The partition function of the Gromov-Witten theory is defined by
\be
Z = \exp F.
\ee

The ultimate goal of Gromov-Witten theory is of course
to have the ability to compute the free energy or
the partition function  or the $n$-point correlators
in closed forms.
Because they all involve infinitely many formal variables,
this does not seem to be possible.
As mentioned in the Introduction,
Dijkgraaf and Witten \cite{Dij-Wit} observed that
one can borrow
ideas from mean field theory in statistical physics
to make this possible.

\subsection{Universal relations among correlators}

The reason that the approach of Dijkgraaf and Witten works
is because of the topological nature
of the Gromov-Witten theory,
in other words,
it is a topologically twisted
$N=2$ superconformal field theory coupled with
2d topological gravity.
Such a theory is based on the cohomology theory
of the Deligne-Mumford moduli spaces $\Mbar_{g,n}$
of algebraic curves.
There are some universal relations among classes
on $\Mbar_{g,n}$ called tautological relations
(see e.g.  \cite{PPZ} for references, see also \cite{Lin-Zho}).
They lead to some universal relations among the correlators
in Gromov-Witten theory.
The following three equations hold in all genera.
The first is the  {\em puncture equation}:
\be \label{eqn:Puncture-Corr}
\begin{split}
& \corr{1\tau_{a_1}(\cO_{\alpha_1}) \cdots
\tau_{a_n}(\cO_{\alpha_n})}_{g,n; \beta} \\
=& \sum_{i=1}^n
\corr{\tau_{ a_1}(\cO_{\alpha_1}),\dots,
\tau_{a_i-1}(\cO_{\alpha_i}),
\dots, \tau_{a_n}(\cO_{\alpha_n})}_{g,n; \beta},
\end{split}
\ee
with the exceptional case:
\be
\corr{1\cO_{\alpha_1}\cO_{\alpha_2}}_{0,2;0}
= \corr{\cO_{\alpha_1}\cO_{\alpha_2}},
\ee
where the right-hand side is the Poincar\'e pairing
of the two classes $\cO_{\alpha_1}$ with $\cO_{\alpha_2}$.
The second is the {\em dilaton equation}:
\be
\begin{split}
& \corr{\tau_1(1),\tau_{a_1}(\alpha_1),\dots \tau_{a_n}(\alpha_n)}_{g,n; \beta} \\
=& (2g-2+n) \cdot
\corr{\tau_{a_1}(\alpha_1),\dots, \tau_{a_i-1}(\alpha_i),
\dots, \tau_{a_n}(\alpha_n)}_{g,n; \beta},
\end{split}
\ee
with the exceptional case:
\be
\corr{\tau_1(1)}_{1,1;0} = \frac{1}{24} \chi(M).
\ee
In terms of the free energy,
these equations can be rewritten as follows:
\bea
&& \frac{\pd F_g}{\pd t_0^0}
= \sum_{n=0}^\infty \sum_{\alpha=0}^r
t_{n+1}^\alpha \frac{\pd F_g}{\pd t_n^\alpha}
+ \delta_{g,0} \half \eta_{\alpha\beta}t_0^\alpha t_0^\beta,
\label{eqn:Puncture}\\
&& \frac{\pd F_g}{\pd t_1^0}
= \sum_{n=0}^\infty \sum_{\alpha=0}^r
t_{n}^\alpha \frac{\pd F_g}{\pd t_n^\alpha}
+ (2g-2)F_g
+ \delta_{g,1} \frac{1}{24} \chi(M),
\label{eqn:Dilaton}
\eea
where $\eta_{\alpha\beta} = \corr{\cO_\alpha, \cO_\beta}_{0,2;0}$.
The third is the {\em divisor equation}:
\be \label{eqn:Divisor}
\begin{split}
& \corr{D\tau_{a_1}(\cO_{\alpha_1})\dots
\tau_{a_n}(\cO_{\alpha_n})}_{g,n; \beta} \\
= & \corr{D, \beta} \cdot
\corr{\tau_{a_1}(\alpha_1)\dots \tau_{a_n}(\alpha_n)}_{g,n; \beta} \\
+& \sum_{i=1}^n
\corr{\tau_{a_1}(\cO_{\alpha_1}),\dots,
\tau_{a_i-1}(\cO_{\alpha_i}D), \dots,
\tau_{a_n}(\cO_{\alpha_n})}_{g,n; \beta}.
\end{split}
\ee

\subsection{Topological recursion relations} \label{sec:TRR}
The {\em topological recursion relation} (TRR) in genus zero is the following
equation:
\be \label{eqn:TRR0}
\frac{\pd^3F_0}{\pd t^\alpha_a\pd t^\beta_b \pd t^\gamma_c}
= \frac{\pd^2F_0}{\pd t^\alpha_{a-1}\pd t^\mu}
\eta^{\mu\nu}
\frac{\pd^3F_0}{\pd t^\nu\pd t^\beta_b\pd t^\gamma_c}.
\ee
Witten \cite{Wit} has shown that it implies the generalized WDVV equations:
\be \label{eqn:Gen-WDVV}
\frac{\pd^3F_0}{\pd t^\alpha_a\pd t^\beta_b \pd t^\mu}\eta^{\mu\nu}
\frac{\pd^3F_0}{\pd t^\nu\pd t^\gamma_c\pd t^\delta_d}
= \frac{\pd^3F_0}{\pd t^\alpha_a\pd t^\gamma_c \pd t^\mu}\eta^{\mu\nu}
\frac{\pd^3F_0}{\pd t^\nu\pd t^\beta_b\pd t^\delta_d}.
\ee
When $a=b=c=d=0$,
this reduces to the WDVV equations:
\be \label{eqn:WDVV}
\frac{\pd^3F_0}{\pd t^\alpha\pd t^\beta \pd t^\mu}\eta^{\mu\nu}
\frac{\pd^3F_0}{\pd t^\nu\pd t^\gamma\pd t^\delta}
= \frac{\pd^3F_0}{\pd t^\alpha\pd t^\gamma \pd t^\mu}\eta^{\mu\nu}
\frac{\pd^3F_0}{\pd t^\nu\pd t^\beta\pd t^\delta}.
\ee
In genus one,
Witten proposed the following TRR \cite[(3.48)]{Wit}:
\be
\frac{\pd F_1}{\pd t^\alpha_n}
= \frac{\pd^2 F_1}{\pd t^\alpha_{n-1}\pd t^\beta}_0
\eta^{\beta\gamma} \frac{\pd F_1}{\pd t^\gamma}
+ \frac{1}{24} \frac{\pd^3F_1}{\pd t^\alpha_{n-1}\pd t^\beta\pd t^\gamma}.
\ee
For its mathematical proof, see Getzler \cite{Get1}.
See also \cite{Get2} for Getzler's TRRs in genus two.

\subsection{Quantum cohomology}
\label{sec:Quant-Cohom}

The WDVV equations have the following algebraic interpretation.
For $\bt \in H^{ev}(M)$,
define
\be
\frac{\pd}{\pd t^\alpha} \circ_\bt \frac{\pd}{\pd t^\beta}
: = c^\gamma_{\alpha\beta}(\bt) \frac{\pd}{\pd t^\gamma},
\ee
where the coefficients $c^\gamma_{\alpha\beta}(\bt)$ are defined by:
\be
c^\gamma_{\alpha\beta}(\bt)
:= \frac{\pd^3F_0(\bt)}{\pd t^\alpha\pd t^\beta\pd t^\delta}\eta^{\delta\gamma}.
\ee
Then $\circ_\bt$ is a commutative multiplication on $H^{e}(M)$.
By WDVV equations,
it is also associative.
Furthermore,
because
\be
\frac{\pd^3F_0(\bt)}{\pd t^0\pd t^\alpha \pd t^\beta} = \eta_{\alpha\beta},
\ee
$\frac{\pd}{\pd t^0}$ is a unit for all $\circ_\bt$,
i.e.,
\be
\frac{\pd}{\pd t^0} \circ_\bt \frac{\pd}{\pd t^\alpha} = \frac{\pd}{\pd t^\alpha}.
\ee
For simplicity of notations we will often write $\circ$ for $\circ_\bt$.

\subsection{Order parameters in GW theory}

Following Dijkgraaf and Witten \cite{Dij-Wit},
choose the order parameters in GW theory to be:
\be
u_\alpha: = \frac{\pd^2F_0}{\pd t^1\pd t^\alpha}
= \corrr{\cO_0\cO_\alpha}_0.
\ee
The genus zero free energy $F_0$ consists of two parts:
The classical part $F_{classical}$ consists
of $3$-point correlators
\be
\begin{split}
F_{classical} & = \sum_{\alpha_1, \alpha_2, \alpha_3=0}^r
\frac{t^{\alpha_1}t^{\alpha_2}t^{\alpha_3}}{3!}
\corr{\cO(\alpha_1), \cO(\alpha_2), \cO_{\alpha_3}}_{0,3;0}  \\
& = \sum_{\alpha_1, \alpha_2, \alpha_3=0}^r
\frac{t^{\alpha_1}t^{\alpha_2}t^{\alpha_3}}{3!}
\int_M \cO_{\alpha_1} \cO_{\alpha_2} \cO_{\alpha_3},
\end{split}
\ee
and $F_{quantum}$  consists
of $n$-point correlators for $n > 3$.
And so
\be
u_\alpha = \eta_{\alpha\beta} t^\beta + \cdots,
\ee
where $\cdots$ involves correlators of the form:
\ben
&& \corr{1, \cO_\alpha, \tau_{a_1}(\cO_{\alpha_1}), \dots,
\tau_{a_{n-2}}(\cO_{\alpha_{n-2}})}_{0,n;\beta}
\een
for $n > 3$.
Note by the puncture equation \eqref{eqn:Puncture-Corr},
\ben
&& \corr{1, \cO_\alpha, \cO_{\alpha_1}, \dots,
\cO_{\alpha_{n-2}}}_0
\een
for $n > 3$.
Therefore,
on the small phase space,
i.e.,
after setting $t^\alpha_n = 0$ for all $n > 0$ and all $\alpha=0,\dots, r$,
we have:
\be
u_\alpha = \eta_{\alpha\beta} t^\beta.
\ee
Equivalently,
on the small phase space,
$t^\beta$ is equal to $\eta^{\alpha\beta}u_\alpha$.

\subsection{Constitutive relations}

Let us now recall a crucial observation of Dijkgraaf and Witten \cite{Dij-Wit}.
On the small phase space
$\frac{\pd^2F_0}{\pd t_a^\alpha \pd t_b^\beta}$
is a formal power series in  $\{t^\alpha\}_{\alpha=0}^r$,
so one can write:
\be
\frac{\pd^2F_0}{\pd t_a^\alpha\pd t_b^\beta}
\biggl|_{t^\gamma_n = 0, n > 0, \gamma=0,\dots, r}
= R_{\alpha,a;\beta, b}(t^0, \dots, t^r).
\ee
Then Dijkgraaf and Witten \cite{Dij-Wit} showed that
on the big phase space
\be \label{eqn:Constitutive}
\frac{\pd^2F_0}{\pd t_a^\alpha\pd t_b^\beta}
= R_{\alpha,a;\beta, b}(u^0, \dots, u^r),
\ee
where $u^\alpha = \eta^{\alpha\beta} u_\beta$.
They call these the {\em constitutive relations}.
Their proof is based on the TRR in genus zero \eqref{eqn:TRR0}.

\subsection{Mean field theory of GW theory}

Let us recall how Dijkgraaf and Witten derived the Landau-Ginzburg
equations for the order parameters,
based on the puncture equation \eqref{eqn:Puncture}
and the constitutive relation \eqref{eqn:Constitutive}
derived by TRR in genus zero \eqref{eqn:TRR0}.

Take $\frac{\pd}{\pd t^\alpha}$ on both sides
of \eqref{eqn:Puncture} to get:
\ben
&& \frac{\pd^2 F_0}{\pd t^0\pd t^\alpha}
= \sum_{n=0}^\infty \sum_{\beta=0}^r
t_{n+1}^\beta \frac{\pd^2 F_0}{\pd t^\alpha\pd t_n^\beta}
+ \eta^{\alpha\beta} t^\beta,
\een
now plug in the constitutive relations \eqref{eqn:Constitutive}
to get:
\be \label{eqn:LG}
u_\alpha
= t_\alpha + \sum_{n=0}^\infty \sum_{\beta=0}^r
t_{n+1}^\beta  R_{\alpha,0;\beta, n}(u_0, \dots, u_r),
\ee
for $\alpha = 0, 1, \dots, r$.
These are the {\em Landau-Ginzburg equations}
for the order parameters $u_0, \dots, u_r$.
When $R_{\alpha, 0; \beta, n}$'s have been computed,
one can solve for $u_\alpha$ as formal power series
in $\{t^\alpha_n\}$.

\subsection{The Landau-Ginzburg potential}

This line of ideas were further developed in Witten \cite{Wit}.
Recall that on the small phase space we have:
\ben
R_{\alpha, 0; \beta, n}(t^0, \dots, t^r)
= \frac{\pd^2F_0}{\pd t^\alpha\pd t^\beta_n}\biggr|_{t^\gamma_n =0,
n > 0, \gamma=0, \dots, r},
\een
Therefore,
if we write
\be
R_{\beta, n}(t^0, \dots, t^r)
= \frac{\pd F_0}{\pd t^\beta_n}\biggr|_{t^\gamma_n =0,
n > 0, \gamma=0, \dots, r}.
\ee
Then we have the following relation:
\ben
R_{\alpha, 0; \beta, n}(t_0, \dots, t^R)
& = & \frac{\pd}{\pd t^\alpha}
\biggl( \frac{\pd F_0}{\pd t^\beta_n}\biggr|_{t^\gamma_n =0,
n > 0, \gamma=0, \dots, r}\biggr) \\
& = & \frac{\pd}{\pd t^\alpha} R_{\beta; n}(t^0, \dots, t^n).
\een
Therefore,
\be
R_{\alpha, 0; \beta, n}(u^0, \dots, u^r)
=  \frac{\pd}{\pd u^\alpha}
R_{\beta; n}(u^0, \dots, u^n).
\ee
Plug this into the Landau-Ginzburg equations \eqref{eqn:LG}:
\be
u_\alpha
= t_\alpha + \sum_{n=0}^\infty \sum_{\beta=0}^r
t_{n+1}^\beta \frac{\pd}{\pd u^\alpha}
R_{\beta; n}(u^0, \dots, u^n).
\ee
Therefore,
if one sets (cf. Witten \cite[(3.40)]{Wit}):
\be \label{eqn:LG-Potential}
W = - \half U_\alpha U^\alpha
+ t^\beta U_\beta
+ \sum_{n=0}^\infty \sum_{\beta=0}^r t^\beta_{n+1}
R_{\beta, n}(U^0, \dots, U^r),
\ee
then the equations for the critical point of $W$
are exactly the Landau-Ginzburg equations \eqref{eqn:LG}.

\begin{rmk}
If we set
$$R_{\beta, -1}(U^0, \dots, U^r) = U_\beta,$$
and make the dilaton shift:
\be
\tilde{t}^\beta_n = t^\beta_n - \delta_{\beta, 0}\delta_{n,1},
\ee
then the Landau-Ginzburg potential can be written
in a more compact form:
\be
W = \sum_{n=-1}^\infty \sum_{\beta=0}^r
\tilde{t}^\beta_{n+1}
R_{\beta, n}(U^0, \dots, U^r).
\ee
Therefore,
if one sets
\be
R_\beta(U;z):=\sum_{n=-1}^\infty R_{\beta, n}(U) z^{-n-1}
\ee
and introduces the {\em source operators}:
\be \label{eqn:SourceOp}
J^\beta(z): =\sum_{n=0}^\infty \tilde{t}^{\beta}_n z^{n},
\ee
then one rewrite
\eqref{eqn:LG-Potential} as a residue:
\be
W(U) = \sum_{\beta=0}^r \res_{z=0}
( J^\beta(z) R_\beta(U;z) \frac{dz}{z}).
\ee
\end{rmk}

\subsection{Integrable hierarchy for genus zero GW invariants}
\label{sec:IH-g=0}

The next observation in Witten \cite{Wit} is the emergence
of an integrable hierarchy for genus zero GW invariants:
\ben
\frac{\pd}{\pd t^\alpha} u_\beta
= \frac{\pd}{\pd t^\alpha} \biggl(\frac{\pd^2F_0}{\pd t^0 \pd t^\beta} \biggr)
= \frac{\pd}{\pd t^0}
\biggl( \frac{\pd^2 F_0}{\pd t^\alpha\pd t^\beta_n}\biggr).
\een
Therefore, by the constitutive relations \eqref{eqn:Constitutive}
he obtained a sequence of equations:
\be \label{eqn:IH}
\frac{\pd u_\beta}{\pd t_n^\alpha}
= \frac{\pd}{\pd t^0} R_{\alpha,0; \beta, n}(u^0, \dots, u^r).
\ee
Witten \cite{Wit} gave the following Hamiltonian formulation
for this integrable hierarchy.
Regard $u_\alpha$ as functions of $x= t^0$,
and define the following Poisson bracket on the loop space
of the  space of order parameters:
\be
\{u_\alpha(x), u_\beta(y)\} = \eta_{\alpha\beta} \pd_x \delta(x-y).
\ee
Then \eqref{eqn:IH} is the system of Hamiltonian of motion
\be
\frac{\pd u_\alpha}{\pd t^\beta_n}
= \{u_\alpha, H_{\beta, n}\},
\ee
where $H_{\beta,n}$ is the Hamiltonian functional
\be
H_{\beta,n} = \int dx R_{\beta,n}(u^0, \dots, u^r).
\ee

\section{Renormalization Theory, New Coordinates on the Big Phase Space
and Operator Formalism of Gromov-Witten Theory}

\label{sec:Higher}

We will present in this Section a method
that computes the $n$-point correlation functions
on the big phase space
based on some works of Eguchi and his collaborators.
It is a continuation of the development of the mean field
theory of GW theory to general $n$-point functions in arbitrary
genera.
We will use an operator formalism that leads to
an interpretation of the partition function
as an element in a bosonic Fock space.

\subsection{Mean field theory for higher genera}

Dijkgraaf and Witten computed the genus one free energy
of the two-dimensional topological gravity
by TRR in genus one.
In general,
Witten \cite{Wit} claimed that if one sets
\be
M_{\alpha\beta} = \frac{\pd^3F_0}{\pd t^0\pd t^\alpha\pd t^\beta}
= \frac{\pd}{\pd t^0} R_{\alpha,\beta}(u_0, \dots, u_r),
\ee
then
\be
F_1 = \frac{1}{24}\log \det (M_{\alpha\beta})
+ E(u_0, \dots, u_r)
\ee
for some function $E$ depending only on $u_\alpha$.
See Eguchi-Getzler-Xiong \cite{Egu-Get-Xio} for a proof.

For higher genera,
Witten \cite{Wit} introduced
\be
u_{\alpha, n}:= \frac{\pd^n}{\pd (t^0)^n}u_\alpha,
\ee
and set $\deg u_{\alpha, n}=n$,
then he conjectured that there are
functions $R_{\alpha,m; \beta,n}^{(g)}$ of degree $2g$
in $\{u_{\alpha,n}\}$ such that
\be
\frac{\pd^2F_g}{\pd t^\alpha_m\pd t^\beta_n}
= R_{\alpha,m; \beta,n}^{(g)}(\{u_{\alpha,n}\}).
\ee
This conjecture implies that one can associate an integrable
hierarchy of differential equations
associated to every compact symplectic manifold
via the GW theory.
One simply sets
\be
U_\alpha = \sum_{g=0}^\infty \lambda^{2g} u_{\alpha}^{(g)},
\ee
where $u_\alpha^{(g)}$ is defined by:
\be
u_{\alpha}^{(g)} = \frac{\pd^2F_g}{\pd t^0\pd t^\alpha},
\ee
then one gets a system of equations:
\be
\frac{\pd U_\alpha}{\pd t^\beta_n}
= \pd_x \cR_{\alpha, 0; \beta, n}
\ee
where
\be
\cR_{\alpha, m; \beta, n} = \sum_{g=0}^\infty \lambda^{2g}
R_{\alpha,m; \beta, n}^{(g)}.
\ee

\subsection{Renormalized coupling constants
as new coordinate systems on the big phase space}
\label{sec:2DGravity}

In the case of two-dimensional topological gravity,
Itzykson and Zuber \cite{Itz-Zub} noted that the following ansatz
is compatible with the KdV hierarchy:
\bea
&& F_1 = \frac{1}{24} \log \frac{1}{1 - I_1}, \\
&& F_g
=  \sum_{\sum_{2\leq k \leq 3g-2} (k-1)l_k=3g-3}
\frac{
\corr{\tau_2^{l_2}\tau_3^{l_3}\cdots \tau_{3g-2}^{l_{3g-2}} }_g
}{(1 - I_1)^{2(g-1)}}
\frac{I^{l_2}_2}{l_2!}
\frac{I^{l_3}_3}{l_3!} \cdots
\frac{I^{l_{3g-2}}_{3g-2}}{l_{3g-2}!},
\eea
where  $I_n$ are defined by:
\be \label{eqn:Ik}
I_k = \sum_{p=0}^\infty t_{k+p} \frac{u_0^p}{p!},
\ee
and $u_0$ is determined by:
\be \label{eqn:u0}
u_0 = \sum_{p=0}^\infty t_{p} \frac{u_0^p}{p!}.
\ee
Eguchi, Yamada and Yang \cite{Egu-Yam-Yan} showed that
there exist formulas of the form:
\be
\begin{split}
F_g
= & \sum_{\sum_k (k-1)l_k=3g-3}
a_{l_2\dots l_{3g-2}}
\frac{u_2^{l_2} \cdots u_{3g-2}^{l_{3g-2}}}
{u_1^{2(1-g)+\sum kl_k}} \\
= &  \sum_{\sum_k (k-1)l_k=3g-3}
b_{l_2\dots l_{3g-2}}
\frac{I_2^{l_2} \cdots I_{3g-2}^{l_{3g-2}}}
{u_1^{2(1-g)+\sum kl_k}}
\end{split}
\ee
for $g > 1$,
where $u_k = \frac{\pd^k u_0}{\pd x^k}$.

They also discussed the $O(N)$ vector model \cite{Nis-Yon}
and noted results similar to that of Itzykson-Zuber in the case
of 2D topological gravity also hold.
This was further studied by the author in \cite{Zhou-1D}.
The new interpretation in that work is that
starting with an action:
\be \label{eqn:Action}
S = - \frac{1}{2}x^2 + \sum_{n \geq 1} t_{n-1} \frac{x^n}{n!},
\ee
one can perform a sequence of renormalizations of the coupling
constants to reach in the limit:
\be
S=  \sum_{k=0}^\infty  \frac{(-1)^k}{(k+1)!}
(I_k+\delta_{k,1}) I_0^{k+1}
- \frac{1}{2}(x-I_0)^2
+ \sum_{n=1}^\infty  I_n \frac{(x-I_0)^{n+1}}{(n+1)!}  .
\ee
So the series $\{I_k\}$ can be interpreted as renormalized
coupling constants.

The Landau-Gingzburg equation \eqref{eqn:u0} has been explicitly
solved in \cite{Zhou-1D} using Lagrange inversion formula:
\be \label{eqn:I0}
u_0=I_0 = \sum_{k=1}^\infty \frac{1}{k}
\sum_{p_1 + \cdots + p_k = k-1} \frac{t_{p_1}}{p_1!} \cdots
\frac{t_{p_k}}{p_k!}.
\ee
Plugging this into \eqref{eqn:Ik},
one can express each $I_k$ as a formal power series
in $\{t_n\}_{n \geq 0}$.
Conversely,
we also showed in \cite{Zhou-1D} that
\be \label{eqn:T-in-I}
t_k = \sum_{n=0}^\infty \frac{(-1)^n I_0^n}{n!}I_{n+k}.
\ee
Therefore,
the renormalized coupling constants can be used
as new coordinates on the big phase space in this case.
It follows that one can express the Virasoro operators
in the case of 2D topological gravity in terms of these new coordinates.
In a joint work with Qisheng Zhang \cite{Zhang-Zhou},
we have shown that one can use such
expressions to prove the Itzykson-Zuber Ansatz and
to find effective algorithms to compute $F_g$  in terms of
$I_k$'s.

\subsection{Jet variables as new coordinate systems on the big phase space}

In general, Eguchi and Xiong \cite{Egu-Xio}
conjectured that: If one sets
\be
u_{\alpha_1\alpha_2\cdots \alpha_n}
= \frac{\pd^{n+1}F_0}{\pd t^0\pd t^{\alpha_1} \cdots
\pd t^{\alpha_n}},
\ee
then for $g \geq 1$,
\be
F_g= F_g(u_\alpha, u_{\alpha_1\alpha_2},\dots, u_{\alpha_1\alpha_2\cdots \alpha_{3g-1}}), \qquad
g \geq  1.
\ee
Note by the constitutive relations and the integrable hierarchy \eqref{eqn:IH},
\ben
u_{\alpha_1\alpha_2\cdots \alpha_n}
& = & \frac{\pd}{\pd t^{\alpha_n}} \cdots \frac{\pd}{\pd t^{\alpha_2}}
u_{\alpha_1} \\
& = &  \frac{\pd}{\pd t^{\alpha_n}} \cdots \frac{\pd}{\pd t^{\alpha_3}}
(R_{\alpha_1,0; \alpha_2, 0}(u))'
\een
$u_{\alpha_1\alpha_2\cdots \alpha_n}$ can be expressed
in terms of $\{u_{\alpha,k}\;|\; \alpha =0, \dots, r,
k=0,\dots, n-1\}$.
So $F_g$ can be expressed in terms of $\{u_{\alpha,k}\;|\; \alpha =0, \dots, r,
k=0,\dots, 3g-2\}$.
This is called the $(3g-2)$-conjecture.
In Eguchi-Getzler-Xiong \cite{Egu-Get-Xio},
it was shown that $\{u_{\alpha,k}\;|\; \alpha =0, \dots, r,
k \geq 0\}$ can be used as new coordinates on the big phase space.
The $(3g-2)$-conjecture has been proved in Dubrovin-Zhang \cite{Dub-Zha}
under a technical condition.
When one compares with the results in the case of two-dimensional
topological gravity in \S \ref{sec:2DGravity},
it is desirable to see whether
these jet variables can be interpreted in terms of
renormalized coupling constants.

\subsection{The loop operators}

Introduce the {\em loop operators}:
\be \label{eqn:LoopOp}
D_{\beta, z}:= \sum_{n=0}^\infty z^{-n-1}
\frac{\pd}{\pd t_{n}^{\beta}}.
\ee
Then the system of the hierarchy can be written in
the following compact form:
\be \label{eqn:IH-Loop}
D_{\beta, z}u_\alpha
= (R_{\alpha,0; \beta}(u_0, \dots, u_r;z))',
\ee
where
\be \label{eqn:alpha-0-beta}
R_{\alpha,0; \beta}(u_0,\dots,u_r; z)
= \eta_{\alpha\beta}
+ \sum_{n =0}^\infty R_{\alpha,0;\beta,n}(u_0, \dots, u_r) z^{-n-1}.
\ee
The reason for adding the extra term $\eta_{\alpha\beta}$
will be clear below.

Now we define the genus zero two-point function by:
\be
\cV_{\alpha,\beta}(z_1,z_2)
= \sum_{m,n=0}^\infty
\frac{\pd^2F_0}{\pd t^\alpha_m \pd t^\beta_n}
z_1^{-m-1}z_2^{-n-1}
= D_{\alpha,z_1}D_{\beta,z_2}F_0.
\ee
By the constitutive relations,
\be
\cV_{\alpha,\beta}(z_1,z_2)
= \sum_{m,n=0}^\infty R_{\alpha,m;\beta,n}(u_0, \dots, u_r)
z_1^{-m-1}z_2^{-n-1}.
\ee

To compute $\cV_{\alpha, \beta}(z_1, z_2)$,
one can proceed as follows:
\ben
\frac{\pd}{\pd t^0}\cV_{\alpha, \beta}(z_1, z_2)
& = & \sum_{m,n=0}^\infty
\frac{\pd^2F_0}{\pd t^0\pd t^\alpha_m \pd t^\beta_n}
z_1^{-m-1}z_2^{-n-1} \\
& = & D_{\alpha,z_1} \sum_{n=0}^\infty
\frac{\pd^2F_0}{\pd t^0 \pd t^\beta_n} z_2^{-n-1} \\
& = & D_{\alpha,z_1} \sum_{n=0}^\infty
R_{0,0;\beta, n}(u_0, \dots, u_r) z_2^{-n-1},
\een
where in the last equality we have used
the constitutive  relation again.
Now we can use \eqref{eqn:IH-Loop} to get:
\ben
&& \frac{\pd}{\pd t^0}\cV_{\alpha, \beta}(z_1, z_2) \\
& = & \sum_{n=0}^\infty
\sum_\gamma
\frac{\pd R_{0,0;\beta, n}(u_0, \dots, u_r)}{\pd u_\gamma}
z_2^{-n-1} \cdot D_{\alpha,z_1}(u_\gamma)   \\
& = & \sum_{n=0}^\infty
\sum_\gamma
\frac{\pd R_{0,0;\beta, n}(u_0, \dots, u_r)}{\pd u_\gamma}
z_2^{-n-1}
\cdot (R_{\gamma,0;\alpha}(u_0, \dots, u_r; z_1))'.
\een
To understand $\frac{\pd R_{0,0;\beta, n}(u_0, \dots, u_r)}{\pd u_\gamma}
$, we go back to the small space to see that it is equal to
\ben
&& \frac{\pd}{\pd t_\gamma}
\biggl(\frac{\pd^3 F_0}{\pd t^0\pd t^\beta_n} \biggr)
=\eta^{\gamma\delta}\frac{\pd^3 F_0}{\pd t^0\pd t^\beta_n\pd t^\delta}
= \eta^{\gamma\delta} \frac{\pd}{\pd t^0}
\biggl(\frac{\pd^2 F_0}{\pd t^\beta_n\pd t^\delta} \biggr)
\een
and so buy string equation it is equal
to $\eta_{\gamma\delta} \frac{\pd^2 F_0}{\pd t^\beta_{n-1}\pd t^\delta}$
restricted to the small phase space,
so  it is equal to
$\eta^{\gamma\delta}
R_{\delta,0; \beta, n-1}(u_0,\dots, u_r)$,
for $n>0$, and for $n = 0$,
$\eta^{\gamma\delta} \eta_{\delta\beta}$,
on the big phase space.
(This explains the reason
why $\eta_{\alpha\beta}$ appears in \eqref{eqn:alpha-0-beta}.)
So we get
\ben
&& \frac{\pd}{\pd t^0}\cV_{\alpha, \beta}(z_1, z_2) \\
& = &
\sum_{\gamma,\delta}\eta^{\gamma\delta}
\biggl(\eta_{\delta\beta}z^{-1}+ \sum_{n=1}^\infty
R_{\delta,0; \beta, n-1}(u)
z_2^{-n-1}\biggr)
\cdot (R_{\gamma,0;\alpha}(u; z_1))' \\
& = & z_2^{-1} \cdot
 \frac{\pd}{\pd t_0}  R_{\gamma, 0; \alpha}(u;z_1)
\cdot \eta^{\gamma\delta} \cdot
R_{\delta, 0; \beta}(u;z_2).
\een
In the same fashion we also have:
\ben
&& \frac{\pd}{\pd t^0}\cV_{\alpha, \beta}(z_1, z_2)
= z_1^{-1} \cdot
R_{\gamma, 0; \alpha}(u;z_1)
\cdot \eta^{\gamma\delta} \cdot
\frac{\pd}{\pd t_0}  R_{\delta, 0; \beta}(u;z_2).
\een
From the above two relations one easily deduce the following equation:
\be
\cV_{\alpha, \beta}(z_1, z_2) \\
=  \frac{1}{z_1+z_2} \biggl(
R_{\gamma, 0; \alpha}(u;z_1)
\cdot \eta^{\gamma\alpha} \cdot
R_{\delta, 0; \beta}(u;z_2)
- \eta_{\alpha\beta}\biggr).
\ee
See Dubrovin \cite{Dub-Int} for a different derivation.

\subsection{Source operators and closed formula for genus zero free energy}

We now recall a result proved in Dubrovin \cite[Proposition 3.6]{Dub-Int}.
See also Eguchi-Yamada-Yang \cite[Proposition 1]{Egu-Yam-Yan}.
The free energy at $g = 0$ is given by
\be
F_0(t) = \frac{1}{2}
\sum_{m,n=0} R_{\alpha,m;\beta,n}
\tilde{t}_m^\alpha \tilde{t}^\beta_n,
\ee
where $\tilde{t^\alpha}_m = t_m^\alpha
- \delta_{m,1}\delta_{\alpha,0}$.
Therefore,
if one sets
\be
\cV_{\alpha,\beta}(z_1,z_2)
=  \sum_{m,n=0}^\infty R_{\alpha,m;\beta,n}(U)
z_1^{-m-1}z_2^{-n-1},
\ee
then
\be \label{eqn:F0-in-u}
F_0=\half \sum_{\alpha,\beta=0}^r \res_{z_1=0}
(\res_{z_2=0} ( J^\alpha(z_1)J^\beta(z_2)
\cV_{\alpha,\beta}(u;z_1,z_2) dz_2)dz_1).
\ee
If we note
\be
D_{t^\alpha, z}J^\beta(w)
= \delta_{\alpha\beta}
\sum_{n=0}^\infty \frac{w^n}{z^{n+1}}
= \frac{\delta_{\alpha\beta}}{z-w},
\ee
then the genus zero one-point function is very easy to find:
\be
\begin{split}
& D_{t^\gamma,z}F_0= \sum_{\alpha=0}^r \res_{z_1=0}
(\res_{z_2=0} ( J^\alpha(z_1)
\cV_{\alpha,\gamma}(u;z_1,z_2) \frac{dz_2}{z-z_2})dz_1) \\
+ & \half \sum_{\alpha,\beta=0}^r \res_{z_1=0}
(\res_{z_2=0} ( J^\alpha(z_1)J^\beta(z_2)
D_{t^\gamma,z}\cV_{\alpha,\beta}(u;z_1,z_2) dz_2)dz_1).
\end{split}
\ee

\subsection{Computations of $n$-point functions by operator formalism}

Now we present an algorithm to compute the $n$-point correlation
functions of GW invariants in genus $g$ when $2g-2+n > 0$.
When $g=0$, this requires that $n \geq 3$.
We have already seen that the $2$-point functions in genus zero
can be expressed in terms of the order parameters.
So for $n \geq 3$,
one can simply apply the loop operators
repeatedly on the two-point functions
and apply the constitutive relations  \eqref{eqn:Constitutive}
and the integrable hierarchy \eqref{eqn:IH-Loop}:
\be
D_{\alpha_1,z_1} \cdots D_{\alpha_n,z_n}F_0
= D_{\alpha_1,z_1} \cdots D_{\alpha_{n-2},z_{n-2}}
\cV_{\alpha_{n-1},\alpha_n}(u; z_{n-1}, z_n).
\ee
Similarly,
for $g \geq 1$,
suppose that $F_g$ is expressed in terms of the jet variables
$\{u_{\alpha,n}\}$,
then the $n$-point functions can be obtained by applying
the loop operators repeatedly on such expressions
with the help of \eqref{eqn:IH-Loop}.

\subsection{Emergent conformal field theory}

Inspired by conformal field theory,
it is natural to regard the source operators
as the generating series of the creation operators,
and the loop operators as the generating series of
annihilation operators,
and combine them into bosonic fields of operators:
\be
\varphi_\alpha(z):=\eta_{\alpha\beta}J^\beta(z)
+ D_{\beta,z}
= \eta_{\alpha\beta}\sum_{n=0}^\infty \tilde{t}^{\beta}_n z^{n}
+ \sum_{n=0}^\infty z^{-n-1}
\frac{\pd}{\pd t_n^\alpha}.
\ee
It is very easy to see that one has the following
operator product expansion:
\be
\varphi_\alpha(z)\varphi_\beta(w)
= \frac{\eta_{\alpha\beta}}{z-w}+:\varphi_\alpha(z)\varphi_\beta(w):,
\ee
where $:\varphi_\alpha(z)\varphi_\beta(w):$ is the normally
ordered product of $\varphi_\alpha(z)$ and $\varphi_\beta(w)$.
So one can regard the GW theory as associating
a vector $|X\rangle$ in the bosonic Fock space of the systems
of bosonic fields $\{\varphi_\alpha\}_{\alpha=0}^r$,
and one can consider
the $n$-point functions:
\be
\langle 0|\varphi_{\alpha_1}(z_1) \cdots \varphi_{\alpha_n}(z_n)
|X\rangle.
\ee
One can also apply the boson-fermion correspondence
to transform $|X\rangle$ into a vector
in $|X\rangle^F$ in the fermionic Fock space.

This naive construction leads to an interesting point of view,
but in order to  obtain interesting results,
we will need a mysterious Laplace transform
suggested by the construction of the Virasoro constraints
in GW theory.

\section{Mean Field Theory of the Genus Zero Gromov-Witten Invariants
of Quintic Calabi-Yau Threefold}

\label{sec:Quintic-0}

In this Section and the next Section we explicitly
carry out the  computations mentioned
in the above two Sections for the quintic CY threefold.
We will focus on the  genus zero part in this Section and treat
the higher genera computation in the next Section.

Our main result in this Section is that we can get explicit formulas
for the Landau-Ginzburg potential function and
for the genus zero partition functions
and the $n$-points functions for $n \geq 1$
in the case of the quintic.
We note that the selection rules for
the CY threefolds play a crucial role
in this case.

\subsection{Selection rules for the GW theory of the quintic}
For the quintic Calabi-Yau threefold $M \subset \bP^4$,
$H^{ev}(M)$ is four dimensional.
It is spanned by four primary operators $P=1$, $Q=j^*H$,
$R=\frac{1}{5}j^*H$, $S=\frac{1}{5}j^*H$,
where $j: M \hookrightarrow \bP^4$ is the inclusion map,
and $H$ is the hyperplane class in $H^2(\bP^4)$.

Because $\deg_\bR P =0$, $\deg_\bR Q = 2$, $\deg_\bR R =4$, $\deg_\bR S= 6$,
and the real dimension of $M$ is $6$,
and by the selection rules \eqref{eqn:Selection-Rules},
if one assigns
\begin{align} \label{eqn:SR-CY1}
\deg t^P_n & = n-1, &
\deg t^Q_n & = n, &
\deg t^R_n & = n +1, &
\deg t^S_n & = n+2,
\end{align}
then $F_g$ is  homogeneous of degree $0$:
\be \label{eqn:SR-CY2}
\cX F_g = 0,
\ee
where
\be\begin{split}
\cX = & \sum_{n=0}^\infty (n-1)t^P_n \frac{\pd}{\pd t^P_n}
+ \sum_{n=0}^\infty nt^Q_n \frac{\pd}{\pd t^Q_n} \\
+ & \sum_{n=0}^\infty (n+1)t^R_n \frac{\pd}{\pd t^R_n}
+ \sum_{n=0}^\infty (n+2)t^S_n \frac{\pd}{\pd t^S_n}.
\end{split}\ee

\subsection{The genus zero free energy on the small phase space}
\label{sec:Small}

By the selection rules \eqref{eqn:Selection-Rules},
for the quintic,
$$\corr{\cO_\alpha\cO_\beta\cO_\gamma}_{0;d =0} \neq 0$$
only if $\deg \cO_\alpha + \deg \cO_\beta + \deg \cO_\gamma = 3$,
so if $\deg \cO_\alpha \leq \deg \cO_\beta \leq \deg \cO_\gamma$,
then there are only three possibilities:
\begin{align}
\corr{PPS}_{0;0} & = 1, & \corr{PQR}_{0;0} & = 1, &
\corr{QQQ}_{0;0} & = 5.
\end{align}
So on the small phase space,
the classical part of $F_0^{small}$ is
\be
F^{small}_{0, classical} = \frac{(t^P)^2t^S}{2} + t^Pt^Qt^R
+ \frac{5}{6} (t^Q)^3.
\ee

For the quantum part of $F_0^{small}$,
we need to consider the intersection numbers on the
moduli spaces $\Mbar_{0,n}(X; d)$ for $d > 0$
or $d=0$ and $n > 3$.
The moduli spaces have  expect dimensions $n$,
and so
\be
\corr{P^{m_0} Q^{m_1}R^{m_2}S^{m_3}}_{0;d} \neq 0
\ee
only if:
\be
m_1 + 2m_2 + 3m_3 = m_0 + m_1 + m_2 + m_3,
\ee
i.e., $m_0 = m_1+2m_2$. But by the puncture equation,
such correlators vanish when $m_0> 0$. So we have only
$\corr{Q^{m_1}}_{0;d}$ to consider.
By the divisor equation,
\ben
\corr{Q^{m_1}}_{0;d}
& = &  d \cdot \corr{Q^{m_1-1}}_{0;d}
= d^{m_1} \corr{1}_{0;d}
= d^{m_1} N_{0,d},
\een
where
\be
N_{0,d}: = \int_{[\Mbar_{0,0}(X;d)]^{virt}} 1.
\ee
So the instanton correction to $F_0^{small}$ is
\ben
F_{0,instanton}^{small}
& = & \sum_{d=1}^\infty \sum_{m_1=0}^\infty
\frac{(t^Q)^{m_1}}{m_1!} q^d
\corr{Q^{m_1}}_{0;d} \\
& = & \sum_{d=1}^\infty \sum_{m_1=0}^\infty \frac{(t^Q)^{m_1}}{m_1!} q^d
N_{0,d} d^{m_1}
= \sum_{d=1}^\infty N_{0,d} e^{dt^Q}q^d.
\een
So the genus zero free energy on the small phase
in this case is
\be \label{eqn:F0Small}
\begin{split}
F_0^{small} &  = F_{0, classical}^{small}
+ F_{0, instanton}^{small} \\
& = \frac{1}{2}(t^P)^2t^S + t^Pt^Qt^R
+ \frac{5}{6} (t^Q)^3
+ \sum_{m=1}^\infty N_{0,m} e^{mt^Q}q^m.
\end{split}
\ee
For simplicity of notations,
we will set:
\be
f_0(x) = \frac{5}{6}x^3
+ \sum_{d=1}^\infty N_{0,d} e^{dx} q^d.
\ee
So we have
\be
F_0^{small} = \frac{1}{2}(t^P)^2t^S + t^Pt^Qt^R
+ f_0(t^Q).
\ee

The metric matrix in this case:
\be (\eta_{\lambda\mu})
= \begin{pmatrix}
0 & 0 & 0 & 1 \\
0 & 0 & 1 & 0 \\
0 & 1 & 0 & 0 \\
1 & 0 & 0 & 0
\end{pmatrix}.
\ee

The gradient of $F_0^{small}$ is
\ben
&& \frac{\pd F_0^{small}}{\pd t^P} = t^Pt^S+t^Qt^R, \\
&& \frac{\pd F_0^{small}}{\pd t^Q} = t^Pt^R + f'_0(t^Q), \\
&& \frac{\pd F_0^{small}}{\pd t^R} = t^Pt^Q, \\
&& \frac{\pd F_0^{small}}{\pd t^S} = \frac{1}{2}(t^P)^2.
\een
The entropy is given by:
\be
\begin{split}
G_0 & = t^P\frac{\pd F_0^{small}}{\pd t^P}
+ t^Q\frac{\pd F_0^{small}}{\pd t^Q}
+ t^R\frac{\pd F_0^{small}}{\pd t^R}
+ t^S\frac{\pd F_0^{small}}{\pd t^S} -F_0^{small} \\
& = (t^P)^2t^S + t^Pt^Qt^R
+ t^Qf_0'(t^Q) - f_0(t^Q).
\end{split}
\ee
The Hessian of $F_0^{small}$ is:
\be
Hess(F_0^{small})
= \begin{pmatrix}
t^S & t^R & t^Q & t^P \\
t^R & f''_0(t^Q) & t^P & 0 \\
t^Q & t^P & 0 & 0 \\
t^P & 0 & 0 & 0
\end{pmatrix}.
\ee
The Euler vector field is given by:
\be
E := t^P \frac{\pd}{\pd t^P} - t^R \frac{\pd}{\pd t^R}
- 2 t^S \frac{\pd}{\pd t^S}.
\ee
It is clear that
\be
EF_0^{small} = 0.
\ee

The explicit formula for $N_{0,d}$ is given by the mirror formula
discovered in \cite{CDGP} and proved in \cite{Giv, LLY}:
\be
\frac{5}{6} (t^Q)^3
+ \sum_{d=1}^\infty N_{0,d}e^{dt^Q}q^d
= \frac{5}{2} \biggl( \frac{\omega_3}{\omega_0}
- \frac{f_1}{f_0} \frac{\omega_2}{\omega_0} \biggr).
\ee
Here $\omega_0, \dots, \omega_3$ are the four solutions of hypergeometric equation
\be
L \omega =0,
\ee
where $L$ is the differential operator
\be
L = \theta^4 - \alpha \prod_{k=1}^4 (\theta + \frac{k}{5}),
\ee
where $\alpha = 5^5e^t$ and $\theta = \alpha \frac{\pd}{\pd \alpha}$.
The variable $t$ is related to $t^Q$ via the mirror formula:
\be
t^Q = \frac{\omega_1}{\omega_0}.
\ee
Furthermore,
\be \label{eqn:4Periods}
(\omega_0, \omega_1, \omega_2, \omega_3)
= \omega_0 \cdot (1, t^Q, f_0'(t^Q), t^Qf_0'(t^Q)-2f_0(t^Q)).
\ee
Even though such structure from special geometry originally appear
on only the tiny phase space with parameter $t^Q$,
we will show later that they natural appear
in the emergent geometry on the small phase space.

\subsection{Quantum cohomology ring and WDVV equation}
\label{sec:Quan-Coh}

Recall the quantum multiplication over the small phase space is defined by:
\be
\frac{\pd}{\pd t^\alpha} \circ \frac{\pd}{\pd t^\beta}
= \frac{\pd^3F_0^{small}}{\pd t^\alpha\pd t_\beta\pd t^\lambda}
\eta^{\lambda\mu} \frac{\pd}{\pd t^\mu}.
\ee
where
$\eta_{\lambda\mu} = \frac{\pd^2F_0}{\pd t^\alpha\pd t^\beta}$
and $(\eta^{\lambda\mu}) = (\eta_{\lambda\mu})^{-1}$.
In our case
\be
(\eta^{\lambda\mu}) = (\eta_{\lambda\mu})^{-1}
= \begin{pmatrix}
0 & 0 & 0 & 1 \\
0 & 0 & 1 & 0 \\
0 & 1 & 0 & 0 \\
1 & 0 & 0 & 0
\end{pmatrix},
\ee
and so we can give the explicit formula for the quantum multiplications:
\ben
\frac{\pd}{\pd t^P} \circ \begin{pmatrix}
\frac{\pd}{\pd t^P} \\
\frac{\pd}{\pd t^Q} \\
\frac{\pd}{\pd t^R} \\
\frac{\pd}{\pd t^S} \end{pmatrix}
& = & \begin{pmatrix}
0 & 0 & 0 & 1 \\
0 & 0 & 1 & 0 \\
0 & 1 & 0 & 0 \\
1 & 0 & 0 & 0
\end{pmatrix}
\cdot \begin{pmatrix}
0 & 0 & 0 & 1 \\
0 & 0 & 1 & 0 \\
0 & 1 & 0 & 0 \\
1 & 0 & 0 & 0
\end{pmatrix}
\cdot \begin{pmatrix}
\frac{\pd}{\pd t^P} \\
\frac{\pd}{\pd t^Q} \\
\frac{\pd}{\pd t^R} \\
\frac{\pd}{\pd t^S} \end{pmatrix} \\
& = & \begin{pmatrix}
1 & 0 & 0 & 0 \\
0 & 1 & 0 & 0 \\
0 & 0 & 1 & 0 \\
0 & 0 & 0 & 1
\end{pmatrix} \cdot
 \begin{pmatrix}
\frac{\pd}{\pd t^P} \\
\frac{\pd}{\pd t^Q} \\
\frac{\pd}{\pd t^R} \\
\frac{\pd}{\pd t^S} \end{pmatrix},
\een

\ben
\frac{\pd}{\pd t^Q} \circ \begin{pmatrix}
\frac{\pd}{\pd t^P} \\
\frac{\pd}{\pd t^Q} \\
\frac{\pd}{\pd t^R} \\
\frac{\pd}{\pd t^S} \end{pmatrix}
& = & \begin{pmatrix}
0 & 0 & 1 & 0 \\
0 & f_0'''(t^Q) & 0 & 0 \\
1 & 0 & 0 & 0 \\
0 & 0 & 0 & 0
\end{pmatrix}
\cdot \begin{pmatrix}
0 & 0 & 0 & 1 \\
0 & 0 & 1 & 0 \\
0 & 1 & 0 & 0 \\
1 & 0 & 0 & 0
\end{pmatrix}
\cdot \begin{pmatrix}
\frac{\pd}{\pd t^P} \\
\frac{\pd}{\pd t^Q} \\
\frac{\pd}{\pd t^R} \\
\frac{\pd}{\pd t^S} \end{pmatrix} \\
& = &
\begin{pmatrix}
0 & 1 & 0 & 0 \\
0 & 0 & f'''_0(t^Q) & 0 \\
0 & 0 & 0 & 1 \\
0 & 0 & 0 & 0
\end{pmatrix} \cdot
 \begin{pmatrix}
\frac{\pd}{\pd t^P} \\
\frac{\pd}{\pd t^Q} \\
\frac{\pd}{\pd t^R} \\
\frac{\pd}{\pd t^S} \end{pmatrix}
\een

\ben
\frac{\pd}{\pd t^R} \circ \begin{pmatrix}
\frac{\pd}{\pd t^P} \\
\frac{\pd}{\pd t^Q} \\
\frac{\pd}{\pd t^R} \\
\frac{\pd}{\pd t^S} \end{pmatrix}
& = & \begin{pmatrix}
0 & 1 & 0 & 0 \\
1 & 0 & 0 & 0 \\
0 & 0 & 0 & 0 \\
0 & 0 & 0 & 0
\end{pmatrix}
\cdot \begin{pmatrix}
0 & 0 & 0 & 1 \\
0 & 0 & 1 & 0 \\
0 & 1 & 0 & 0 \\
1 & 0 & 0 & 0
\end{pmatrix}
\cdot \begin{pmatrix}
\frac{\pd}{\pd t^P} \\
\frac{\pd}{\pd t^Q} \\
\frac{\pd}{\pd t^R} \\
\frac{\pd}{\pd t^S} \end{pmatrix} \\
& = & \begin{pmatrix}
0 & 0 & 1 & 0 \\
0 & 0 & 0 & 1 \\
0 & 0 & 0 & 0 \\
0 & 0 & 0 & 0
\end{pmatrix} \cdot
 \begin{pmatrix}
\frac{\pd}{\pd t^P} \\
\frac{\pd}{\pd t^Q} \\
\frac{\pd}{\pd t^R} \\
\frac{\pd}{\pd t^S} \end{pmatrix},
\een

\ben
\frac{\pd}{\pd t^S} \circ \begin{pmatrix}
\frac{\pd}{\pd t^P} \\
\frac{\pd}{\pd t^Q} \\
\frac{\pd}{\pd t^R} \\
\frac{\pd}{\pd t^S} \end{pmatrix}
& = & \begin{pmatrix}
1 & 0 & 0 & 0 \\
0 & 0 & 0 & 0 \\
0 & 0 & 0 & 0 \\
0 & 0 & 0 & 0
\end{pmatrix}
\cdot \begin{pmatrix}
0 & 0 & 0 & 1 \\
0 & 0 & 1 & 0 \\
0 & 1 & 0 & 0 \\
1 & 0 & 0 & 0
\end{pmatrix}
\cdot \begin{pmatrix}
\frac{\pd}{\pd t^P} \\
\frac{\pd}{\pd t^Q} \\
\frac{\pd}{\pd t^R} \\
\frac{\pd}{\pd t^S} \end{pmatrix} \\
& = & \begin{pmatrix}
0 & 0 & 0 & 1 \\
0 & 0 & 0 & 0 \\
0 & 0 & 0 & 0 \\
0 & 0 & 0 & 0
\end{pmatrix} \cdot
 \begin{pmatrix}
\frac{\pd}{\pd t^P} \\
\frac{\pd}{\pd t^Q} \\
\frac{\pd}{\pd t^R} \\
\frac{\pd}{\pd t^S} \end{pmatrix}.
\een
At a first sight,
the effect of  quantum multiplication
on the cohomology ring seems to be minor:
It only modifies the multiplication of $\frac{\pd}{\pd t^Q}$
with $\frac{\pd}{\pd t^Q}$.
Furthermore,
in this case,
the original cohomology ring is not semisimple,
after the quantum deformation over the small phase space,
it is still not semisimple.
So the reconstruction theory of Dubrovin and Zhang \cite{Dub-Zha}
based on semisimple Frobenius manifolds cannot
be applied here directly.

One can also check that the associativity
of the quantum multiplication,
i.e., the WDVV equations,
does not impose any constraints on the function $f_0$.

\subsection{The genus zero one-point functions
on the small phase space}

\label{sec:One-Point-0}

In this Subsection,
we will show how to compute the genus zero one-point correlators
of the quintic Calabi-Yau threefold from $F_0^{small}$ by the
selection rules.

We will use $\corrr{\cdots}$ to denote the correlator on
the small phase space.
For example, $\corrr{\tau_n(S)}_0$ means
$\frac{\pd F_0}{\pd t^Q_n}$ restricted to the small phase space.

\begin{prop}
One can obtain by applications of the selection rules to get
the following formulas for quintic CY threefolds
from the genus zero free energy
on the small phase space:
\be \label{eqn:1Point-P}
\begin{split}
& t^S+ \sum_{n\geq 0}^\infty z^{-n-1} \corrr{\tau_n(P)}_0 \\
= & e^{t^P/z}t^S
+ \frac{1}{z} e^{t^P/z}t^Qt^R
+ \frac{1}{z^2} e^{t^P/z} (t^Qf_0'(t^Q)-f_0(t^Q)).
\end{split}
\ee
\be \label{eqn:1Point-Q}
t^R+ \sum_{n=0}^\infty z^{-n-1}\corrr{\tau_n(Q)}_0
= e^{t^P/z} t^R
+ \frac{e^{t^P/z}}{z}f'_0(t^Q).
\ee
\be \label{eqn:1Point-R}
t^Q + \sum_{n\geq 0}^\infty z^{-n-1}\corrr{\tau_n(R)}_0
= e^{t^P/z}t^Q.
\ee
\be \label{eqn:1Point-S}
t^P+\sum_{n\geq 0}^\infty z^{-n-1}\corrr{\tau_n(S)}_0
= z\biggl(e^{t_{P}/z}-1\biggr).
\ee
\end{prop}

\begin{proof}
From the formula \eqref{eqn:F0Small},
we have
\be \label{eqn:Corr-g=0}
\begin{split}
& \corr{\tau_0(P)^{m_0} \cdots \tau_0(S)^{m_3}}_0 \\
= & \begin{cases}
1, & \text{if $(m_0,\dots, m_3) = (2,0,0,1)$}, \\
1, & \text{if $(m_0,\dots, m_3) = (1,1,1,0)$}, \\
5\delta_{m_1,3}
+ \sum_{d=1}^\infty N_dd^{m_1}q^{d}, &
\text{if $(m_0,\dots, m_3) = (0,m_1,0,0)$}, \\
0, & \text{otherwise}.
\end{cases}
\end{split}
\ee
To compute $\corrr{\tau_n(P)}_0$ we need to compute
$\corr{\tau_n(P) \tau_0(P)^{m_0} \cdots \tau_0(S)^{m_3}}_0$.
By the selection rule,
\ben
&& n + m_1 + 2m_2 + 3m_3 = 1 + m_0 + m_1 + m_2 + m_3,
\een
and so we get:
\be
m_0 = n-1+m_2+2m_3.
\ee
For $n=0$,
by \eqref{eqn:Corr-g=0} one easily sees that:
\ben
\corrr{\tau_0(P)}_0 = t^Pt^S + t^Qt^R.
\een
For $n=1$,
we use dilaton equation to get:
\ben
&& \corr{\tau_1(P) \tau_0(P)^{m_0} \cdots
\tau_0(S)^{m_3}}_0 \\
& = & (m_0+\cdots + m_3-2) \cdot
\corr{\tau_0(P)^{m_0} \cdots
\tau_0(S)^{m_3}}_0 \\
& = & \begin{cases}
1, & \text{if $(m_0,\dots, m_3) = (2,0,0,1)$}, \\
1, & \text{if $(m_0,\dots, m_3) = (1,1,1,0)$}, \\
(m_1-2)\biggl(5\delta_{m_1,3}
+ \sum_{d=1}^\infty N_dd^{m_1}q^{d}\biggr), &
\text{if $(m_0,\dots, m_3) = (0,m_1,0,0)$}, \\
0, & \text{otherwise}.
\end{cases}
\een
For $n \geq 2$,
\ben
&& \corr{\tau_n(P) \tau_0(P)^{m_0} \cdots
\tau_0(S)^{m_3}}_0 \\
& = & \corr{\tau_n(P) \tau_0(P)^{n-1+m_2+2m_3}
\tau_0(Q)^{m_1} \cdots
\tau_0(S)^{m_3}}_0 \\
& = & \corr{\tau_0(P)^{m_2+2m_3}\tau_1(P)\tau_0(Q)^{m_1}\tau_0(R)^{m_2}
\tau_0(S)^{m_3}}_0.
\een
If $m_2=m_3=0$,
then   by the dilaton equation one has
\ben
\corr{\tau_1(P)\tau_0(Q)^{m_1}}_0
= (m_1-2) \corr{\tau_0(Q)^{m_1}}_0
= \delta_{m_1,3} 5 + (m_1-2) \sum_{d \geq 1} N_d d^{m_1}q^d.
\een
If $m_2 > 0$ or $m_3> 0$,
then
\ben
&&\corr{\tau_0(P)^{m_2+2m_3}\tau_1(P)\tau_0(Q)^{m_1}\tau_0(R)^{m_2}
\tau_0(S)^{m_3}}_0 \\
&=&\corr{\tau_0(P)^{m_2+2m_3-1}\tau_0(P)\tau_0(Q)^{m_1}\tau_0(R)^{m_2}
\tau_0(S)^{m_3}}_0 \\
& = & \corr{\tau_0(P)^{m_2+2m_3}\tau_0(Q)^{m_1}\tau_0(R)^{m_2}
\tau_0(S)^{m_3}}_0,
\een
which is nonvanishing only when $(m_1,m_2,m_3) = (1,1,0)$
or $(0,0,1)$.
Now we can compute the generating series of
\ben
\corrr{\tau_n(P)}_0
& = & \sum_{m_0, m_1, m_2, m_3\geq 0} \frac{t_{0,P}^{m_0}}{m_0!}
\cdots \frac{t_{0,S}^{m_3}}{m_3!}
\cdot \corr{\tau_n(P) \tau_0(P)^{m_0} \cdots
\tau_0(S)^{m_3}}_0.
\een
as follows:
\ben
&& \sum_{n\geq 0}^\infty z^{-n-1}
\corrr{\tau_n(P)}_0  \\
& = & \frac{1}{z}(t^Pt^S+t^Qt^R)
+ \frac{1}{z^2}\biggl(\frac{(t^P)^2}{2}t^S + t^Pt^Qt^R \\
&& + \sum_{m_1 \geq 0} \frac{(t^Q)^{m_1}}{m_1!}
(m_1-2)\biggl(5\delta_{m_1,3}
+ \sum_{d=1}^\infty N_dd^{m_1}q^{d}\biggr)
\biggr) \\
& + & \sum_{n\geq 2}^\infty z^{-n-1}\sum_{m_1\geq 0}
\frac{t_{0,P}^{n-1}}{(n-1)!}
\frac{t_{0,Q}^{m_1}}{m_1!}
\cdot \corr{\tau_n(P) \tau_0(P)^n\tau_0(Q)^{m_1}}_0 \\
& + & \sum_{n\geq 2}^\infty z^{-n-1}
\corr{\tau_n(P) \tau_0(P)^{n}
\tau_0(Q)\tau_0(R)}_0 \frac{t_{0,P}^{n}}{n!}t_{0,Q}t_{0,R} \\
& + &  \sum_{n\geq 2}^\infty z^{-n-1}
\corr{\tau_n(P) \tau_0(P)^{n+1}
\tau_0(S)}_0 \frac{t_{0,P}^{n+1}}{(n+1)!} t_{0,S} \\
& = &  \sum_{n\geq 1}^\infty z^{-n-1}\sum_{m_1\geq 0}
\frac{(t^P)^{n-1}}{(n-1)!}\frac{(t^Q)^{m_1}}{m_1!}
\cdot (\delta_{m_1,3} 5 + (m_1-2) \sum_{d \geq 1} N_d d^{m_1}q^d) \\
& + & \frac{1}{z} e^{t^P/z}t^Qt^R + (e^{t^P/z}-1)t^S \\
& = & \frac{1}{z^2} \frac{5(t^Q)^3}{6}e^{t^P/z}
+ \frac{1}{z^2} e^{t^P/z} \sum_{m_1\geq 0}
\frac{(t^Q)^{m_1}}{m_1!} (m_1-2) \sum_{d \geq 1} N_d d^{m_1} q^d \\
& + & \frac{1}{z} e^{t^P/z}t^Qt^R + (e^{t^P/z}-1)t^S \\
& = & (e^{t^P/z}-1)t^S
+ \frac{1}{z} e^{t^P/z}t^Qt^R
+ \frac{1}{z^2} e^{t^P/z} \biggl(\frac{5(t^Q)^3}{6} \\
& + & \sum_{d \geq 1} N_d (dt^Q-2)e^{dt^Q}q^d \biggr).
\een
The other three formulas can be proved in the same way.
\end{proof}

\begin{rmk}
Here we recover some formulas in Example 5.3 in Dubrovin \cite{Dub-Int}
by selection rules.
We do this without solving any differential equations,
only using the selection rules.
\end{rmk}

\subsection{Explicit formulas for Landau-Ginzburg potential
and Landau-Ginzburg equations}

By changing $t^P$ to $u_S$, $t^Q$ to $u_R$, $t^R$ to $u_Q$
and $t^S$ to $u_P$ in \eqref{eqn:1Point-P}-\ref{eqn:1Point-S},
we then get:
\be
\begin{split}
& \sum_{n=-1}^\infty z^{-n-1}
R_{P,n}(u_P,u_Q,u_R, u_S)   \\
= & e^{u_S/z}u_P
+ \frac{1}{z} e^{u_S/z}u_Qu_R
+ \frac{1}{z^2} e^{u_S/z} (u_Rf_0'(u_R)-2f_0(u_R)).
\end{split}
\ee
\be
\sum_{n=-1}^\infty z^{-n-1} R_{Q,n}(u_P, \dots, u_S)
= e^{u_S/z} u_Q
+ \frac{e^{u_S/z}}{z} f_0'(u_R).
\ee

\be
\sum_{n= -1}^\infty z^{-n-1}R_{R,n}(u_P, \dots, u_S)
= e^{u_S/z}u_R.
\ee
\be
\sum_{n=-1}^\infty z^{-n-1}R_{S,n}(u_0, \dots, u_S)
= z\biggl(e^{U_{S}/z}-1\biggr).
\ee
With these formulas,
we can write down the Landau-Ginzburg equation
for the GW theory of quintic CY threefold
explicitly  as follows:

\begin{prop}
The mean field theory of the GW theory of quintic CY threefold
has the following Landau-Ginzburg potential:
\be
\begin{split}
W = & -U_PU_S-U_QU_R \\
+ & \sum_{n=0}^\infty t^P_n \biggl(U_P\frac{U_S^n}{n!}
+   \frac{U_QU_RU_S^{n-1}}{(n-1)!} \\
+ &  (U_Rf'_0(U_R)-2f_0(U_R)) \frac{U_S^{n-2}}{(n-2)!}
\biggr) \\
+ & \sum_{n=0}^\infty t^Q_n \biggl(U_Q\frac{U_S^n}{n!}
+   \frac{U_S^{n-1}}{(n-1)!}f'_0(U_R) \biggr) \\
+ & \sum_{n=0}^\infty t^R_n U_R\frac{U_S^n}{n!}
+ \sum_{n=0}^\infty t^S_n \frac{U_S^{n+1}}{(n+1)!}.
\end{split}
\ee
\end{prop}

We also write $W$ in the following form:
\be
\begin{split}
W & = U_P
\biggl(-U_S + \sum_{n=0}^\infty t^P_n \frac{U_S^n}{n!} \biggr) \\
& + U_Q \biggl(-U_R
+ \sum_{n=0}^\infty t^Q_n \frac{U_S^n}{n!}
+ U_R \sum_{n=1}^\infty t^P_n
\frac{U_S^{n-1}}{(n-1)!}  \biggr) \\
& + U_R
\sum_{n=0}^\infty t^R_n \frac{U_S^n}{n!}
+ f_0'(U_R) \sum_{n=1}^\infty t^Q_n
\frac{U_S^{n-1}}{(n-1)!}  \\
& + g_0(U_R) \sum_{n=2}^\infty t^P_n
\frac{U_S^{n-2}}{(n-2)!}
+  \sum_{n=0}^\infty t^S_n \frac{U_S^{n+1}}{(n+1)!},
\end{split}
\ee
where $g_0(U_R) = U_Rf_0'(U_R)-2f_0(U_R)$.
By computing the gradient of $W$,
the Landau-Ginzburg equations can be written downs
explicitly as follows:
\be \label{eqn:u-S}
u_S = \sum_{n=0}^\infty t^P_n \frac{u_S^n}{n!},
\ee
\be \label{eqn:u-R}
u_R
= \sum_{n=0}^\infty t^Q_n \frac{u_S^n}{n!}
+ u_R \sum_{n=1}^\infty t^P_n
\frac{u_S^{n-1}}{(n-1)!}.
\ee
\be \label{eqn:u-Q}
\begin{split}
u_Q
= &\sum_{n=0}^\infty t^R_n \frac{u_S^n}{n!}
+ u_Q \sum_{n=1}^\infty t^P_n
\frac{u_S^{n-1}}{(n-1)!} \\
+ & f_0''(u_R) \sum_{n=1}^\infty t^Q_n
\frac{u_S^{n-1}}{(n-1)!}
+ g_0'(u_R) \sum_{n=2}^\infty t^P_n
\frac{u_S^{n-2}}{(n-2)!} .
\end{split}
\ee
\be \label{eqn:u-P}
\begin{split}
u_P & = u_P\sum_{n=1}^\infty t^P_n \frac{u_S^{n-1}}{(n-1)!}  \\
& + u_Q \biggl(
\sum_{n=1}^\infty t^Q_n \frac{u_S^{n-1}}{(n-1)!}
+ u_R \sum_{n=2}^\infty t^P_n
\frac{u_S^{n-2}}{(n-2)!}  \biggr) \\
& + u_R
\sum_{n=1}^\infty t^R_n \frac{u_S^{n-1}}{(n-1)!}
+ f_0'(u_R) \sum_{n=2}^\infty t^Q_n
\frac{u_S^{n-2}}{(n-2)!} \\
& + g_0(u_R) \sum_{n=3}^\infty t^P_n
\frac{u_S^{n-3}}{(n-3)!}
+  \sum_{n=0}^\infty t^S_n \frac{u_S^{n}}{n!}.
\end{split}
\ee

\subsection{Solutions of the Landau-Ginzburg equations}
These equation can be easily solved.
As we have mentioned in \S \ref{sec:2DGravity},
equation \eqref{eqn:u-S} has been explicitly solved in \cite{Zhou-1D}:
\be \label{eqn:u-S-Formula}
u_S = \sum_{k=1}^\infty \frac{1}{k}
\sum_{p_1 + \cdots + p_k = k-1} \frac{t^P_{p_1}}{p_1!} \cdots
\frac{t^P_{p_k}}{p_k!}.
\ee
We now give another formula which express it
as a summation over partitions of nonnegative integers.
A partition of $N \geq 0$
can be written as
$1^{m_1} \cdots N^{m_N}$,
where $m_1, \dots, m_N$ are nonnegative integers such that:
\be
m_11+\cdots +m_N N = N.
\ee
Then we have
\be \label{eqn:u-S-Formula2}
u_S = \sum_{N=0}^\infty
\sum_{\sum_i im_i = N}
\frac{N!(t^P)^{N-\sum_i m_i+1} }{m_1!\cdots m_N!
\cdot (N-\sum_i m_i+1)!}
 \prod_{i=1}^N
\biggl(\frac{t^P_i}{i!}\biggr)^{m_i}.
\ee
The following are the first few terms:
\ben
u_S & = & t^P+t^Pt^P_1
+ \biggl((t^P)^2\frac{t^P_2}{2!} + t^P(t^P_1)^2 \biggr) \\
& + & \biggl( (t^P)^3\frac{t^P_3}{3!}
+3(t^P)^2 \frac{t^P_2}{2!}t^P_1
+ t^P(t^P_1)^3  \biggr) \\
& + & \biggl( (t^P)^4\frac{t^P_4}{4!}+4(t^P)^3 \frac{t^P_3}{3!}t^P_1
+ 2(t^P)^3(\frac{t^P_2}{2!})^2
+ 6(t^P)^2 \frac{t^P_2}{2!}(t^P_1)^2+ t^P(t^P_1)^4\biggr) + \cdots.
\een
The coefficients $(1),(1,1),(1,3,1)$, $(1,4,2,6,1)$
are integer  sequence A134264
on \cite{OEIS}.

Taking $\frac{\pd}{\pd t^P}$ on both sides of \eqref{eqn:u-S}
one can get:
\be \label{eqn:u-S1}
u_{S,1} = \biggl(1- \sum_{n=1}^\infty t^P_n
\frac{u_S^{n-1}}{(n-1)!}\biggr)^{-1}.
\ee
From \eqref{eqn:u-R} we get:
\be \label{eqn:u-R-sol}
u_R
= \biggl(1- \sum_{n=1}^\infty t^P_n
\frac{u_S^{n-1}}{(n-1)!}\biggr)^{-1}
\cdot \sum_{n=0}^\infty t^Q_n \frac{u_S^n}{n!}
= u_{S,1} \sum_{n=0}^\infty t^Q_n \frac{u_S^n}{n!}.
\ee
From this we get:
\be
\sum_{n=0}^\infty t^Q_n \frac{u_S^n}{n!}
= \frac{u_R}{u_{S,1}}.
\ee

\begin{prop}
The following formula for $u_R$ holds:
\be \label{eqn:u-R-Formula}
u_R = \sum_{k=1}^\infty
\sum_{p_1 + \cdots + p_k = k-1} \frac{t^Q_{p_1}}{p_1!}
\frac{t^P_{p_2}}{p_2!} \cdots
\frac{t^P_{p_k}}{p_k!}.
\ee
\end{prop}

\begin{proof}
Consider an operator
\be
D:=\sum_{n=0}^\infty t^Q_n \frac{\pd}{\pd t^P_n}.
\ee

Claim. The action of this operator on $u_S$ changes it to $u_R$:
\be
Du_S = u_R.
\ee
This can be proved as follows.
We will actually use some results from \S \ref{sec:IH-Genus-zero} below.
From
\ben
D_{P,z}(u_S)=\sum_{n\geq 0} z^{-n-1} \frac{\pd u_P}{\pd t_{n, P}}
= \biggl[ e^{u_S/z}\biggr]',
\een
we get
\be
\frac{\pd u_S}{\pd t^P_n}
= \frac{u_S^n}{n!} u_{S,1},
\ee
Therefore, we have
\ben
Du_S & = &
\sum_{n=0}^\infty t^Q_n \frac{\pd u_S}{\pd t^P_n}
= u_{S,1} \sum_{n=0}^\infty t^Q_n \frac{u_S^n}{n!}
= u_R.
\een
In other words,
we have shown that
\be
u_R = \sum_{n=0}^\infty t^Q_n
\frac{\pd u_S}{\pd t^P_n}.
\ee
Now \eqref{eqn:u-R-Formula} follows from \eqref{eqn:u-S-Formula}.
\end{proof}

From the proof we also get the following formula:
\be \label{eqn:u-S-tP-n}
u_{S, 1} \frac{u_S^n}{n!}
= \frac{\pd u_S}{\pd t^P_n}
= \frac{1}{n!}\sum_{k=n}^\infty
\sum_{p_1 + \cdots + p_k = k-n} \frac{t^P_{p_1}}{p_1!} \cdots
\frac{t^P_{p_k}}{p_k!}.
\ee
It can also be written as:
\ben
\frac{\pd u_S}{\pd t^P_n}
 = \frac{1}{n!}\sum_{N=n}^\infty
\sum_{\sum_i im_i = N-n}
\frac{N!(t^P)^{N-\sum_i m_i} }{m_1!\cdots m_N!
\cdot (N-\sum_i m_i)!}
 \prod_{i=1}^N
\biggl(\frac{t^P_i}{i!}\biggr)^{m_i}.
\een
For example,
\ben
&& u_{S,1} = \frac{\pd u_S}{\pd t^P} \\
& = & 1+ t^P_1
+ \biggl(2t^P\frac{t^P_2}{2!} + (t^P_1)^2 \biggr) \\
& + & \biggl( 3(t^P)^2\frac{t^P_3}{3!}
+6 t^P \frac{t^P_2}{2!}t^P_1 + (t^P_1)^3  \biggr) \\
& + & \biggl( 4(t^P)^3\frac{t^P_4}{4!}+12(t^P)^2 \frac{t^P_3}{3!}t^P_1
+ 6(t^P)^2(\frac{t^P_2}{2!})^2
+ 12t^P \frac{t^P_2}{2!}(t^P_1)^2+ (t^P_1)^4\biggr) + \cdots,
\een
where the coefficients $1$, $(2,1)$, $(3,6,1)$, ($4,1,2,6,12,1)$, $\dots$ are
the sequence A035206 on \cite{OEIS}.
\ben
\frac{\pd u_S}{\pd t^P_1} & = &  t^P + 2t^P t^P_1
+ \biggl(  3(t^P)^2 \frac{t^P_2}{2!}
+ 3t^P(t^P_1)^2  \biggr) \\
& + & \biggl( 4(t^P)^3 \frac{t^P_3}{3!}
+ 12(t^P)^2 \frac{t^P_2}{2!}t^P_1+ 4 t^P(t^P_1)^3\biggr) + \cdots.
\een
where the coefficients $1$, $2$, $(3,3)$, $(4,12,4)$, $\dots$ are not on
\cite{OEIS},
instead,
after rescaling we get
$1$, $1$, $(1,1)$, $(1,3,1)$, $\dots$ which are
the sequence A134264 on \cite{OEIS} again.
We also have:
\ben
\frac{\pd u_S}{\pd t^P_2} & = &
\frac{1}{2!} \biggl[ (t^P)^2  +3(t^P)^2  t^P_1
+ \biggl( 4(t^P)^3 \frac{t^P_2}{2!}
+ 6(t^P)^2 (t^P_1)^2 \biggr) + \cdots \biggr], \\
\frac{\pd u_S}{\pd t^P_3}  & = & \frac{1}{3!}  \biggl( (t^P)^3
+4(t^P)^3 t^P_1 + \cdots \biggr).
\een
From these we can write the first few terms of $u_R$:
\ben
u_R & = & t^Q \frac{\pd u_S}{\pd t^P}
+ t^Q_1 \frac{\pd u_S}{\pd t^P_1} + t^Q_2 \frac{\pd U_S}{\pd t^P_2}
+ t^Q_3 \frac{\pd U_S}{\pd t^P_3} + \cdots \\
& = & t^Q + (t^Qt^P_1+t^Q_1t^P) \\
& + & \biggl(t^Q \biggl(2t^P\frac{t^P_2}{2!} + (t^P_1)^2 \biggr)
+  2t^Q_1\cdot t^Pt^P_1 + \frac{t^Q}{2!}\cdot (T^P)^2\biggr) + \cdots.
\een

From \eqref{eqn:u-Q} and \eqref{eqn:u-S1} we get:
\be \label{eqn:u-Q-Sol}
\begin{split}
u_Q
= u_{S,1} \cdot
\biggl( &\sum_{n=0}^\infty t^R_n \frac{u_S^n}{n!}
+ f_0''(u_R) \sum_{n=1}^\infty t^Q_n
\frac{u_S^{n-1}}{(n-1)!} \\
+ & g_0'(u_R) \sum_{n=2}^\infty t^P_n
\frac{u_S^{n-2}}{(n-2)!} \biggr).
\end{split}
\ee
Similarly, from \eqref{eqn:u-P} and \eqref{eqn:u-S1} we get:
\be
\begin{split}
u_P = u_{S,1} \cdot
& \biggl[  u_Q \biggl(
\sum_{n=1}^\infty t^Q_n \frac{u_S^{n-1}}{(n-1)!}
+ u_R \sum_{n=2}^\infty t^P_n
\frac{u_S^{n-2}}{(n-2)!}  \biggr) \\
& + u_R
\sum_{n=1}^\infty t^R_n \frac{u_S^{n-1}}{(n-1)!}
+ f_0'(u_R) \sum_{n=2}^\infty t^Q_n
\frac{u_S^{n-2}}{(n-2)!} \\
& + g_0(u_R) \sum_{n=3}^\infty t^P_n
\frac{u_S^{n-3}}{(n-3)!}
+  \sum_{n=0}^\infty t^S_n \frac{u_S^{n}}{n!}
\biggr].
\end{split}
\ee
So all the order parameters have been determined.

\subsection{The jet variables as new coordinates}

Note $u_R$ is linear in $t^Q_n$'s,
$u_Q$ is linear in $t^R_n$'s,
and $u_P$ is linear in $t^S_n$'s.
For example,
\ben
u_R & = & t^Q \biggl(1+t^P_1
+ (t^Pt^P_2+(t^P_1)^2)
+  (3t^P_2t^Pt^P_1+\frac{1}{2}t^P_3(t^P)^2
+(t^P_1)^3) + \cdots \biggr) \\
& + & t^Q_1 \biggl(t^P + 2t^Pt^P_1
+ 3t^P(t^P_1)^2+\frac{3}{2}(t^P)^2t^P_2 + \cdots\biggr) \\
& + & \frac{1}{2} t^Q_2 \biggl( (t^P)^2
+ 3 (t^P)^2t^P_1 +\cdots\biggr)
+\frac{1}{6} t^Q_3 \biggl((t^P)^3 + \cdots
\biggr)+\cdots,
\een
after repeatedly taking $\frac{\pd}{\pd t^P}$,
we have:
\ben
u_{R,1} & = & t^Q \biggl(
t^P_2 + (3t^P_2t^P_1+t^P_3t^P) + \cdots \biggr) \\
& + & t^Q_1 \biggl(1 + 2t^P_1
+ 3(t^P_1)^2+ 3t^Pt^P_2 + \cdots\biggr) \\
& + & t^Q_2 \biggl( t^P
+ 3 t^Pt^P_1 +\cdots\biggr)
+\frac{1}{2} t^Q_3 \biggl((t^P)^2 + \cdots
\biggr)+\cdots,
\een
\ben
u_{R,2} & = & t^Q \biggl(
t^P_3 + \cdots \biggr) +
t^Q_1 \biggl( 3 t^P_2 + \cdots\biggr) \\
& + & t^Q_2 \biggl( 1 + 3 t^P_1 +\cdots\biggr)
+ t^Q_3 \biggl(t^P + \cdots
\biggr)+\cdots,
\een
etc.
One can easily see that $u_{S,n}$, $u_{R,s}$, $u_{Q,n}$
and $u_{P,n}$ are equal to $t^P_n$, $t^Q_n$, $t^R_n$
and $t^S_n$ respectively when restricted to
the small phase space,
and so they can be used as new coordinates.

\subsection{Renormalized coupling constants}

Now we generalize an idea in \cite{Zhou-1D}.
Write $U_S=\tilde{U}_S+u_s$, etc,
and plug them into the  Landau-Ginzburg potential.
Let us focus on the following part of $W$
and treat the rest of $W$ as perturbation:
\be
\begin{split}
W & = U_P
\biggl(-U_S + \sum_{n=0}^\infty t^P_n \frac{U_S^n}{n!} \biggr) \\
& + U_Q \biggl(-U_R
+ \sum_{n=0}^\infty t^Q_n \frac{U_S^n}{n!}
+ U_R \sum_{n=1}^\infty t^P_n
\frac{U_S^{n-1}}{(n-1)!}  \biggr) \\
& + U_R
\sum_{n=0}^\infty t^R_n \frac{U_S^n}{n!}
+  \sum_{n=0}^\infty t^S_n \frac{U_S^{n+1}}{(n+1)!}+\cdots.
\end{split}
\ee
Then we get:
\ben
W & = & (\hat{U}_P+u_P)
\biggl(-(\hat{U}_S+u_S)
+ \sum_{n=0}^\infty t^P_n \frac{(\hat{U}_S+u_S)^n}{n!} \biggr) \\
& + & (\hat{U}_Q+u_Q)\biggl(-(\hat{U}_R+u_R)
+ \sum_{n=0}^\infty t^Q_n \frac{(\hat{U}_S+u_S)^n}{n!} \\
& + & (\hat{U}_R+u_R) \sum_{n=1}^\infty t^P_n
\frac{(\hat{U}_S+u_S)^{n-1}}{(n-1)!}  \biggr) \\
& + & (\hat{U}_R+u_R)
\sum_{n=0}^\infty t^R_n \frac{(\hat{U}_S+u_S)^n}{n!}
+ \sum_{n=0}^\infty t^S_n \frac{(\hat{U}_S+u_S)^{n+1}}{(n+1)!}
+ \cdots.
\een
so we get:
\be
\begin{split}
W & = W_0 + \hat{U}_P
\biggl(-\hat{U}_S + \sum_{n=1}^\infty \hat{t}^P_n \frac{U_S^n}{n!} \biggr) \\
& + U_Q \biggl(-U_R
+ \sum_{n=0}^\infty t^Q_n \frac{U_S^n}{n!}
+ U_R \sum_{n=1}^\infty t^P_n
\frac{U_S^{n-1}}{(n-1)!}  \biggr) \\
& + U_R
\sum_{n=0}^\infty t^R_n \frac{U_S^n}{n!}
 +  \sum_{n=0}^\infty t^S_n \frac{U_S^{n+1}}{(n+1)!}
+\cdots,
\end{split}
\ee
where $W_0$ is independent of $\hat{U}_\alpha$,
and $\hat{t}^\alpha$ are deformations of
$t^\alpha_n$,
for example,
\ben
&& \hat{t}^P_n =
\sum_{k=0}^\infty t^P_{n+k} \frac{u_S^k}{k!}, \\
&& \hat{t}^Q_n = \sum_{l=0}^\infty
(t^Q_{n+l} + u_Rt^P_{n+l+1}) \frac{u_S^l}{l!}, \\
&& \hat{t}^R_n = \sum_{l=0}^\infty
(t^R_{n+l}+u_Qt^P_{n+l+1})
\frac{u_S^l}{l!}, \\
&& \hat{t}^S_n=\sum_{l=0}^\infty (t^S_{n+l}
+ u_Rt_{n+l+1}^R +u_Q(t^Q_{n+l+1} + u_Rt^P_{n+l+2})
+ u_Pt^P_{n+l+1}) \frac{u_S^l}{l!}.
\een
We refer to $\hat{t}^\alpha_n$ as the renormalized coupling constants.
They can also be used as new coordinate system on the
big phase space.

\subsection{The constitutive relations for two-point
functions in genus zero. I}
\label{sec:Two-Point-I}

Taking $\frac{\pd}{\pd t^P}, \dots, \frac{\pd}{\pd t^S}$ on both sides
of \eqref{eqn:1Point-P}-\eqref{eqn:1Point-S},
we get:
\ben
S_{P, P}(z)
& = & e^{t^P/z}  \biggl(\frac{t^S}{z}
+ \frac{t^Qt^R}{z^2}
+ \frac{1}{z^3} (t^Qf'_0(t^Q)-2f_0(t^Q)) \biggr), \\
S_{Q,P}(z)
 & = &  \frac{1}{z} e^{t^P/z}t^R
+ \frac{1}{z^2} e^{t^P/z}(t^Qf''_0(t^Q)-f_0'(t^Q)), \\
S_{R,P}(z)  & = & \frac{1}{z} e^{t^P/z}t^Q, \\
S_{S,P}(z) & = & e^{t^P/z}.
\een

\ben
S_{P,Q}(z)
& = & \frac{e^{t^P/z}}{z^2}f_0'(t^Q)
+ \frac{1}{z}e^{t^P/z}t^R, \\
S_{Q,Q}(z)
& = & \frac{e^{t^P/z}}{z}f_0''(t^Q), \\
S_{R,Q}(z)
& = &   e^{t^P/z}, \\
S_{S,Q}(z)
& = & 0.
\een

\begin{align*}
S_{P,R}(z) & = \frac{1}{z} e^{t^P/z}t^Q, &
S_{Q,R}(z) & = e^{t^P/z}, &
S_{R,R}(z) & = 0, &
S_{S,R}(z) & = 0.
\end{align*}

\begin{align*}
S_{P,S}(z) & = e^{t^P/z}, &
S_{Q,S}(z) & = 0, &
S_{R,S}(z) & = 0, &
S_{S,S}(z) & = 0.
\end{align*}
Here we have used the following notation:
\be
S_{\alpha,\beta}
= \eta_{\alpha\beta}
+ \sum_{n\geq 0}^\infty z^{-n-1}
\corrr{\tau_0(\alpha)\tau_n(\beta)}_0.
\ee

The first couple of terms of
the expansion of the matrix $(S_{\alpha,\beta})_{\alpha, \beta=P,Q,R,S}$
are
\ben
&& (S_{\alpha\beta}(z)) = \begin{pmatrix}
0 & 0 & 0 & 1 \\
0 & 0 & 1 & 0 \\
0 & 1 & 0 & 0 \\
1 & 0 & 0 & 0
\end{pmatrix}
+ \frac{1}{z}
\begin{pmatrix}
t^S & t^R & t^Q & t^P \\
t^R & f''_0(t^Q) & t^P & 0 \\
t^Q & t^P & 0 & 0 \\
t^P & 0 & 0 & 0
\end{pmatrix} \\
& + & \frac{1}{z^2} \begin{pmatrix}
t^St^P + t^Qt^R & t^Rt^P + t^Qf_0'(t^Q)-2f(t^Q)
& t^Qt^P & \frac{(t^P)^2}{2!} \\
t^Rt^P + f_0'(t^Q) & t^Pf''_0(t^Q) & \frac{(t^P)^2}{2!} & 0 \\
t^Qt^P & \frac{(t^P)^2}{2!} & 0 & 0 & \\
\frac{(t^P)^2}{2!} & 0 & 0 & 0
\end{pmatrix}
+ \cdots
\een
The leading term and the subleading term are
the metric matrix and the Hessian matrix
in \S \ref{sec:Small}.

Next we will check some other properties of the matrix
$(S_{\alpha,\beta})_{\alpha, \beta=P,Q,R,S}$.

\subsection{Quantum differential equations}
\label{sec:QDE}

It is straightforward to check that
the following equations are satisfied by $S_{\alpha,\beta}$
on the small phase space:
\be
z\frac{\pd}{\pd t^P}
\begin{pmatrix}
S_{P, \gamma} \\ S_{Q, \gamma} \\ S_{R, \gamma} \\ S_{S,\gamma}
\end{pmatrix}
= \begin{pmatrix}
1 & 0 & 0 & 0 \\
0 & 1 & 0 & 0 \\
0 & 0 & 1 & 0 \\
0 & 0 & 0 & 1
\end{pmatrix}
\cdot
\begin{pmatrix}
S_{P, \gamma} \\ S_{Q, \gamma} \\ S_{R, \gamma} \\ S_{S,\gamma}
\end{pmatrix},
\ee

\be
z\frac{\pd}{\pd t^Q}
\begin{pmatrix}
S_{P, \gamma} \\ S_{Q, \gamma} \\ S_{R, \gamma} \\ S_{S,\gamma}
\end{pmatrix}
= \begin{pmatrix}
0 & 1 & 0 & 0 \\
0 & 0 & f_0'''(t^Q) & 0 \\
0 & 0 & 0 & 1 \\
0 & 0 & 0 & 0
\end{pmatrix}
\cdot
\begin{pmatrix}
S_{P, \gamma} \\ S_{Q, \gamma} \\ S_{R, \gamma} \\ S_{S,\gamma}
\end{pmatrix},
\ee

\be
z\frac{\pd}{\pd t^R}
\begin{pmatrix}
S_{P, \gamma} \\ S_{Q, \gamma} \\ S_{R, \gamma} \\ S_{S,\gamma}
\end{pmatrix}
= \begin{pmatrix}
0 & 0 & 1 & 0 \\
0 & 0 & 0 & 1 \\
0 & 0 & 0 & 0 \\
0 & 0 & 0 & 0
\end{pmatrix}
\cdot
\begin{pmatrix}
S_{P, \gamma} \\ S_{Q, \gamma} \\ S_{R, \gamma} \\ S_{S,\gamma}
\end{pmatrix},
\ee

\be
z\frac{\pd}{\pd t^S}
\begin{pmatrix}
S_{P, \gamma} \\ S_{Q, \gamma} \\ S_{R, \gamma} \\ S_{S,\gamma}
\end{pmatrix}
= \begin{pmatrix}
0 & 0 & 0 & 1 \\
0 & 0 & 0 & 0 \\
0 & 0 & 0 & 0 \\
0 & 0 & 0 & 0
\end{pmatrix}
\cdot
\begin{pmatrix}
S_{P, \gamma} \\ S_{Q, \gamma} \\ S_{R, \gamma} \\ S_{S,\gamma}
\end{pmatrix},
\ee
where the matrices on the right-hand sides
of these equations are computed in \S
\ref{sec:Quan-Coh}.

\subsection{The genus zero two-point function}
\label{sec:Two-Point}

One can also check the the following orthogonality relations:
\be
S_{\delta, \alpha}(z) \eta^{\delta\epsilon}
S_{\epsilon,\beta}(-z)
= \eta_{\alpha, \beta}.
\ee
Let
\be
V_{\alpha,\beta}(z_1, z_2)
= \frac{1}{z_1+z_2} \biggl(
\sum_{\delta,\epsilon} \eta^{\delta\epsilon}
S_{\delta,\alpha}(z_1) S_{\epsilon,\beta}(z_2)
- \eta_{\alpha\beta}\biggr).
\ee
It is clear that
\be
V_{\alpha,\beta}(z_1,z_2) = V_{\beta,\alpha}(z_2,z_1).
\ee
One easily get:
\ben
V_{P,P}(z_1,z_2)
& = & \frac{e^{t^P/z_1+t^P/z_2}}{z_1z_2}
\biggl(t^S
+ (\frac{1}{z_1}+\frac{1}{z_2})t^Qt^R \\
& + & (\frac{1}{z_1^2}-\frac{1}{z_1z_2}
+ \frac{1}{z_2^2}) (t^Qf'(t^Q)-2f(t^Q)) \\
& + & \frac{1}{z_1z_2}
t^Q(t^Qf''_0(t^Q)-f_0'(t^Q)) \biggr), \\
\een
\ben
V_{P,Q}(z_1,z_2)
& = & \frac{e^{t^P/z_1+t^P/z_2}}{z_1z_2}
\biggr( t^R
+ \frac{1}{z_1}t^Qf_0''(t^Q)+ (\frac{1}{z_2}-\frac{1}{z_1})
f_0'(t^Q) \biggr),
\een

\ben
V_{P,R}(z_1,z_2)
=  \frac{ e^{t^P/z_1+t^P/z_2} }{z_1z_2}t^Q,
\een
\ben
V_{P,S}(z_1,z_2) = \frac{e^{t^P/z_1+t^P/z_2}-1}{z_1+z_2}.
\een
\ben
V_{Q,Q}(z_1,z_2) =
\frac{ e^{t^P/z_1+t^P/z_2}}{z_1z_2} f_0''(t_Q).
\een
\ben
V_{Q,R}(z_1,z_2)= \frac{e^{t^P/z_1+t^P/z_2}-1}{z_1+z_2}.
\een
See also Dubrovin \cite[Example 5.2]{Dub-Int}.
By changing $t^P$ to $u_S$, $t^Q$ to $u_R$,
$t^R$ to $u_Q$ and $t^S$ to $u_P$,
one then get the following formula for two-point function
of GW invariants of quintic CY threefold in genus zero
on the big phase space:
\ben
\cV_{P,P}(z_1,z_2)
& = & \frac{e^{u_S/z_1+u_S/z_2}}{z_1z_2}
\biggl(u_P
+ (\frac{1}{z_1}+\frac{1}{z_2})u_Qu_R \\
& + & (\frac{1}{z_1^2}-\frac{1}{z_1z_2}
+ \frac{1}{z_2^2}) (u_Rf'(u_R)-2f(u_R)) \\
& + & \frac{1}{z_1z_2} u_R
(u_R f''_0(u_R)-f_0'(u_R)) \biggr), \\
\een
\ben
\cV_{P,Q}(z_1,z_2)
& = & \frac{e^{u_S/z_1+u_S/z_2}}{z_1z_2}
\biggr( u_Q
+ \frac{1}{z_1}u_R f_0''(u_R)
+ (\frac{1}{z_2}- \frac{1}{z_1})
f_0'(u_R) \biggr),
\een

\ben
\cV_{P,R}(z_1,z_2)
=  \frac{ e^{u_S/z_1+u_S/z_2}}{z_1z_2} u_R,
\een
\ben
\cV_{P,S}(z_1,z_2)
= \frac{e^{u_R/z_1+u_R/z_2}-1}{z_1+z_2}.
\een
\ben
\cV_{Q,Q}(z_1,z_2) =
\frac{ e^{u_S/z_1+u_S/z_2} }{z_1z_2} f_0''(u_R).
\een
\ben
\cV_{Q,R}(z_1,z_2)
= \frac{e^{u_S/z_1+u_S/z_2}-1}{z_1+z_2}.
\een

\subsection{The integrable hierarchy in genus zero}
\label{sec:IH-Genus-zero}

Recall the system of the hierarchy
in genus zero can be written in
the following compact form:
\be
D_{\beta, z}u_\alpha
= (R_{\alpha,0; \beta}(u_0, \dots, u_r;z))',
\ee
where
\be
R_{\alpha,0; \beta}(u_0,\dots,u_r; z)
= \eta_{\alpha\beta}
+ \sum_{n =0}^\infty R_{\alpha,0;\beta,n}(u_0, \dots, u_r) z^{-n-1}.
\ee
and  $D_{\beta,z}$ is the {\em loop operators}:
\be
D_{\beta, z}:= \sum_{n=0}^\infty z^{-n-1}
\frac{\pd}{\pd t^\beta_{n}}.
\ee

Therefore,
using the explicit formula is \S \ref{sec:Two-Point-I},
we get:
\ben
R_{P,0; P}(z)
& = & e^{u_R/z}  \biggl(\frac{u_P}{z}
+ \frac{u_Qu_RR}{z^2}
+ \frac{1}{z^3} (u_R f'_0(u_R)-2f_0(u_R)) \biggr), \\
R_{Q,0;P}(z)
 & = &  \frac{1}{z} e^{u_S/z}u_Q
+ \frac{1}{z^2} e^{u_S/z}(u_R f''_0(u_R)-f_0'(u_R)), \\
R_{R,0;P}(z)  & = & \frac{1}{z} e^{u_S/z}u_R, \\
R_{S,P}(z) & = & e^{u_S/z}.
\een

\ben
R_{P,0;Q}(z)
& = & \frac{e^{u_S/z}}{z^2}f'(u_R)
+ \frac{1}{z}e^{u_S/z}u_Q, \\
R_{Q,0;Q}(z)
& = & \frac{e^{u_S/z}}{z}f_0''(u_R), \\
R_{R,0;Q}(z) & = &   e^{u_S/z}, \\
R_{S,0;Q}(z) & = & 0.
\een

\begin{align*}
R_{P,0;R}(z) & = \frac{1}{z} e^{u_S/z}u_R, &
R_{Q,0;R}(z) & = e^{u_S/z}, &
R_{R,0;R}(z) & = 0, &
R_{S,0;R}(z) & = 0.
\end{align*}

\begin{align*}
R_{P,0;S}(z) & = e^{u_S/z}, &
R_{Q,0;S}(z) & = 0, &
R_{R,0;S}(z) & = 0, &
R_{S,0;S}(z) & = 0.
\end{align*}
Therefore, the integrable hierarchy in genus zero
can be written down explicitly as operations
of the loop operators on the order parameters as follows.
\ben
&& D_{P,z}(u_P)
=\sum_{n\geq 0} z^{-n-1} \frac{\pd u_S}{\pd t_{n, P}} \\
& = & \biggl[\frac{ e^{u_S/z}}{z} u_P
+ \frac{ e^{u_S/z}}{z^2} u_Qu_R
+ \frac{e^{u_S/z}}{z^3}
(u_Rf'_0(u_R)-f_0(u_R))\biggr]'.
\een
\ben
D_{P,z}(u_Q)
& = & \sum_{n\geq 0} z^{-n-1} \frac{\pd u_R}{\pd t_{n, P}} \\
& = & \biggl[\frac{1}{z} e^{u_S/z}u_Q
+ \frac{1}{z^2} e^{u_S/z} (u_Rf''_0(u_R)-f'_0(u_R)) \biggr]'.
\een

\ben
D_{P,z}(u_R)=\sum_{n\geq 0} z^{-n-1} \frac{\pd u_Q}{\pd t_{n, P}}
= \biggl[ \frac{1}{z} e^{u_S/z}u_R \biggr]'.
\een

\ben
D_{P,z}(u_S)=\sum_{n\geq 0} z^{-n-1} \frac{\pd u_P}{\pd t_{n, P}}
= \biggl[ e^{u_S/z}\biggr]'.
\een

\ben
D_{Q,z}(u_P)
& = & \sum_{n\geq 0} z^{-n-1}
\frac{\pd u_S}{\pd t_{n, Q}}
= \biggl[\frac{1}{z}e^{u_S/z}u_Q
+ \frac{e^{u_S/z}}{z^2} f'_0(u_R)  \biggr]'.
\een

\ben
D_{Q,z}(u_Q)
& = & \sum_{n=0}^\infty z^{-n-1}
\frac{\pd u_R}{\pd t_{n,Q}}
=\biggl[ \frac{e^{u_S/z}}{z} f''_0(u_R) \biggr]'.
\een
\ben
D_{Q,z}(u_R) = \sum_{n=0}^\infty z^{-n-1}\frac{\pd u_Q}{\pd t_{n,Q}}
= (e^{u_S/z})'.
\een
\ben
D_{Q,z}(u_S) = \sum_{n=0}^\infty z^{-n-1}\frac{\pd u_P}{\pd t_{n,Q}} =0.
\een
\ben
D_{R,z}(u_P)=\sum_{n=0}^\infty z^{-n-1}\frac{\pd u_S}{\pd t_{n,R}}
=\biggl[ \frac{1}{z} e^{u_S/z}u_R\biggr]'.
\een
\ben
D_{R,z}(u_Q)=\sum_{n=0}^\infty z^{-n-1}\frac{\pd u_R}{\pd t_{n,R}}
=\biggl[ e^{u_S/z}\biggr]'.
\een
\ben
D_{R,z}(u_R)=\sum_{n=0}^\infty z^{-n-1}\frac{\pd u_Q}{\pd t_{n,R}}
= 0.
\een
\ben
D_{R,z}(u_S)=\sum_{n=0}^\infty z^{-n-1}\frac{\pd u_P}{\pd t_{n,R}}
= 0.
\een
\ben
D_{S,z}(u_P)=\sum_{n=0}^\infty z^{-n-1}\frac{\pd u_S}{\pd t_{n,S}}
=\biggl[ e^{u_S/z}\biggr]'.
\een
\ben
D_{S,z}(u_Q)=\sum_{n=0}^\infty z^{-n-1}\frac{\pd u_R}{\pd t_{n,S}}
= 0.
\een
\ben
D_{S,z}(u_R)=\sum_{n=0}^\infty z^{-n-1}\frac{\pd u_Q}{\pd t_{n,S}}
= 0.
\een
\ben
D_{S,z}(u_S)=\sum_{n=0}^\infty z^{-n-1}\frac{\pd u_P}{\pd t_{n,S}}
= 0.
\een

\subsection{Three-point functions}

We now show how to combine the results of the
preceding two Subsections
to obtain a recursive algorithm to compute
$n$-point functions of
genus zero GW invariants of the quintic
defined as follows: For $n \geq 3$, define:
\be
\begin{split}
& \cV_{\alpha_1, \dots, \alpha_n}(z_1, \dots, z_n):
= D_{\alpha_1,z_1} \cdots D_{\alpha_{n},z_{n}}F_0  \\
= & \sum_{m_1, \dots, m_n \geq 0}
z_1^{-m_1-1} \cdots z_n^{-m_n-1}
\frac{\pd^nF_0}{\pd t^{\alpha_1}_{m_1} \cdots
\pd t_{m_n}^{\alpha_n}}.
\end{split}
\ee

Their computation is based on the following observation.
When $n \geq 3$,
the $n$-point functions can be computed by applying
the loop operators repeatedly on the two-point functions:
\be
\cV_{\alpha_1, \dots, \alpha_n}(z_1, \dots, z_n)
= D_{\alpha_1,z_1} \cdots D_{\alpha_{n-2},z_{n-2}}
\cV_{\alpha_{n-1}, \alpha_n}(z_{n-1}, z_n).
\ee
For example,
\ben
\cV_{P,Q,R}(z_1,z_2,z_3)
& = & D_{P,z_1} \cV_{Q,R}(z_2,z_3) \\
& = & D_{P,z_1} \frac{e^{u_S/z_2+u_S/z_3}-1}{z_2+z_3} \\
& = & \frac{e^{u_S/z_2+u_S/z_3}}{z_2+z_3}\cdot
(\frac{1}{z_2}+\frac{1}{z_3})
D_{P,z_1}u_S \\
& = & \frac{e^{u_S/z_2+u_S/z_3}}{z_2+z_3}\cdot
(\frac{1}{z_2}+\frac{1}{z_3})
\cdot \biggl[ e^{u_S/z_1}\biggr]'.
\een
After simplification we get:
\ben
\cV_{P,Q,R}(z_1,z_2,z_3)
= \frac{e^{u_S/z_1}}{z_1} \cdot
\frac{e^{u_s/z_2}}{z_2} \cdot \frac{e^{u_S/z_3}}{z_3}
\cdot u_S'.
\een
In this way we get the other nonvanishing three-point functions (up to permutations):
\ben
\cV_{Q,Q,Q} = \frac{e^{u_S/z_1}}{z_1} \cdot
\frac{e^{u_s/z_2}}{z_2} \cdot \frac{e^{u_S/z_3}}{z_3}
\cdot u_S'f'''_0(u_R).
\een
\ben
\cV_{P,Q,Q}(z_1,z_2,z_3)
& = &\frac{e^{u_S/z_1}}{z_1} \cdot
\frac{e^{u_s/z_2}}{z_2} \cdot \frac{e^{u_S/z_3}}{z_3} \\
&& \cdot \biggl(u_R'f'''_0(u_R)
+ \frac{u_R f'''_0(u_R)u_S'}{z_1}
+ \frac{f_0''(u_R)u_S'}{z_2}
+ \frac{f_0''(u_r)u_S'}{z_3} \biggr).
\een

\ben
\cV_{S,P,S}(z_1,z_2,z_3) = V_{R,P,S}(z_1,z_2,z_3)
=V_{Q,P,S}(z_1,z_2,z_3) = 0.
\een

\ben
\cV_{P,P,S}(z_1,z_2,z_3)
= \frac{e^{u_S/z_1}}{z_1} \cdot
\frac{e^{u_s/z_2}}{z_2} \cdot \frac{e^{u_S/z_3}}{z_3}
\cdot u_S'.
\een

\ben
\cV_{P,P,R}(z_1,z_2,z_3)
=  \frac{e^{u_S/z_1}}{z_1} \cdot
\frac{e^{u_s/z_2}}{z_2} \cdot \frac{e^{u_S/z_3}}{z_3}
\cdot \biggl[ u_R'+
\biggl(\frac{1}{z_1}+\frac{1}{z_2}+\frac{1}{z_3}
\biggr) u_R u_S' \biggr].
\een

\ben
\cV_{P,P,Q}(z_1,z_2,z_3)
& = &  \frac{e^{u_S/z_1}}{z_1} \cdot
\frac{e^{u_s/z_2}}{z_2} \cdot \frac{e^{u_S/z_3}}{z_3}
\cdot \biggl[
u_Q'+(\frac{1}{z_1}+\frac{1}{z_2}
+ \frac{1}{z_3})u_Qu_S' \\
& + & (\frac{1}{z_3^2}-\frac{1}{z_1^2} - \frac{1}{z_2^2})
u_s'f'_0(u_R) \\
& + & \biggl(\frac{u_R'}{z_3}
+(\frac{1}{z_1z_3}+\frac{1}{z_2z_3}+\frac{1}{z_1^2}
+ \frac{1}{z_2^2})u_S'u_R\biggr) f''_0(u_R) \\
& + & \biggl((\frac{1}{z_1}+\frac{1}{z_2})u_Ru_R'
+ \frac{u_R^2}{z_1z_2}\biggr)f'''_0(u_R)\biggr].
\een
We omit a similar formula for $\cV_{P,P,P}$.

\subsection{A formula for general $n$-point functions}

By repeatedly applying the loop operators
on the nonvanishing three-point functions,
one gets $n$-point functions.
Among them,
we are particularly interested in
$n$-point functions of the form
\ben
&& V_{P, \dots, P, Q,R}(z_1, \dots, z_n, z_{n+1},z_{n+2}) \\
& = &
V_{P, \dots, P,P,S}(z_1, \dots, z_n, z_{n+1},z_{n+2}).
\een
To see why we expect some nice formulas for them,
let us first do the calculations for $n=2$ and $n=3$.

For $n=2$,
we have
\ben
&& V_{P,P,Q,R}(z_1,z_2,z_3,z_4) = D_{P,z_1}
\biggl(\frac{e^{u_S/z_1}}{z_1} \cdot
\frac{e^{u_s/z_2}}{z_2} \cdot \frac{e^{u_S/z_3}}{z_3}
\cdot u_S' \biggr) \\
& = & \prod_{i=1}^3 \frac{e^{u_S/z_i}}{z_i} \cdot
\biggl(\sum_{i=1}^3 \frac{1}{z_i} \cdot D_{P,z_1}u_S \cdot u_S'
+ (D_{P,z_1}u_S)'\biggr) \\
& = & \prod_{i=1}^3 \frac{e^{u_S/z_i}}{z_i} \cdot
\biggl(\sum_{i=1}^3 \frac{1}{z_i} \cdot (e^{u_S/z_1})'
+ (e^{u_S/z_1})''\biggr) \\
& = & \prod_{i=1}^3 \frac{e^{u_S/z_i}}{z_i} \cdot
\biggl(\sum_{i=1}^3 \frac{1}{z_i} \cdot \frac{e^{u_S/z_1}}{z_1}u_S'
+ \frac{e^{u_S/z_1}}{z_1^2} (u_S')^2
+\frac{e^{u_S/z_1}}{z_1}u_S''
\biggr) \\
& = & \prod_{i=1}^4 \frac{e^{u_S/z_i}}{z_i} \cdot
\biggl( u_S''
+ \sum_{i=1}^4 \frac{1}{z_i}(u_S')^2
\biggr).
\een

For $n=3$,
one can similarly find that
\ben
&& V_{P,P,P,Q,R}(z_1,\dots,z_5)  \\
& = & \prod_{i=1}^5 \frac{e^{u_S/z_i}}{z_i} \cdot
\biggl( u_S'''
+ 3 \cdot \sum_{i=1}^5 \frac{1}{z_i} \cdot u_S' u_S''
+ \biggl(\sum_{i=1}^5 \frac{1}{z_i} \biggr)^2 (u_S')^3
\biggr).
\een
The situation is similar to Section 7.6 in \cite{Zhou-1D}.
One can show that for  $n \geq 1$,
\be
\begin{split}
& \cV_{P, \dots, P, Q, R}(z_1, \dots, z_{n+2})
= \cV_{P, \dots, P, P, S}(z_1, \dots, z_{n+2}) \\
= &  \prod_{i=1}^{n+2} \frac{e^{u_S/z_i}}{z_i} \cdot
n! \sum_{\substack{m_1, \dots, m_n \geq 0 \\
\sum_{j=1}^n m_jj = n }}
\biggl(\sum_{i=1}^{n+2} \frac{1}{z_i} \biggr)^{m_1+\cdots + m_n-1}
\prod_{j=1}^n \frac{u_{S,j}^{m_j}}{(j!)^{m_j} m_j!}.
\end{split}
\ee
The coefficients $(1),(1,1),(1,3,1)$, $(1,4,3,6,1)$
are integer  sequence A036040 or A080575
on \cite{OEIS}. They are related to the complete Bell polynomials:
\be
\exp \sum_{j=1}^\infty \frac{x_j}{j!}t^j
= \sum_{n=0}^\infty t^n
 \sum_{\substack{m_1, \dots, m_n \geq 0 \\
\sum_{j=1}^n m_jj = n }}
\prod_{j=1}^n \frac{x_j^{m_j}}{(j!)^{m_j} m_j!}.
\ee

A more compact formula is as follows:
\be
\begin{split}
& \cV_{P, \dots, P, Q, R}(z_1, \dots, z_{n+2})
= \cV_{P, \dots, P, P, S}(z_1, \dots, z_{n+2}) \\
= &
\biggl( \sum_{i=1}^{n+2} \frac{1}{z_i}\biggr)^{-1} \cdot
\biggl(\frac{\pd}{\pd t^P_0} \biggr)^n
\biggl( \prod_{i=1}^{n+2} \frac{e^{u_S/z_i}}{z_i} \biggr).
\end{split}
\ee

In the same fashion we also get:
\be
\begin{split}
& \cV_{P, \dots, P, R}(z_1, \dots, z_{n+2}) \\
= &  \prod_{i=1}^{n+2} \frac{e^{u_S/z_i}}{z_i} \cdot
n! \sum_{\substack{l, m_1, \dots, m_n \geq 0 \\
l+\sum_{j=1}^n m_jj = n }}
\biggl(\sum_{i=1}^{n+2} \frac{1}{z_i} \biggr)^{m_1+\cdots + m_n-1} \\
& \cdot
\frac{u_{R,l}}{l!}
\prod_{j=1}^n \frac{u_{S,j}^{m_j}}{(j!)^{m_j} m_j!}.
\end{split}
\ee
A more compact formula is as follows:
\be
\begin{split}
\cV_{P, \dots, P, P, R}(z_1, \dots, z_{n+2})
=
\biggl( \sum_{i=1}^{n+2} \frac{1}{z_i}\biggr)^{-1} \cdot
\biggl(\frac{\pd}{\pd t^P_0} \biggr)^n
\biggl( u_R \prod_{i=1}^{n+2} \frac{e^{u_S/z_i}}{z_i} \biggr).
\end{split}
\ee
For example,
\ben
\cV_{P,P,R}(z_1,z_2,z_3)
= \prod_{i=1}^3 \frac{e^{u_S/z_i}}{z_i} \cdot
\cdot \biggl[ u_R'+
\sum_{i=1}^3 \frac{1}{z_i}u_S' \biggr].
\een

\ben
&& \cV_{P,P,P,R}(z_1,z_2,z_3,z_4) \\
& = & \prod_{i=1}^4 \frac{e^{u_S/z_i}}{z_i}
\cdot \biggl[ u_{R,2}+
\sum_{i=1}^4 \frac{1}{z_i} (2u_{S,1} u_{R,1}+u_{S,2}u_R)
+ \biggl(\sum_{i=1}^4 \frac{1}{z_i} \biggr)^2 u_{S,1}^2u_R \biggr].
\een

\ben
&& \cV_{P,\dots,P,R}(z_1,\dots,z_5) \\
& = & \prod_{i=1}^5 \frac{e^{u_S/z_i}}{z_i}
\cdot \biggl[ u_{R,3}+
\sum_{i=1}^5 \frac{1}{z_i}\cdot
(3 u_{S,1} u_{R,2} +3 u_{S,2} u_{R,1}
+ u_{S,3}u_R) \\
& + & \biggl(\sum_{i=1}^5 \frac{1}{z_i} \biggr)^2
(3u_{S,1}^2u_{R,1} +3u_{S,1}u_{S,2}u_R)
+ \biggl(\sum_{i=1}^5 \frac{1}{z_i} \biggr)^3
\cdot u_{S,1}^3u_R \biggr].
\een

\subsection{Genus zero free energy of the quintic}

Let us now compute the genus zero partition function
of the quintic on the big phase space.
Recall the general formula
\be
F_0(t) = \frac{1}{2}
\sum_{m,n=0} R_{\alpha,m;\beta,n}
\tilde{t}_m^\alpha \tilde{t}^\beta_n,
\ee
where $\tilde{t^\alpha}_m = t_m^\alpha
- \delta_{m,1}\delta_{\alpha,0}$,
and
\be
\cV_{\alpha,\beta}(z_1,z_2)
=  \sum_{m,n=0}^\infty R_{\alpha,m;\beta,n}(U)
z_1^{-m-1}z_2^{-n-1}.
\ee
Since for the quintic $\cV_{\alpha,\beta}(z_1,z_2)$
is already explicitly known from \S \ref{sec:Two-Point},
so it is possible to explicitly compute $F_0$ on the big phase
space for the quintic. The result is the following:

\begin{thm} \label{thm:F0}
We have the following formulas for $F_0$ of the quintic on the big phase space:
\be
F_0 = \frac{f_0(u_R)}{u_{S,1}^2}
+ \sum_{m,n=0}^\infty  \frac{u_S^{m+n+1}}{m+n+1}
\frac{\tilde{t}^P_m}{m!} \frac{t^S_n}{n!}
+ \sum_{m,n=0}^\infty  \frac{u_S^{m+n+1}}{m+n+1}
\frac{t^Q_m}{m!} \frac{t^R_n}{n!}.
\ee
\be
F_0 = \frac{1}{u_{S,1}^2}
\biggl( f_0(u_R)
+ \sum_{m,n=0}^\infty
\frac{\tilde{t}^P_m t^S_n + t^Q_m t^R_n}{m+n+1}u_S
\frac{\pd u_S}{\pd t^P_m}\frac{\pd u_S}{\pd t^P_n} \biggr).
\ee
\end{thm}

\begin{proof}
From the expansion:
\ben
&& \cV_{P,P}(z_1,z_2) \\
& = & \frac{e^{u_S/z_1+u_S/z_2}}{z_1z_2}
\biggl(u_P
+ (\frac{1}{z_1}+\frac{1}{z_2})u_Qu_R
+ (\frac{1}{z_1^2}-\frac{1}{z_1z_2}
+ \frac{1}{z_2^2})  (u_Rf'(u_R)-2f(u_R)) \\
& + & \frac{1}{z_1z_2} u_R
(u_Rf''_0(u_R)-f_0'(u_R)) \biggr) \\
& = & \sum_{m,n=0}^\infty \frac{u_S^m}{m!}\frac{u_S^n}{n!}
z_1^{-m-1}z_2^{-n-1}
\biggr( u_P
+ (\frac{1}{z_1}+\frac{1}{z_2})u_Qu_R \\
& + & (\frac{1}{z_1^2}-\frac{1}{z_1z_2}
+ \frac{1}{z_2^2}) (u_Rf'(u_R)-2f(u_R)) \\
& + & \frac{1}{z_1z_2} u_R
(u_Rf''_0(u_R)-f_0'(u_R)) \biggr) \\
& = & \sum_{m,n=0}^\infty
z_1^{-m-1}z_2^{-n-1}
\biggr( u_P \frac{u_S^m}{m!}\frac{u_S^n}{n!}
+ u_Qu_R \frac{u_S^{m-1}}{(m-1)!}\frac{u_S^n}{n!}   \\
& + & u_Qu_R\frac{u_S^m}{m!}\frac{u_S^{n-1}}{(n-1)!}+
(u_Rf'(u_R)-2f(u_R)) \frac{u_S^{m-2}}{(m-2)!}\frac{u_S^n}{n!} \\
& - & (u_Rf'(u_R)-2f(u_R)) \frac{u_S^{m-1}}{(m-1)!} \frac{u_S^{n-1}}{(n-1)!} \\
& + & (u_Rf'(u_R)-2f(u_R)) \frac{u_S^{m}}{m!}
\frac{u_S^{n-2}}{(n-2)!} \\
& + & (u_R^2f''_0(u_R)-u_Rf_0'(u_R)) \frac{u_S^{m-1}}{(m-1)!}\frac{u_S^{n-1}}{(n-1)!} \biggr),
\een
we get:
\ben
&& \sum_{m,n=0}^\infty R_{P,m;P,n} \tilde{t}^P_m\tilde{t}^P_n \\
& = &
u_P \sum_{m=0}^\infty   \tilde{t}^P_m \frac{u_S^m}{m!}
\sum_{n=0}^\infty \tilde{t}^P_n \frac{u_S^n}{n!}
+ u_Qu_R \sum_{m=1}^\infty\tilde{t}^P_m \frac{u_S^{m-1}}{(m-1)!}
\sum_{n=0}^\infty \tilde{t}^P_n \frac{u_S^n}{n!}   \\
& + & u_Qu_R\sum_{m=0}^\infty \tilde{t}^P_m\frac{u_S^m}{m!}
\sum_{n=0}^\infty \tilde{t}^P_n \frac{u_S^{n-1}}{(n-1)!} \\
& + &
(u_Rf'(u_R)-2f(u_R))\sum_{m=2}^\infty
\tilde{t}^P_m \frac{u_S^{m-2}}{(m-2)!}
\sum_{n=0}^\infty \tilde{t}^P_n \frac{u_S^n}{n!} \\
& - & (u_Rf'(u_R)-2f(u_R))
\sum_{m=1}^\infty \tilde{t}^P_m\frac{u_S^{m-1}}{(m-1)!}
\sum_{n=1}^\infty \tilde{t}^P_n\frac{u_S^{n-1}}{(n-1)!} \\
& + & (u_Rf'(u_R)-2f(u_R))
\sum_{m=0}^\infty   \tilde{t}^P_m\frac{u_S^{m}}{m!}
\sum_{n=2}^\infty \tilde{t}^P_n\frac{u_S^{n-2}}{(n-2)!} \\
& + & (u_R^2f''_0(u_R)-u_Rf_0'(u_R))
\sum_{m=1}^\infty\tilde{t}^P_m\frac{u_S^{m-1}}{(m-1)!}
\sum_{n=1}^\infty \tilde{t}^P_n\frac{u_S^{n-1}}{(n-1)!} \biggr) \\
& = & \frac{1}{u_{S,1}^2} \biggl(u_R^2f_0''(u_R)
- 2u_R f_0'(u_R) + 2 f_0(u_R) \biggr).
\een
Here we have used \eqref{eqn:u-S} and \eqref{eqn:u-S1}.
From the following expansion:
\ben
\cV_{P,Q}(z_1,z_2)
& = & \frac{e^{u_S/z_1+u_S/z_2}}{z_1z_2}
\biggr( u_Q
+ \frac{1}{z_1}u_Rf_0''(u_R)+ (\frac{1}{z_2}-\frac{1}{z_1})
f_0'(u_R) \biggr) \\
& = & \sum_{m,n=0}^\infty \frac{u_S^m}{m!}\frac{u_S^n}{n!}
z_1^{-m-1}z_2^{-n-1}
\biggr( u_Q
+ \frac{1}{z_1}u_Rf_0''(u_R)+ (\frac{1}{z_2}-\frac{1}{z_1})
f_0'(u_R) \biggr) \\
& = & \sum_{m,n=0}^\infty
z_1^{-m-1}z_2^{-n-1}
\biggr( u_Q \frac{u_S^m}{m!}\frac{u_S^n}{n!}
+ \frac{u_S^{m-1}}{(m-1)!}\frac{u_S^n}{n!} u_Rf_0''(u_R) \\
& - & f_0'(u_R) \frac{u_S^{m-1}}{(m-1)!}\frac{u_S^n}{n!}
+ f'_0(u_R) \frac{u_S^m}{m!}\frac{u_S^{n-1}}{(n-1)!} \biggr),
\een
we get:
\ben
&& \sum_{m,n=0}^\infty R_{P,m;Q,n} \tilde{t}^P_m\tilde{t}^Q_n \\
& = & \sum_{m,n=0}^\infty  \tilde{t}^P_mt^Q_n
\biggr( u_Q \frac{u_S^m}{m!}\frac{u_S^n}{n!}
+ \frac{u_S^{m-1}}{(m-1)!}\frac{u_S^n}{n!} u_Rf_0''(u_R) \\
& + & f_0'(u_R) \frac{u_S^{m-1}}{(m-1)!}\frac{u_S^n}{n!}
- f'_0(u_R) \frac{u_S^m}{m!}\frac{u_S^{n-1}}{(n-1)!} \biggr) \\
& = & u_Q \sum_{m=0}^\infty  \tilde{t}^P_m  \frac{u_S^m}{m!} \cdot
\sum_{n=0}^\infty t^Q_n  \frac{u_S^n}{n!}
+ u_Rf_0''(u_R)
\sum_{m=1}^\infty  \tilde{t}^P_m\frac{u_S^{m-1}}{(m-1)!}\cdot
\sum_{n=0}^\infty t^Q_n \frac{u_S^n}{n!}   \\
& - & f_0'(u_R) \sum_{m=1}^\infty
\tilde{t}^P_m\frac{u_S^{m-1}}{(m-1)!} \cdot
\sum_{n=0}^\infty t^Q_n \frac{u_S^n}{n!}
- f'_0(u_R)\sum_{m=1}^\infty
\tilde{t}^P_m \frac{u_S^m}{m!}\cdot
\sum_{n=1}^\infty t^Q_n \frac{u_S^{n-1}}{(n-1)!}   \\
& = & \frac{1}{u_{S,1}^2} \biggl(
- u_R^2f_0''(u_R) + u_Rf_0'(u_R) \biggr).
\een
Here we have used \eqref{eqn:u-S} and \eqref{eqn:u-R-sol}.
From the following expansion:
\ben
\cV_{P,R}(z_1,z_2)
=  \frac{ e^{u_S/z_1+u_S/z_2} }{z_1z_2} u_R
= u_R \sum_{m,n=0}^\infty \frac{u_S^m}{m!}\frac{u_S^n}{n!}z_1^{-m-1}z_2^{-n-1},
\een
we get:
\ben
&& \sum_{m,n=0}^\infty R_{P,m;R,n} \tilde{t}^P_m\tilde{t}^R_n
= u_R \sum_{m,n=0}^\infty  \frac{u_S^m}{m!} \frac{u_S^n}{n!}
\tilde{t}^P_m t^R_n
= u_R \sum_{m=0}^\infty \tilde{t}^P_m  \frac{u_S^m}{m!}
\cdot  \sum_{n=0}^\infty t^R_n \frac{u_S^n}{n!}
= 0.
\een
From the following  expansion:
\ben
\cV_{P,S}(z_1,z_2)
& = &\frac{e^{u_R/z_1+u_R/z_2}-1}{z_1+z_2}
= \frac{1}{z_1z_2}\sum_{M=1}^\infty
(\frac{1}{z_1}+\frac{1}{z_2})^{M-1} \frac{u_S^M}{M!} \\
& = & \sum_{m,n=0}^\infty \frac{u_S}{m+n+1}\frac{u_S^m}{m!}
\frac{u_S^n}{n!}z_1^{-m-1}z_2^{-n-1},
\een
we get:
\ben
&& \sum_{m,n=0}^\infty R_{P,m;S,n} \tilde{t}^P_m\tilde{t}^S_n
= \sum_{m,n=0}^\infty  \frac{u_S^{m+n+1}}{m+n+1}
\frac{\tilde{t}^P_m}{m!} \frac{t^S_n}{n!}.
\een
From the following  expansion:
\ben
\cV_{Q,Q}(z_1,z_2) =
\frac{ e^{u_S/z_1+u_S/z_2} }{z_1z_2} f_0''(u_R)
= f_0''(u_R)
\sum_{m,n=0}^\infty \frac{u_s^m}{m!}\frac{u_S^n}{n!}z_1^{-m-1}z_2^{-n-1},
\een
we get:
\ben
&& \sum_{m,n=0}^\infty R_{Q,m;Q,n} \tilde{t}^Q_m\tilde{t}^Q_n
= \sum_{m,n=0}^\infty  \frac{u_S^m}{m!} \frac{u_S^n}{n!} f''_0(u_R)
t^Q_m t^Q_n \\
& = & f''_0(u_R)
\biggl(\sum_{m=0}^\infty t^Q_m  \frac{u_S^m}{m!} \biggr)^2
=\frac{1}{u_{S,1}^2} u_R^2 f''(u_R).
\een
Here we have used \eqref{eqn:u-R-sol}.
From the following expansion:
\ben
\cV_{Q,R}(z_1,z_2)
= \frac{e^{u_S/z_1+u_S/z_2}-1}{z_1+z_2}
= \sum_{m,n=0}^\infty \frac{u_S}{m+n+1}\frac{u_S^m}{m!}
\frac{u_S^n}{n!},\een
we get:
\ben
&& \sum_{m,n=0}^\infty R_{Q,m;R,n} \tilde{t}^Q_m\tilde{t}^R_n
= \sum_{m,n=0}^\infty  \frac{u_S^{m+n+1}}{m+n+1}
\frac{t^Q_m}{m!} \frac{t^R_n}{n!}.
\een
The proof is finished by putting all these together.
\ben F_0
& = & \frac{1}{2} \frac{1}{u_{S,1}^2} \biggl(u_R^2f_0''(u_R)
- 2u_R f_0'(u_R) + 2 f_0(u_R) \biggr) \\
& + & \frac{1}{u_{S,1}^2} \biggl(
- u_R^2f_0''(u_R) + u_Rf_0'(u_R) \biggr)
+ \half \frac{1}{u_{S,1}^2} u_R^2 f_0''(u_R) \\
& + & \sum_{m,n=0}^\infty  \frac{u_S^{m+n+1}}{m+n+1}
\frac{\tilde{t}^P_m}{m!} \frac{t^S_n}{n!}
+ \sum_{m,n=0}^\infty  \frac{u_S^{m+n+1}}{m+n+1}
\frac{t^Q_m}{m!} \frac{t^R_n}{n!} \\
& = &  \frac{f_0(u_R)}{u_{S,1}^2}
+ \sum_{m,n=0}^\infty  \frac{u_S^{m+n+1}}{m+n+1}
\frac{\tilde{t}^P_m}{m!} \frac{t^S_n}{n!}
+ \sum_{m,n=0}^\infty  \frac{u_S^{m+n+1}}{m+n+1}
\frac{t^Q_m}{m!} \frac{t^R_n}{n!}.
\een
This proves the first formula.
The second formula is proved by applying \eqref{eqn:u-S-tP-n}.
\end{proof}

\subsection{Genus zero one-point functions}

\begin{thm}
For the quintic the following formulas
for the one-point functions in genus zero hold:
\be
\begin{split}
D_{P,z}F_0
= & \sum_{m,n=0}^\infty  \frac{u_S^{m+n+1}}{m+n+1}
\frac{t^S_m}{m!} \frac{z^{-n-1}}{n!}  \\
+ & \frac{e^{u_S/z}}{z} \cdot
\biggl[ \frac{u_Qu_R}{u_{S,1}}
- \frac{u_Rf''_0(u_R)- f_0'(u_R) }{u_{S,1}^2}u_{R,1} \\
+ &\frac{u_R f'_0(u_R)-2f_0(u_R)}{u_{S,1}}
\biggl( \frac{u_{S,2}}{u_{S,1}^2} + \frac{1}{z}\biggr)
\biggr].
\end{split}
\ee
\be
D_{Q,z}F_0
=\frac{f_0'(u_R)}{u_{S,1} }   \cdot \frac{e^{u_S/z}}{z}
+ \sum_{m,n=0}^\infty  \frac{u_S^{m+n+1}}{m+n+1}
\frac{t^R_m}{m!} \frac{z^{-n-1}}{n!}.
\ee
\be
D_{R,z}F_0
= \sum_{m,n=0}^\infty  \frac{u_S^{m+n+1}}{m+n+1}
\frac{t^Q_m}{m!} \frac{z^{-n-1}}{n!}.
\ee
\be
D_{S,z}F_0
= \sum_{m,n=0}^\infty  \frac{u_S^{m+n+1}}{m+n+1}
\frac{\tilde{t}^P_m}{m!} \frac{z^{-n-1}}{n!}.
\ee

\end{thm}

\begin{proof}
Note  we have
\begin{align*}
D_{S,z}u_S& =D_{S,z}u_R  = 0, &   \\
D_{S,z} \tilde{t}^P_n & = D_{S,z}t^Q_n = D_{S,z} t^R_n = 0, &
D_{S,z}t^S_n & = z^{-n-1}.
\end{align*}
So we get:
\ben
D_{S,z}F_0
& = &  D_{S,z} \biggl(\frac{f_0(u_R)}{u_{S,1}^2}
+\sum_{m,n=0}^\infty  \frac{u_S^{m+n+1}}{m+n+1}
\frac{\tilde{t}^P_m}{m!} \frac{t^S_n}{n!}
+ \sum_{m,n=0}^\infty  \frac{u_S^{m+n+1}}{m+n+1}
\frac{t^Q_m}{m!} \frac{t^R_n}{n!} \biggr) \\
& = & \sum_{m,n=0}^\infty  \frac{u_S^{m+n+1}}{m+n+1}
\frac{\tilde{t}^P_m}{m!} \frac{z^{-n-1}}{n!}.
\een

Similar,  we have
\begin{align*}
D_{R,z}u_S& =D_{R,z}u_R = 0, \\
D_{R,z} \tilde{t}^P_n & = D_{R,z}t^Q_n = D_{R,z} t^S_n = 0, &
D_{R,z}t^R_n & = z^{-n-1}.
\end{align*}
So we get
\ben
D_{R,z}F_0
& = &  D_{R,z} \biggl(\frac{f_0(u_R)}{u_{S,1}^2}
+ \sum_{m,n=0}^\infty  \frac{u_S^{m+n+1}}{m+n+1}
\frac{\tilde{t}^P_m}{m!} \frac{t^S_n}{n!}
+ \sum_{m,n=0}^\infty  \frac{u_S^{m+n+1}}{m+n+1}
\frac{t^Q_m}{m!} \frac{t^R_n}{n!} \biggr) \\
& = & \sum_{m,n=0}^\infty  \frac{u_S^{m+n+1}}{m+n+1}
\frac{t^Q_m}{m!} \frac{z^{-n-1}}{n!}.
\een

Similar,  we have
\begin{align*}
D_{Q,z}u_S& =0, & D_{Q,z}u_R &= \frac{e^{u_S/z}}{z}u_{S,1}, \\
D_{Q,z} \tilde{t}^P_n & = D_{R,z}t^R_n = D_{R,z} t^S_n = 0, &
D_{Q,z}t^Q_n & = z^{-n-1}.
\end{align*}
So we get
\ben
D_{Q,z}F_0
& = &  D_{Q,z} \biggl(\frac{f_0(u_R)}{u_{S,1}^2}
+ \sum_{m,n=0}^\infty  \frac{u_S^{m+n+1}}{m+n+1}
\frac{\tilde{t}^P_m}{m!} \frac{t^S_n}{n!}
+ \sum_{m,n=0}^\infty  \frac{u_S^{m+n+1}}{m+n+1}
\frac{t^Q_m}{m!} \frac{t^R_n}{n!} \biggr) \\
& = & \frac{f_0'(u_R)}{u_{S,1} }   \cdot \frac{e^{u_S/z}}{z}
+ \sum_{m,n=0}^\infty  \frac{u_S^{m+n+1}}{m+n+1}
\frac{t^R_m}{m!} \frac{z^{-n-1}}{n!}.
\een
Finally,
we have:
\begin{align*}
D_{P,z} u_S& = \frac{e^{u_S/z}}{z}u_{S,1}, &
D_{P,z}u_R &= \frac{e^{u_S/z}}{z}
(u_{R,1} + u_R \frac{u_{S,1}}{z}), \\
D_{P,z} t^Q_n & = D_{R,z}t^R_n = D_{R,z} t^S_n = 0, &
D_{P,z} \tilde{t}^P_n & = z^{-n-1}.
\end{align*}
So we get
\ben
&& D_{P,z}F_0 \\
& = &  D_{P,z} \biggl(\frac{f_0(u_R)}{u_{S,1}^2}
+ \sum_{m,n=0}^\infty  \frac{u_S^{m+n+1}}{m+n+1}
\frac{\tilde{t}^P_m}{m!} \frac{t^S_n}{n!}
+ \sum_{m,n=0}^\infty  \frac{u_S^{m+n+1}}{m+n+1}
\frac{t^Q_m}{m!} \frac{t^R_n}{n!} \biggr) \\
& = & \frac{f_0'(u_R)}{u_{S,1}^2}
\cdot \frac{e^{u_S/z}}{z}(u_{R,1} + u_R \frac{u_{S,1}}{z})
- 2\frac{f_0(u_R)}{u_{S,1}^3} \cdot
\biggl(\frac{e^{u_S/z}}{z}u_{S,1} \biggr)'  \\
& + & \sum_{m,n=0}^\infty  \frac{u_S^{m+n+1}}{m+n+1}
\frac{t^S_m}{m!} \frac{z^{-n-1}}{n!} \\
& + & \sum_{m,n=0}^\infty  u_S^{m+n}\cdot
\frac{e^{u_S/z_1}}{z_1}u_{S,1} \cdot
\frac{\tilde{t}^P_m}{m!} \frac{t^S_n}{n!}
+ \sum_{m,n=0}^\infty   u_S^{m+n}\cdot
\frac{e^{u_S/z_1}}{z_1}u_{S,1} \cdot
\frac{t^Q_m}{m!} \frac{t^R_n}{n!}  \\
& = & \sum_{m,n=0}^\infty  \frac{u_S^{m+n+1}}{m+n+1}
\frac{t^S_m}{m!} \frac{z^{-n-1}}{n!} \\
& + &  \frac{e^{u_S/z}}{z}
\biggl(  \frac{f_0'(u_R)}{u_{S,1}^2} \cdot u_{R,1}
- 2\frac{f_0(u_R)}{u_{S,1}^3} \cdot u_{S,2}
+ \frac{1}{z}
\frac{u_R f_0'(u_R) - 2f_0(u_R)}{u_{S,1}}  \biggr) \\
& + & \frac{e^{u_S/z}}{z}u_{S,1} \cdot
\sum_{m=0}^\infty \frac{t^Q_m}{m!} u_S^{m}
\cdot
\sum_{n=0}^\infty   \frac{t^R_n}{n!} u_S^n.
\een
The summations on the third line of the right-hand side of the last equality can be performed as follows.
By \eqref{eqn:u-R-sol}
we have
\be \label{eqn:Sum-t-Qn}
\sum_{n=0}^\infty   \frac{t^Q_n}{n!} u_S^n = \frac{u_R}{u_{S,1}}.
\ee
The summations on the third line of the right-hand side of the last equality can be performed as follows.
By \eqref{eqn:u-Q-Sol}
we have
\ben
\sum_{n=0}^\infty   \frac{t^Q_n}{n!} u_S^n
& = & \frac{u_Q}{u_{S,1}} - f_0''(u_R) \sum_{n=1}^\infty t^Q_n
\frac{u_S^{n-1}}{(n-1)!}
- g_0'(u_R) \sum_{n=2}^\infty t^P_n
\frac{u_S^{n-2}}{(n-2)!}.
\een
Take $\frac{\pd}{\pd t^P}$ on both sides of eqref{eqn:Sum-t-Qn},
we get:
\ben
&& \sum_{n=1}^\infty   \frac{t^Q_n}{(n-1)!} u_S^{n-1} \cdot u_{S,1}
= \frac{u_{R,1}}{u_{S,1}} - \frac{u_R}{u_{S,1}^2}\cdot u_{S,2},
\een
and so
\be
\sum_{n=1}^\infty   \frac{t^Q_n}{(n-1)!} u_S^{n-1}
= \frac{u_{R,1}}{u_{S,1}^2} - \frac{u_R}{u_{S,1}^3}\cdot u_{S,2}.
\ee
From \eqref{eqn:u-S1} we get:
\be
\sum_{n=1}^\infty t^P_n \frac{u_S^{n-1}}{(n-1)!} = 1- \frac{1}{u_{S,1}}.
\ee
After taking $\frac{\pd}{\pd t^P}$ and dividing by $u_{S,1}$ on both sides we get:
\be
\sum_{n=2}^\infty t^P_n \frac{u_S^{n-2}}{(n-2)!} = \frac{u_{S,2}}{u_{S,1}^3}.
\ee
Also note:
$$g_0'(u_R) = u_R f_0''(u_R)-f_0'(u_R),$$
therefore,
\ben
&& \sum_{n=0}^\infty   \frac{t^Q_n}{n!} u_S^n \\
& = & \frac{u_Q}{u_{S,1}} - f_0''(u_R)  \cdot \biggl(\frac{u_{R,1}}{u_{S,1}^2}
- \frac{u_R}{u_{S,1}^3}\cdot u_{S,2} \biggr)
- (u_R f''_0(u_R)-f_0'(u_R)) \frac{u_{S,2}}{u_{S,1}^3} \\
& = & \frac{u_Q}{u_{S,1}} - f_0''(u_R)  \cdot \frac{u_{R,1}}{u_{S,1}^2}
+ f_0'(u_R) \frac{u_{S,2}}{u_{S,1}^3}.
\een
Now we can resume the computation of $D_{P,z}F_0 $:
\ben
&& D_{P,z}F_0 \\
& = & \sum_{m,n=0}^\infty  \frac{u_S^{m+n+1}}{m+n+1}
\frac{t^S_m}{m!} \frac{z^{-n-1}}{n!}  \\
& + & \frac{e^{u_S/z}}{z}
\biggl(  \frac{f_0'(u_R)}{u_{S,1}^2} \cdot u_{R,1}
- 2\frac{f_0(u_R)}{u_{S,1}^3} \cdot u_{S,2}
+ \frac{1}{z} \frac{u_Rf_0'(u_R) - 2f_0(u_R)}{u_{S,1}}
  \biggr) \\
& + & \frac{e^{u_S/z}}{z} \cdot u_R
\cdot \biggl( \frac{u_Q}{u_{S,1}}
- \frac{f''_0(u_R) u_{R,1}}{u_{S,1}^2}
+ \frac{f'_0(u_R)}{u_{S,1}^3} u_{S,2}  \biggr) \\
& = & \sum_{m,n=0}^\infty  \frac{u_S^{m+n+1}}{m+n+1}
\frac{t^S_m}{m!} \frac{z^{-n-1}}{n!}  \\
& + & \frac{e^{u_S/z}}{z} \cdot
\biggl[ \frac{u_Qu_R}{u_{S,1}}
- \frac{u_Rf''_0(u_R)- f_0'(u_R) }{u_{S,1}^2}u_{R,1} \\
& + &\frac{u_R f'_0(u_R)-2f_0(u_R)}{u_{S,1}}
\biggl( \frac{u_{S,2}}{u_{S,1}^2} + \frac{1}{z}\biggr)
\biggr].
\een
\end{proof}

Note for $m \geq 0$,
\ben
\sum_{n=0}^\infty  \frac{u_S^{n+1}}{m+n+1} \frac{z^{-n-1}}{n!}
= e^{u_S/z} \sum_{j=0}^m (-1)^j \frac{m!}{(m-j)!}
\biggl(\frac{z}{u_s}\biggr)^j
-(-1)^mm! \frac{z^m}{u_S^m},
\een
so we get:
\ben
D_{S,z}F_0
& = &\sum_{m=0}^\infty \frac{\tilde{t}^P_mu_S^m}{m!}
\biggl(e^{u_S/z} \sum_{j=0}^m (-1)^j \frac{m!}{(m-j)!}
\biggl(\frac{z}{u_s}\biggr)^j
-(-1)^mm! \frac{z^m}{u_S^m}\biggr) \\
& = & e^{u_S/z} \sum_{m=0}^\infty \tilde{t}^P_m
\sum_{j=0}^m (-1)^j \frac{u_S^{m-j}}{(m-j)!}
 z^j
- \sum_{m=0}^\infty (-1)^m \tilde{t}^P_m z^m \\
& = & e^{u_S/z} \sum_{j=0}^\infty (-z)^j \sum_{k=0}^\infty
\tilde{t}^P_{j+k}  \frac{u_S^k}{k!}
- \sum_{m=0}^\infty (-1)^m \tilde{t}^P_m z^m .
\een
We can also rewrite other
$\frac{\pd F_0}{\pd t^\alpha}$
in the same fashion.

\section{Mean Field Theory Computations for GW Theory  of the Quintic in Genus $g \geq 1$}
\label{sec:Quintic-Higher}

In this Section we extend the discussions
to higher genera.
Again we will see that the selection rules play a crucial role.

We will also write down the integrable hierarchy
associated with
the GW theory of the quintic in all genera.

\subsection{Degree zero part of the free energy
vs. the instanton part of the free energy}

We will split the free energy in genus $g\geq 0$
into two parts:
\be
F_h = F_{g,0} + F_{g,inst},
\ee
where $F_{g,0}$ is the total contributions
of the GW invariants of degree zero
(called the degree zero part of the free energy),
and $F_{g,inst}$ is the total contributions
of the GW invariants of positive degrees
(called the instanton part of the free energy).

The degree 0
gravitational descendent invariants of $X$ are the integrals:
\be
\begin{split}
& \corr{\tau_{k_1}(\cO_{\alpha_1})\cdots
\tau_{k_n}(\cO_{\alpha_n})}_{g,n;0}^X  \\
= & \int_{X\times \Mbar_{g,n}}
\cO_{\alpha_1} \cdots  \cO_{\alpha_n} \psi_1^{k_1}
 \cdots \psi_n^{k_n} \cup e(TX \boxtimes \bE^\vee),
\end{split}
\ee
where $\bE \to \Mbar_{g,n}$ is the Hodge bundle.
When $g=0$,
$\bE$ is trivial,
and the only degree zero contributions are given by:
\be
\corr{\tau_{k_1}(\cO_{\alpha_1})\cdots
\tau_{k_n}(\cO_{\alpha_n})}_{0,n;0}^X  \\
= \int_{X}
\cO_{\alpha_1} \cdots  \cO_{\alpha_n}
\cdot \int_{\Mbar_{0,n}}
\psi_1^{k_1} \cdots  \psi_n^{k_n}.
\ee
By Theorem \ref{thm:F0}
we have seen  that:
\be
\begin{split}
F_{0,0}
= & \frac{1}{u_{S,1}^2} \cdot \frac{5}{6}u_R^3
+ \sum_{m,n=0}^\infty  \frac{u_S^{m+n+1}}{(m+n+1)m!n!}
(\tilde{t}^P_m t^S_n + t^Q_m t^R_n ) \\
= & \frac{1}{u_{S,1}^2}
\biggl(\frac{5}{6}u_R^3
+ \sum_{m,n=0}^\infty
\frac{\tilde{t}^P_m t^S_n + t^Q_m t^R_n}{m+n+1}u_S
\frac{\pd u_S}{\pd t^P_m}\frac{\pd u_S}{\pd t^P_n}
\biggr),
\end{split}
\ee
and
\be
F_{0,inst} = \frac{1}{u_{S,1}^2}
\biggl( f_0(u_R) - \frac{5}{6}u_R^3 \biggr).
\ee

\subsection{Degree zero contribution to genus one free energy}

The $g=1$ case is also exceptional:
\be
e(TX \boxtimes \bE^\vee) = c_3(X) - c_2(X)\lambda_1,
\ee
and for $g > 1$,
\be \label{eqn:Euler-Deg0}
e(TX \boxtimes \bE^\vee) =
\frac{(-1)^g}{2}
(c_3(X) - c_2(X)c_1(X)) \lambda^3_{g-1}.
\ee
For the quintic,
the total Chern class is given by:
\ben
c(X) = (1+H)^5(1+5H)^{-1}|_X = (1+10H^2-40H^3)|_X.
\een
and so
\begin{align*}
c_1(X) & = 0, & c_2(X) & = 10H^2|_X, &
c_3(X) & = -40 H^3|_X.
\end{align*}
So in genus one we have
 \ben
&& \corr{\tau_{k_1}(\cO_{\alpha_1})
\cdots \tau_{k_n}(\cO_{\alpha_n})}_{1,n;0}^X  \\
& = & \int_{X\times \Mbar_{1,n}}
\cO_{\alpha_1} \cdots  \cO_{\alpha_n} \psi_1^{k_1}
 \cdots \psi_n^{k_n} \cup (-40H^3 + 10 H^2 \lambda_1) \\
 & = & - \int_X \cO_{\alpha_1} \cdots  \cO_{\alpha_n} \cup 40H^3 \cdot \int_{\Mbar_{1,n}} \psi_1^{k_1}
 \cdots \psi_n^{k_n} \\
& + & \int_X \cO_{\alpha_1} \cdots  \cO_{\alpha_n} \cup 10H^2 \cdot \int_{\Mbar_{1,n}} \psi_1^{k_1}
 \cdots \psi_n^{k_n} \lambda_1.
\een
So either
\ben
&& \deg \cO_{\alpha_1} +  \cdots  + \deg \cO_{\alpha_n} =0, \\
&& k_1 + \cdots +k_n = n,
\een
or
\ben
&& \deg \cO_{\alpha_1} +  \cdots  + \deg \cO_{\alpha_n} =1, \\
&& k_1 + \cdots +k_n = n-1,
\een

Restricted to the small phase space,
i.e.,
when $k_1 =\cdots = k_n =0$,
we must have $n=1$ and $\deg \cO_{\alpha_1} = 1$,
this gives us
\be
\corr{Q}_{1,1;0}^X = \int_X Q \cup 10 H^2 \cdot \int_{\Mbar_{1,1}}\lambda_1
= 50 \cdot \frac{1}{24} = \frac{25}{12},
\ee
and so there is a contribution of
$ \frac{25}{12}t^Q$
to $F_1^{small}$.

On the big phase space we have
\ben
&& \corr{\tau_{k_1}(Q) \tau_{k_2}(P)
\cdots \tau_{k_n}(P)}_{1,0}^X
= \int_X H \cup 10H^2 \cdot \int_{\Mbar_{1,n}} \psi_1^{k_1}
\cdots \psi_n^{k_n} \lambda_1 \\
& = & \frac{50}{24} \binom{n-1}{k_1, \dots, k_n}.
\een
Their total contributions to $F_{1,0}$ is:
\ben
&& \sum_{n=1}^\infty
\sum_{k_1+\cdots +k_n=n-1} \frac{50}{24} \binom{n-1}{k_1, \dots, k_n} t^Q_{k_1}
\frac{1}{(n-1)!}t^P_{k_2} \cdots t^P_{k_n} \\
& = &  \frac{50}{24}
\sum_{n=1}^\infty \sum_{k_1+\cdots +k_n=n-1}
\frac{t^Q_{k_1}}{k_1!} \cdot \frac{t^P_{k_2}}{k_2!} \cdots
\frac{t^P_{k_n}}{k_n!}
= \frac{25}{12} u_R.
\een
For the other case,
\ben
&& \corr{\tau_{k_1}(P)\cdots \tau_{k_n}(P)}_{1,0}^X
 = - \int_X  40H^3 \cdot \int_{\Mbar_{1,n}} \psi_1^{k_1}
 \cdots \psi_n^{k_n} \\
& = & - 200 \int_{\Mbar_{1,n}} \psi_1^{k_1}
 \cdots \psi_n^{k_n}.
\een
It vanishes unless
\be
k_1+ \cdots + k_n = n.
\ee
Their total contributions to $F_{1,0}$ is:
\ben
&& \sum_{n=1}^\infty \frac{t^P_{k_1}\cdots t^P_{k_n}}{n!}
\corr{\tau_{k_1}(P)\cdots \tau_{k_n}(P)}_{1,0}^X \\
& = & -200 \sum_{n=1}^\infty \frac{t^P_{k_1}\cdots t^P_{k_n}}{n!}
\int_{\Mbar_{1,n}} \psi_1^{k_1} \cdots \psi_n^{k_n} \\
& = & -200 \cdot \frac{1}{24} \log u_{S,1}.
\een
So we get for the quintic the following formula for
the degree zero part $F_{1,0}$ of the free energy in genus one:
\be
F_{1,0} = - \frac{25}{3} \log u_{S,1}
+ \frac{25}{12} u_R.
\ee

\subsection{Degree zero contribution to  free energy in genus $g > 1$}
Now by \eqref{eqn:Euler-Deg0},
when $g > 1$,
\ben
&& \corr{\tau_{k_1}(\cO_{\alpha_1})\cdots
\tau_{k_n}(\cO_{\alpha_n})}_{g,0}^X  \\
& = & \int_{X\times \Mbar_{g,n}}
\cO_{\alpha_1} \cdots  \cO_{\alpha_n} \psi_1^{k_1}
 \cdots \psi_n^{k_n} \cup \frac{(-1)^g}{2}(-40H^3) \lambda_{g-1}^3 \\
 & = & - \int_X \cO_{\alpha_1} \cdots  \cO_{\alpha_n} \cup 40H^3
 \cdot \int_{\Mbar_{g,n}} \psi_1^{k_1}
 \cdots \psi_n^{k_n}   \lambda_{g-1}^3.
\een
Therefore,
we can only have $\cO_{\alpha_1}= \cdots = \cO_{\alpha_n} = P$,
and
\ben
&& \corr{\tau_{k_1}(P)\cdots \tau_{k_n}(P)}_{g,0}^X
 = - \int_X  40H^3 \cdot \int_{\Mbar_{1,n}} \psi_1^{k_1}
 \cdots \psi_n^{k_n} \lambda_{g-1}^3 \\
& = & - 200 \int_{\Mbar_{1,n}} \psi_1^{k_1}
 \cdots \psi_n^{k_n} \lambda_{g-1}^3.
\een
Therefore,
for $g \geq 2$,
\ben
F_{g,0} & = &
\sum_{n=1}^\infty \frac{t^P_{k_1}\cdots t^P_{k_n}}{n!}
\corr{\tau_{k_1}(P)\cdots \tau_{k_n}(P)}_{g,0}^X \\
& = & -200 \sum_{n=1}^\infty \frac{t^P_{k_1}\cdots t^P_{k_n}}{n!}
\int_{\Mbar_{1,n}} \psi_1^{k_1} \cdots \psi_n^{k_n} \lambda_{g-1}^3 \\
& = & -200 \cdot \frac{1}{2(2g-2)!}\frac{|B_{2g-2}|}{2g-2} \frac{|B_{2g}|}{2g} u_{S,1}^{2g-2}.
\een
Here in the second equality we have used a result in
Dubrovin-Liu-Yang-Zhang \cite[Cor. 3.10]{DLYZ}.

\subsection{Instanton part of the free energy in genus $\geq 1$}

We will first consider $F_{g,instanton}$ restricted to the small phase space for $g \geq 1$.
Then we need to consider the correlators
\be
\corr{P^{m_0}Q^{m_1}R^{m_2}S^{m_3}}_{g,m_0+\cdots +m_3 ;d}
\ee
By the selection rules \eqref{eqn:Selection-Rules},
we have
\be
m_0 = m_2+2m_3.
\ee
By the string equation we must have $m_0=m_2=m_3=0$ to get
a nonzero correlator
\be
\corr{Q^{m_1}}_{g,m_1;d} = d^{m_1} \corr{1}_{g,0;d} = d^{m_1} N_{g,d},
\ee
where
\be
N_{g,d} :=\int_{[\Mbar_{g,0}(X,d)]^{virt}} 1.
\ee
So the instanton part of the free energy in genus $g \geq 1$ restricted to the small phase space
is given by
\be
\sum_{d=1}^\infty  q^d
\sum_{m_1=0}^\infty \frac{(t^Q)^{m_1}}{m_1!} \corr{Q^{m_1}}_{g,m_1;d}
= \sum_{d=1}^\infty  q^d N_{g,d} e^{dt^Q}.
\ee
Next we will show that the selection rules \eqref{eqn:SR-CY1} and \eqref{eqn:SR-CY2}
implies the following:

\begin{prop}
For $ g\geq 1$,
\be
F_{g, instanton} = u_{S,1}^{2g-2} \sum_{d=1}^\infty  q^d N_{g,d} e^{du_R}.
\ee
\end{prop}

\begin{proof}
First recall $F_{g, instanton}$ as a formal power series in $\{t^\alpha_n\}$
is homogeneous of degree $0$,
where the degree of $t^\alpha_n$ is given by \eqref{eqn:SR-CY1}.
By changing to the jet variables $\{u_{\alpha,n}\}$,
it is also homogeneous of degree zero,
where
\begin{align}
\deg u_{S,n} & = n-1, &
\deg u_{R,n} & = n, \\
\deg u_{Q,n} & = n +1, &
\deg u_{P,n} & = n+2.
\end{align}
In Eguchi-Getzler-Xiong \cite{Egu-Get-Xio} i has been shown that
one can rewrite the puncture operator as
\be
L_{-1} = - \frac{\pd}{\pd u_S}.
\ee
So we know that $F_g$ does not depend on $u_S$ when $g\geq 1$,
and it depends only on the two degree zero jet variables $u_{S,1}$ and $u_R$.
To find the explicit expression,
let us first consider the correlators of the form:
\be
\corr{P^{m_0}\tau_1(P)^n Q^{m_1}R^{m_2}S^{m_3}}_{g,n+m_0+\cdots +m_3 ;d}.
\ee
By the selection rules \eqref{eqn:Selection-Rules},
\be
m_0 = m_2 + 2m_3.
\ee
By the string equation we must have $m_0=m_2=m_3=0$ to get
a nonzero correlator of the form:
\ben
&& \corr{\tau_1(P)^n Q^{m_1}}_{g,m_1;d}
= (2g-2+m_1+n-1) \corr{\tau_1(P)^{n-1} Q^{m_1}}_{g,m_1;d} \\
& = & \cdots
= (2g-2+m_1+n-1) \cdots (2g-2+m_1) \cdot \corr{Q^{m_1}}_{g,m_1;d} \\
& = & (2g-2+m_1+n-1) \cdots (2g-2+m_1) \cdot d^{m_1} N_{g,d}.
\een
Their contribution to the free energy is
\ben
&& \sum_{d=1}^\infty q^d \sum_{n,m_1\geq 0} \frac{(t^P_1)^n}{n!}
\frac{(t^Q)^{m_1}}{m_1!} \prod_{j=0}^{n-1} (2g-2+m_1+j) \cdot d^{m_1} N_{g,d} \\
& = & \sum_{d=1}^\infty q^d   N_{g,d}
\sum_{m_1\geq 0} \frac{(t^Q)^{m_1}}{m_1!} d^{m_1}
\sum_{n\geq 0} \frac{(t^P_1)^n}{n!}
 \prod_{j=0}^{n-1} (2g-2+m_1+j)   \\
 & = & \sum_{d=1}^\infty q^d   N_{g,d}
\sum_{m_1\geq 0} \frac{(t^Q)^{m_1}}{m_1!} d^{m_1} \cdot \frac{1}{(1-t^P_1)^{2g-2+m_1}} \\
& = & \frac{1}{(1-t^P_1)^{2g-2}} \cdot \sum_{d=1}^\infty q^d   N_{g,d} e^{t^Q/(1-t^P_1)}.
\een
By \eqref{eqn:u-S1} we have
Recall that
\ben
u_{S,1} = \biggl( 1- \sum_{n=1}^\infty t^P_n \frac{u_S^{n-1}}{(n-1)!}\biggr)^{-1}
= \frac{1}{1-t^P_1},
\een
and by \eqref{eqn:u-R-sol},
\ben
u_R = \biggl(1- \sum_{n=1}^\infty t^P_n
\frac{u_S^{n-1}}{(n-1)!}\biggr)^{-1}
\cdot \sum_{n=0}^\infty t^Q_n \frac{u_S^n}{n!}
= \frac{t^Q}{1-t^P_1}.
\een
\end{proof}

\begin{rmk}
One can use the same method to treat
the degree zero part of $F_g$.
A similar argument yields:
\ben
F_{g,0}
= -200u_{S,1}^{2g-2}\int_{\Mbar_g} \lambda_{g-1}^3
= -200 u_{S,1}^{2g-2} \cdot
\frac{1}{2(2g-2)!}\frac{|B_{2g-2}|}{2g-2} \frac{|B_{2g}|}{2g}.
\een
Here in the second equality we have used
the results predicted by
Mari\~no and Moore \cite{Mar-Moo} and Gopakumar and Vafa \cite{Gop-Vaf}
and proved by Faber and Pandharipande \cite{Fab-Pan}.
\end{rmk}

To summarize,
we have

\begin{thm}
For the quintic,
the free energy in genus $g \geq 1$ on the big phase
space has the following form:
\be
F_1 =
-200 \cdot \frac{1}{24} \log u_{S,1}
+ \frac{25}{12} u_R
+ \sum_{d=1}^\infty N_{1,d} e^{d u_R} q^d
\ee
and for $g\geq 2$,
\be
F_g = -200 \cdot \frac{1}{2(2g-2)!}\frac{|B_{2g-2}|}{2g-2} \frac{|B_{2g}|}{2g} u_{S,1}^{2g-2}
+ u_{S,1}^{2g-2} \sum_{d=1}^\infty  q^d N_{g,d} e^{du_R}.
\ee
\end{thm}

As in the genus zero case,
one can apply the loop operators
repeatedly to get the $n$-point functions
in genus $g$.

\subsection{Deformation of the order parameter}

Define
\be
U_P = \sum_{g=0}^\infty \lambda^{2g}
\frac{\pd^2F_g}{\pd t^P\pd t^P},
\ee
and similarly define $U_Q$, $U_R$ and $U_S$.

Write
\be
F_g = -\frac{25}{3}\delta_{g,1} \log u_{S,1}
+ u_{S,1}^{2g-2}f_g(u_R).
\ee
Their first derivatives are:
\be  \label{eqn:Fg-tP}
\frac{\pd F_g}{\pd t^P}
= - \frac{25}{3} \frac{u_{S,2}}{u_{S,1}} \delta_{g,1}
+  (2g-2)u_{S,1}^{2g-3}u_{S,2}f_g(u_R)
+ u_{S,1}^{2g-2}f_g'(u_R)u_{R,1}.
\ee

\be \label{eqn:Fg-tQ}
\frac{\pd F_g}{\pd t^Q} =
u_{S,1}^{2g-2} f_g'(u_R)
\cdot \frac{\pd u_R}{\pd t^Q}
= u_{S,1}^{2g-1}f_g'(u_R),
\ee

\be
\frac{\pd F_g}{\pd t^R} =
\frac{\pd F_g}{\pd t^S} = 0.
\ee
Taking $\frac{\pd}{\pd t^P}$ we get:
\be
\begin{split}
\frac{\pd^2 F_g}{\pd t^P\pd t^P}
& = - \frac{25}{3} \frac{u_{S,3}}{u_{S,1}} \delta_{g,1}
+ \frac{25}{3} \frac{u_{S,2}^2}{u_{S,1}^2} \delta_{g,1} \\
& + (2g-2)(2g-3)u_{S,1}^{2g-4}u_{S,2}^2f_g(u_R) \\
& + 2(2g-2)u_{S,1}^{2g-3}u_{S,2}f_g'(u_R)u_{R,1} \\
& + (2g-2) u_{S,1}^{2g-3}u_{S,3} f_g(u_R) \\
& + u_{S,1}^{2g-2}f_g'(u_R)u_{R,2}
+ u_{S,1}^{2g-2}f_g''(u_R)u_{R,1}^2.
\end{split}\ee
\ben
\frac{\pd^2 F_g}{\pd t^P\pd t^Q}
& = & \frac{\pd}{\pd t^P} \biggl(
u_{S,1}^{2g-2} f_g'(u_R) \frac{\pd u_R}{\pd t^Q} \biggr)
= \frac{\pd}{\pd t^P} \biggl(
u_{S,1}^{2g-1} f_g'(u_R)   \biggr) \\
& = & (2g-1) u_Su_{S,1}^{2g-2}u_{S,2}
f'_g(u_R) + u_{S,1}^{2g-1}u_{R,1} f_g''(u_R).
\een
\ben
&&\frac{\pd^ F_g}{\pd t^P\pd t^R} =
\frac{\pd^2 F_g}{\pd t^P\pd t^S} = 0.
\een
So we get:
\be \label{eqn:UP}
\begin{split}
U_P & = u_P
- \frac{25\lambda^2}{3}(\frac{u_{S,3}}{u_{S,1}}
- \frac{u_{S,2}^2}{u_{S,1}^2}) \\
& + \sum_{g=1}^\infty
\lambda^{2g} \biggl(u_{S,1}^{2g-2}f_g''(u_R) u_{R,1}^2
+ u_{S,1}^{2g-2}f_g'(u_R) u_{R,2} \\
& + (4g-4) u_{S,1}^{2g-3}u_{S,2}u_{R,1} f_g'(u_R)
+ (2g-2) u_{S,1}^{2g-3}u_{S,3} f_g(u_R) \\
& +  (2g-2)(2g-3) u_{S,1}^{2g-4}u_{S,2}^2 f_g(u_R) \biggr).
\end{split}
\ee
\be \label{eqn:UQ}
U_Q = u_Q
+ \sum_{g=1}^\infty \lambda^{2g}
\biggl( (2g-1) u_{S,1}^{2g-2}u_{S,2}
f'_g(u_R) + u_{S,1}^{2g-1}u_{R,1} f_g''(u_R)
\biggr).
\ee
\begin{align} \label{eqn:UR-US}
U_R & = u_R, & U_S & = u_S.
\end{align}

\subsection{Deformation of the flow equations}
The flow equations for $u_S$ in genus zero are:
\begin{align*}
\frac{\pd u_S}{\pd t^P_n} & = \biggl( \frac{u_S^{n+1}}{(n+1)!}\biggr)', &
\frac{\pd u_S}{\pd t^Q_n} & = 0, &
\frac{\pd u_S}{\pd t^R_n} & = 0, &
\frac{\pd u_S}{\pd t^S_n} & = 0,
\end{align*}
they simply become
\begin{align*}
\frac{\pd U_S}{\pd t^P_n} & = \biggl( \frac{U_S^{n+1}}{(n+1)!}\biggr)', &
\frac{\pd U_S}{\pd t^Q_n} & = 0, &
\frac{\pd U_S}{\pd t^R_n} & = 0, &
\frac{\pd U_S}{\pd t^S_n} & = 0,
\end{align*}
The flow equations for $u_R$ in genus zero are:
\begin{align*}
\frac{\pd u_R}{\pd t^P_n} & = \biggl( \frac{u_Ru_S^{n}}{(n)!}\biggr)', &
\frac{\pd u_R}{\pd t^Q_n} & = \biggl( \frac{u_S^{n+1}}{(n+1)!}\biggr)', &
\frac{\pd u_R}{\pd t^R_n} & = 0, &
\frac{\pd u_R}{\pd t^S_n} & = 0,
\end{align*}
they simply become
\begin{align*}
\frac{\pd U_R}{\pd t^P_n} & = \biggl( \frac{U_RU_S^{n}}{n!}\biggr)', &
\frac{\pd U_R}{\pd t^Q_n} & = \biggl( \frac{U_S^{n+1}}{(n+1)!}\biggr)', &
\frac{\pd U_R}{\pd t^R_n} & = 0, &
\frac{\pd U_R}{\pd t^S_n} & = 0,
\end{align*}
So we will focus on the flow equations
of $U_P$ and $U_Q$.
Since they involve only $u_S$ and $u_R$,
we have
\ben
\frac{\pd U_Q}{\pd t^S_n} = \frac{\pd u_Q}{\pd t^S_n}
= 0,
\een
\ben
\frac{\pd U_Q}{\pd t^R_n} = \frac{\pd u_Q}{\pd t^R_n}
= \biggl(\frac{u_S^{n+1}}{(n+1)!} \biggr)'
= \biggl(\frac{U_S^{n+1}}{(n+1)!} \biggr)',
\een
\ben
\frac{\pd U_P}{\pd t^S_n} = \frac{\pd u_P}{\pd t^S_n}
\biggl(\frac{u_S^{n+1}}{(n+1)!} \biggr)'
= \biggl(\frac{U_S^{n+1}}{(n+1)!} \biggr)',
\een
\ben
\frac{\pd U_P}{\pd t^R_n} = \frac{\pd u_P}{\pd t^R_n}
= \biggl(u_R\frac{u_S^{n}}{n!} \biggr)'
= \biggl(\frac{U_RU_S^{n}}{n!} \biggr)'.
\een
So we will focus on the equations
for $\frac{\pd U_P}{\pd t^P_n}$, $\frac{\pd U_P}{\pd t^Q_n}$,
$\frac{\pd U_Q}{\pd t^P_n}$ and $\frac{\pd U_Q}{\pd t^Q_n}$.
For $n=0$,
we need to consider the deformations
of
\begin{align} \label{eqn:uP-t}
\frac{\pd u_P}{\pd t^P} & = u_P', &
\frac{\pd u_P}{\pd t^Q} & = u_Q',
\end{align}
and
\begin{align} \label{eqn:uQ-t}
\frac{\pd u_Q}{\pd t^P} & = u_Q', &
\frac{\pd u_Q}{\pd t^Q} & = (f_0''(u_R))'.
\end{align}
The first three equations are linear,
it is clear that they get deformed to:
\begin{align*}
\frac{\pd U_P}{\pd t^P} & = U_P', &
\frac{\pd U_P}{\pd t^Q} & = U_Q', &
\frac{\pd U_Q}{\pd t^P} & = U_Q',
\end{align*}
respectively.
Now let us examine the deformation of the
second equation in \eqref{eqn:uQ-t}.
By \eqref{eqn:Fg-tQ} we get:
\ben
\frac{\pd^2 F_g}{\pd t^Q\pd t^Q}
=  \frac{\pd}{\pd t^Q} (u_{S,1}^{2g-1}f_g'(u_R))
= u_{S,1}^{2g}f_g''(u_R)
= U_{S,1}^{2g}f_g''(U_R).
\een
It follows that
\ben
\frac{\pd U_Q}{\pd t^Q}
& = & \frac{\pd}{\pd t^P} \biggl(
\frac{\pd^2 F_0}{\pd t^Q\pd t^Q}
+ \lambda^2 \frac{\pd^2F_1}{\pd t^Q\pd t^Q} + \cdots
\biggr) \\
& = & (f_0''(U_R) + \lambda^2 U_{S,1}^2f_1''(U_R)
+ \cdots)',
\een
and so we get:
\be
\frac{\pd U_Q}{\pd t^Q}
= \biggl( \sum_{g=0}^\infty U_{S,1}^{2g} f_g''(U_R)  \biggr)'.
\ee

For $n=1$ we need to consider the deformation of the following equations:
\begin{align} \label{eqn:uP-t1}
\frac{\pd u_P}{\pd t^P_{1}} & = (u_Pu_S+u_Qu_R)',  &
\frac{\pd u_P}{\pd t^Q_{1}} & = (u_Su_Q+ f_0'(u_R))',
\end{align}

\begin{align} \label{eqn:uQ-t1}
\frac{\pd u_Q}{\pd t^P_{1}} & =
(u_Su_Q + u_Rf_0''(u_R)-f_0'(u_R))', &
\frac{\pd u_Q}{\pd t^Q_{1}} &
= (u_Sf_0''(u_R))'.
\end{align}
Let us begin with the second equation in \eqref{eqn:uQ-t1}.
Its treatment is similar to the case of $\frac{\pd u_Q}{\pd t^Q}$.
First note:
\ben
\frac{\pd^2 F_g}{\pd t^Q\pd t^Q_1}
& = & \frac{\pd}{\pd t^Q_1} (u_{S,1}^{2g-1} f_g'(u_R))
= u_{S,1}^{2g-1}f_g''(u_R) \cdot \frac{\pd u_R}{\pd t^Q_1} \\
& = & u_{S,1}^{2g-1} f_g''(u_R) \cdot u_S u_{S,1}
= U_SU_{1,1}^{2g} f_g''(U_R).
\een
Then from
\ben
\frac{\pd U_Q}{\pd t^Q_1}
& = & \frac{\pd}{\pd t^P} \biggl(
\sum_{g=0}^\infty \lambda^{2g} \frac{\pd^2F_g}{\pd t^Q\pd t^Q_1}
\biggr),
\een
we get
\be
\frac{\pd U_Q}{\pd t^Q_{1}}
= \biggl(U_S \sum_{g=0}^\infty
\lambda^{2g} U_{S,1}^{2g} f_g''(u_R) \biggr)'.
\ee
In the same fashion we see that for general $n$,
\be
\frac{\pd U_Q}{\pd t^Q_{n}}
= \biggl(\frac{U_S^n}{n!}
\sum_{g=0}^\infty \lambda^{2g} U_{S,1}^{2g} f_g''(u_R) \biggr)'.
\ee

Let us now examine the deformation of the second equation
in \eqref{eqn:uP-t1}.
\ben
&& \frac{\pd^2 F_g}{\pd t^P\pd t^Q_1}
= \frac{\pd}{\pd t^P} \biggl(
u_{S,1}^{2g-2} f_g'(u_R) \frac{\pd u_R}{\pd t^Q_1} \biggr)
= \frac{\pd}{\pd t^P} \biggl(
u_S u_{S,1}^{2g-1} f_g'(u_R)   \biggr) \\
& = & u_{S,1}^{2g} f_g'(u_R) + (2g-1) u_Su_{S,1}^{2g-2}u_{S,2}
f'_g(u_R) + u_Su_{S,1}^{2g-1}u_{R,1} f_g''(u_R).
\een
So we have
\ben
&& \frac{\pd U_P}{\pd t^Q_1}
= \lambda^2 \frac{\pd}{\pd t^P} \biggl( \frac{\pd^2 F}{\pd t^P\pd t^Q_1}\biggr) \\
& = & \biggl[ u_Su_Q+ f_0'(u_R) + \sum_{g=1}^\infty
\lambda^{2g} \biggl( u_{S,1}^{2g} f_g'(u_R) \\
&&  + (2g-1) u_Su_{S,1}^{2g-2}u_{S,2}
f'_g(u_R) + u_Su_{S,1}^{2g-1}u_{R,1} f_g''(u_R)
\biggr)
\biggr]' \\
& = & \biggl[ u_S \biggl(u_Q
 +  \sum_{g=1}^\infty
\lambda^{2g}  \biggl( (2g-1) u_{S,1}^{2g-2}u_{S,2}
f'_g(u_R) + u_{S,1}^{2g-1}u_{R,1} f_g''(u_R)
\biggr) \biggr) \\
& + & f_0'(u_R) + \sum_{g=1}^\infty
\lambda^{2g} u_{S,1}^{2g} f_g'(u_R)
\biggr]' \\
\een
Therefore,
by \eqref{eqn:UQ} and \eqref{eqn:UR-US} we get:
\be
\frac{\pd U_P}{\pd t^Q_{1}}
= \biggl(U_S U_Q
+ \sum_{g=0}^\infty
\lambda^{2g} U_{S,1}^{2g} f_g'(U_R)\biggr)'.
\ee
For general $n$,
\be
\frac{\pd U_P}{\pd t^Q_{n}}
= \biggl(\frac{U_S^n}{n!} U_Q
+  \frac{U_S^{n-1}}{(n-1)!}
\sum_{g=0}^\infty \lambda^{2g} U_{S,1}^{2g} f_g'(u_R) \biggr)'.
\ee
Now we consider the deformation of the first equation in \eqref{eqn:uQ-t1}.
We differentiate \eqref{eqn:Fg-tQ} to get:
\ben
\frac{\pd^2 F_g}{\pd t_1^P\pd t^Q}
& = & \frac{\pd}{\pd t_1^P} \biggl(
u_{S,1}^{2g-2} f_g'(u_R) \frac{\pd u_R}{\pd t^Q} \biggr)
= \frac{\pd}{\pd t^P_1} \biggl(
u_{S,1}^{2g-1} f_g'(u_R)   \biggr) \\
& = & (2g-1) u_{S,1}^{2g-2} \frac{\pd u_{S,1}}{\pd t^P_1} f_g'(u_R)
+ u_{S,1}^{2g-1} f_g''(u_R) \frac{\pd u_R}{\pd t^P_1} \\
& = & (2g-1)  u_{S,1}^{2g-2}(u_S u_{S,2} + u_{S,1}^2)
f'_g(u_R) \\
& + & u_{S,1}^{2g-1}(u_Su_{R,1} + u_Ru_{S,1}) f_g''(u_R).
\een
From this we get:
\ben
&& \frac{\pd U_Q}{\pd t^P_1}
= \lambda^2 \frac{\pd}{\pd t^P}
\biggl( \frac{\pd^2 F}{\pd t^P_1\pd t^Q}\biggr) \\
& = & \biggl[u_Su_Q + u_Rf_0''(u_R)-f_0'(u_R) \\
& + & \sum_{g=1}^\infty
\lambda^{2g} \biggl((2g-1)  u_{S,1}^{2g-2}(u_S u_{S,2}
+ u_{S,1}^2)
f'_g(u_R) \\
&& + u_{S,1}^{2g-1}(u_Su_{R,1} + u_Ru_{S,1}) f_g''(u_R)
\biggr)
\biggr]' \\
& = & \biggl[ u_S \biggl(u_Q
 +  \sum_{g=1}^\infty
\lambda^{2g}  \biggl( (2g-1) u_{S,1}^{2g-2}u_{S,2}
f'_g(u_R) + u_{S,1}^{2g-1}u_{R,1} f_g''(u_R)
\biggr) \biggr) \\
& + & u_R \biggl( f_0''(u_R) + \sum_{g=1}^\infty
\lambda^{2g} u_{S,1}^{2g} f_g''(u_R) \biggr) \\
& - & f_0'(u_R) + \sum_{g=1}^\infty (2g-1)
\lambda^{2g} u_{S,1}^{2g} f_g'(u_R)
\biggr]',
\een
so using \eqref{eqn:UQ} and \eqref{eqn:UR-US} we get:
\be
\begin{split}
\frac{\pd U_Q}{\pd t^P_1}
= & \biggl[ U_Q U_S
+ U_R\sum_{g=0}^\infty
\lambda^{2g} U_{S,1}^{2g} f_g''(U_R) \\
+& \sum_{g=0}^\infty (2g-1)
\lambda^{2g} U_{S,1}^{2g} f_g'(U_R)
\biggr]'.
\end{split}
\ee
In general we have:
\be
\begin{split}
\frac{\pd U_Q}{\pd t^P_1}
= & \biggl[ U_Q \frac{U_S^n}{n!}
+ \frac{U_S^{n-1}}{(n-1)!}U_R\sum_{g=0}^\infty
\lambda^{2g} U_{S,1}^{2g} f_g''(U_R) \\
+& \frac{U_S^{n-1}}{(n-1)!}\sum_{g=0}^\infty (2g-1)
\lambda^{2g} U_{S,1}^{2g} f_g'(U_R)
\biggr]'.
\end{split}
\ee

Now we consider the deformation of the first equation in \eqref{eqn:uQ-t1}.
We begin with:
\ben
\frac{\pd U_P}{\pd t^P_1}
& = &    \frac{\pd}{\pd t^P} \frac{\pd }{\pd t^P_1}
\frac{\pd F}{\pd t^P}. \een
Then we take $\frac{\pd}{\pd t^P_1}$ on both sides
of \eqref{eqn:Fg-tP} to get:
\ben
&& \frac{\pd }{\pd t^P_1}
\frac{\pd F}{\pd t^P} \\
& = & u_Pu_S + u_Qu_R
 - \frac{25\lambda^2}{3}
\frac{\pd}{\pd t^P_1}(\frac{u_{S,2}}{u_{S,1}} )  \\
& + & \sum_{g=1}^\infty
\lambda^{2g} \frac{\pd}{\pd t^P_1}
(u_{S,1}^{2g-2} f_g'(u_R) u_{R,1}
+ (2g-2) u_{S,1}^{2g-3}u_{S,2} f_g(u_R))  \\
& = & u_Pu_S + u_Qu_R
 - \frac{25\lambda^2}{3}
 \biggl(\frac{\frac{\pd u_{S,2}}{\pd t^P_1}}{u_{S,1}}
 - \frac{u_{S,2}}{u_{S,1}^2} \frac{\pd u_{S,1}}{\pd t^P_1} \biggr) \\
& + & \sum_{g=1}^\infty
\lambda^{2g}  \biggl(u_{S,1}^{2g-2}
f_g''(u_R) \frac{\pd u_R}{\pd t^P_1} u_{R,1}
+ u_{S,1}^{2g-2}f_g'(u_R) \frac{\pd u_{R,1}}{\pd t^P_1} \\
& + & (2g-2) u_{S,1}^{2g-2} \frac{\pd u_{S,1}}{\pd t^P_1} f_g'(u_R) u_{R,1}
+ (2g-2) u_{S,1}^{2g-3}u_{S,2} f_g'(u_R) \frac{\pd u_R}{\pd t^P_1} \\
& + & (2g-2) u_{S,1}^{2g-3} \frac{\pd u_{S,2}}{\pd t^P_1} f_g(u_R)
+ (2g-2)(2g-3) u_{S,1}^{2g-4} \frac{\pd u_{S,1}}{\pd t^P_1}
u_{S,2} f_g(u_R) \biggr) \\
& = & u_Pu_S + u_Q u_R
 - \frac{25\lambda^2}{3}
 \biggl(\frac{(u_S^2/2!)'''}{u_{S,1}}
 - \frac{u_{S,2}}{u_{S,1}^2} (u_S^2/2!)'' \biggr) \\
& + & \sum_{g=1}^\infty
\lambda^{2g}  \biggl(u_{S,1}^{2g-2} f_g''(u_R) (u_S u_R)'  u_{R,1}
+ u_{S,1}^{2g-2} f_g'(u_R) (u_Su_R)'' \\
& + & (2g-2) u_{S,1}^{2g-2} (u_S^2/2!)'' f_g'(u_R) u_{R,1}
+ (2g-2) u_{S,1}^{2g-3}u_{S,2} f_g'(u_R)
(u_Su_R)' \\
& + & (2g-2) u_{S,1}^{2g-3} (u_S^2/2!)''' f_g(u_R)
+ (2g-2)(2g-3) u_{S,1}^{2g-4} (u_S^2/2!)''
u_{S,2} f_g(u_R) \biggr) \\
& = & u_Pu_S + u_Q u_R
 - \frac{25\lambda^2}{3}
 \biggl( u_S ( \frac{u_{S,3}}{u_{S,1}}
- \frac{u_{S,2}^2}{u_{S,1}^2} ) + 2u_{S,2} \biggr) \\
& + & \sum_{g=1}^\infty
\lambda^{2g}  \biggl(u_{S,1}^{2g-2} f_g''(u_R) (u_{S,1} u_R+u_Su_{R,1})  u_{R,1} \\
& + & u_{S,1}^{2g-2}  f_g'(u_R)
(u_{S,2}u_R+2u_{S,1}u_{R,1} + u_Su_{R,2})  \\
& + & (2g-2) u_{S,1}^{2g-3} (u_Su_{S,2}+u_{S,1}^2)
f_g'(u_R) u_{R,1} \\
& + & (2g-2) u_{S,1}^{2g-3}u_{S,2} f_g'(u_R)
(u_{S,1}u_R+u_Su_{R,1}) \\
& + & (2g-2) u_{S,1}^{2g-3} (u_Su_{S,3}+3u_{S,1}u_{S,2}) f_g(u_R)  \\
& + & (2g-2)(2g-3) u_{S,1}^{2g-4} (u_Su_{S,2} + u_{S,1}^2)
u_{S,2} f_g(u_R) \biggr).
\een
We rewrite the right-hand side of the last equality as follows:
\ben
\frac{\pd }{\pd t^P_1} \frac{\pd F}{\pd t^P}
& = & u_S \biggl( u_P
- \frac{25\lambda^2}{3} ( \frac{u_{S,3}}{u_{S,1}}
- \frac{u_{S,2}^2}{u_{S,1}^2} ) \\
& + & \sum_{g=1}^\infty
\lambda^{2g}  \biggl(u_{S,1}^{2g-2} f_g''(u_R) u_{R,1}^2
+ u_{S,1}^{2g-2}  f_g'(u_R) u_{R,2})  \\
& + & (4g-4) u_{S,1}^{2g-3}u_{S,2} f_g'(u_R) u_{R,1}
+ (2g-2) u_{S,1}^{2g-3} u_{S,3} f_g(u_R)  \\
& + & (2g-2)(2g-3) u_{S,1}^{2g-4} u_{S,2}^2 f_g(u_R) \biggr)\\
& + & u_R \biggl( u_Q
+ \sum_{g=1}^\infty
\lambda^{2g}  \biggl(u_{S,1}^{2g-1} f_g''(u_R) u_{R,1}
+ (2g-1) u_{S,1}^{2g-2}  f_g'(u_R) u_{S,2} \biggr)  \biggr) \\
&  - & \frac{50\lambda^2}{3} u_{S,2}
+  \sum_{g=1}^\infty \lambda^{2g}
\biggl(2gu_{S,1}^{2g-1}f'_g(u_R)u_{R,1}
+ 2g(2g-2)u_{S,1}^{2g-2}u_{S,2}f_g(u_R) \biggr).
\een
Now we use \eqref{eqn:UP}-\eqref{eqn:UR-US} to get
\be
\begin{split}
& \frac{\pd U_P}{\pd t^P_1}
=  U_PU_S + U_QU_R
-\frac{50\lambda^2}{3} U_{S,2} \\
+ & \sum_{g=1}^\infty \lambda^{2g}
\biggl(2gU_{S,1}^{2g-1}f'_g(U_R)U_{R,1}
+ 2g(2g-2)U_{S,1}^{2g-2}U_{S,2}f_g(U_R) \biggr).
\end{split}
\ee
In general,
we have:
\ben
&& \frac{\pd }{\pd t^P_n}
\frac{\pd F}{\pd t^P} \\
& = & u_P \frac{u_S^n}{n!} + u_Qu_R \frac{u_S^{n-1}}{(n-1)!}
+ \frac{u_S^{n-2}}{(n-2)!}
(u^Rf_0'(u^Q)-2f_0(u_R))
 - \frac{25\lambda^2}{3}
\frac{\pd}{\pd t^P_n}(\frac{u_{S,2}}{u_{S,1}} )  \\
& + & \sum_{g=1}^\infty
\lambda^{2g} \frac{\pd}{\pd t^P_n}
(u_{S,1}^{2g-2} f_g'(u_R) u_{R,1}
+ (2g-2) u_{S,1}^{2g-3}u_{S,2} f_g(u_R))  \\
& = & u_P \frac{u_S^n}{n!} + u_Qu_R \frac{u_S^{n-1}}{(n-1)!}
+ \frac{u_S^{n-2}}{(n-2)!}
(u^Rf_0'(u^Q)-2f_0(u_R)) \\
& - & \frac{25\lambda^2}{3}
 \biggl(\frac{\frac{\pd u_{S,2}}{\pd t^P_n}}{u_{S,1}}
 - \frac{u_{S,2}}{u_{S,1}^2} \frac{\pd u_{S,1}}{\pd t^P_n} \biggr) \\
& + & \sum_{g=1}^\infty
\lambda^{2g}  \biggl(u_{S,1}^{2g-2}
f_g''(u_R) \frac{\pd u_R}{\pd t^P_n} u_{R,1}
+ u_{S,1}^{2g-2}f_g'(u_R) \frac{\pd u_{R,1}}{\pd t^P_n} \\
& + & (2g-2) u_{S,1}^{2g-3} \frac{\pd u_{S,1}}{\pd t^P_n} f_g'(u_R) u_{R,1}
+ (2g-2) u_{S,1}^{2g-3}u_{S,2} f_g'(u_R) \frac{\pd u_R}{\pd t^P_n} \\
& + & (2g-2) u_{S,1}^{2g-3} \frac{\pd u_{S,2}}{\pd t^P_n} f_g(u_R)
+ (2g-2)(2g-3) u_{S,1}^{2g-4} \frac{\pd u_{S,1}}{\pd t^P_n}
u_{S,2} f_g(u_R) \biggr) \\
& = &  u_P \frac{u_S^n}{n!} + u_Qu_R \frac{u_S^{n-1}}{(n-1)!}
+ \frac{u_S^{n-2}}{(n-2)!}
(u^Rf_0'(u^Q)-2f_0(u_R))  \\
& - & \frac{25\lambda^2}{3}
 \biggl(\frac{(u_S^{n+1}/(n+1)!)'''}{u_{S,1}}
 - \frac{u_{S,2}}{u_{S,1}^2} (u_S^{n+1}/(n+1)!)'' \biggr) \\
& + & \sum_{g=1}^\infty
\lambda^{2g}  \biggl(u_{S,1}^{2g-2} f_g''(u_R)
(\frac{u_S^n}{n!} u_R)'  u_{R,1}
+ u_{S,1}^{2g-2} f_g'(u_R) (\frac{u_S^n}{n!}u_R)'' \\
& + & (2g-2) u_{S,1}^{2g-3} (u_S^{n+1}/(n+1)!)'' f_g'(u_R) u_{R,1}
+ (2g-2) u_{S,1}^{2g-3}u_{S,2} f_g'(u_R)
(\frac{u_S^n}{n!}u_R)' \\
& + & (2g-2) u_{S,1}^{2g-3} (u_S^{n+1}/(n+1)!)''' f_g(u_R)
+ (2g-2)(2g-3) u_{S,1}^{2g-4} (u_S^{n+1}/(n+1)!)''
u_{S,2} f_g(u_R) \biggr)
\een
\ben
&& \frac{\pd }{\pd t^P_n}
\frac{\pd F}{\pd t^P} \\
& = & u_P \frac{u_S^n}{n!} + u_Qu_R \frac{u_S^{n-1}}{(n-1)!}
+ \frac{u_S^{n-2}}{(n-2)!}
(u^Rf_0'(u^Q)-2f_0(u_R))  \\
& - & \frac{25\lambda^2}{3}
 \biggl(\frac{ u_{S,3}u_S^{n}/n!+3u_{S,2}u_{S,1}u_S^{n-1}/(n-1)!
 +u_{S,1}^3 u_S^{n-2}/(n-2)!}{u_{S,1}} \\
& - & \frac{u_{S,2}}{u_{S,1}^2}
 (u_{S,2}\frac{u_S^n}{n!}+u_{S,1}^2\frac{u_S^{n-1}}{(n-1)!}) \biggr) \\
& + & \sum_{g=1}^\infty
\lambda^{2g}  \biggl(u_{S,1}^{2g-2} f_g''(u_R)
(\frac{u_S^n}{n!} u_{R,1}+\frac{u_S^{n-1}}{(n-1)!}u_{S,1}u_R)  u_{R,1} \\
& + & u_{S,1}^{2g-2} f_g'(u_R)
(\frac{u_S^n}{n!} u_{R,1} + 2 \frac{u_S^{n-1}}{(n-1)!}u_{S,1}u_{R,1}
+ \frac{u_S^{n-1}}{(n-1)!} u_{S,2}u_R
+ \frac{u_S^{n-2}}{(n-2)!} u_{S,1}^2u_R) \\
& + & (2g-2) u_{S,1}^{2g-3} (u_{S,2}\frac{u_S^n}{n!}+u_{S,1}^2\frac{u_S^{n-1}}{(n-1)!})
f_g'(u_R) u_{R,1} \\
& + & (2g-2) u_{S,1}^{2g-3}u_{S,2} f_g'(u_R)
(\frac{u_S^n}{n!}u_{R_1} + \frac{u_S^{n-1}}{(n-1)!}u_{S,1}u_R) \\
& + & (2g-2) u_{S,1}^{2g-3}
(u_{S,3}\frac{u_S^n}{n!}
+3u_{S,2}u_{S,1} \frac{u_S^{n-1}}{(n-1)!}
+u_{S,1}^3 \frac{u_S^{n-2}}{(n-2)!}) f_g(u_R) \\
& + & (2g-2)(2g-3) u_{S,1}^{2g-4} (u_{S,2}\frac{u_S^n}{n!}+u_{S,1}^2\frac{u_S^{n-1}}{(n-1)!})
u_{S,2} f_g(u_R) \biggr)
\een

\ben
&& \frac{\pd }{\pd t^P_n}
\frac{\pd F}{\pd t^P} \\
& = & \frac{u_S^n}{n!}
\biggl[u_P - \frac{25\lambda^2}{3}
\biggl( \frac{u_{S,3}}{u_{S,1}} - \frac{u_{S,2}^2}{u_{S,1}^2} \biggr)  \\
& + & \sum_{g=1}^\infty
\lambda^{2g}  \biggl(u_{S,1}^{2g-2} f_g''(u_R) u_{R,1}^2
+ u_{S,1}^{2g-2} f_g'(u_R) u_{R,1}
+ 2(2g-2) u_{S,1}^{2g-3} u_{S,2} f_g'(u_R) u_{R,1} \\
& + & (2g-2) u_{S,1}^{2g-3} u_{S,3}  f_g(u_R)
+  (2g-2)(2g-3) u_{S,1}^{2g-4} u_{S,2}^2  f_g(u_R) \biggr) \biggr]\\
& + & \frac{u_S^{n-1}}{(n-1)!}u_R  \biggr[ u_Q
+ \sum_{g=1}^\infty
\lambda^{2g}  \biggl(u_{S,1}^{2g-2} f_g''(u_R) u_{S,1}u_R u_{R,1}
+ (2g-1) u_{S,1}^{2g-2} f_g'(u_R) u_{S,2}  \biggr) \biggr] \\
& + & \frac{u_S^{n-2}}{(n-2)!} \biggl[
(u_Rf_0'(u_R)-2f_0(u_R))  +  \sum_{g=1}^\infty \lambda^{2g}
u_{S,1}^{2g} (u_R f_g'(u_R) +(2g-2) f_g(u_R) )
\biggr]  \\
& - & \frac{25\lambda^2}{3}
 \biggl(2u_{S,2} \frac{u_S^{n-1}}{(n-1)!}
+u_{S,1}^2 \frac{u_S^{n-2}}{(n-2)!}   \biggr) \\
& + & \frac{u_S^{n-1}}{(n-1)!}
\sum_{g=1}^\infty
\lambda^{2g} \biggl( 2g u_{S,1}^{2g-1} f_g'(u_R)   u_{R,1}
+ 2g (2g-2) u_{S,1}^{2g-2} u_{S,2} f_g(u_R)
\biggr).
\een
To summarize,
we have proved the following:

\begin{thm}
For the quintic,
the deformed order parameters
satisfy the following system of integrablke hierarchy:
\begin{align*}
\frac{\pd U_S}{\pd t^P_n} & = \biggl( \frac{U_S^{n+1}}{(n+1)!}\biggr)', &
\frac{\pd U_S}{\pd t^Q_n} & = 0, &
\frac{\pd U_S}{\pd t^R_n} & = 0, &
\frac{\pd U_S}{\pd t^S_n} & = 0,
\end{align*}
\begin{align*}
\frac{\pd U_R}{\pd t^P_n} & = \biggl( \frac{U_RU_S^{n}}{n!}\biggr)', &
\frac{\pd U_R}{\pd t^Q_n} & = \biggl( \frac{U_S^{n+1}}{(n+1)!}\biggr)', &
\frac{\pd U_R}{\pd t^R_n} & = 0, &
\frac{\pd U_R}{\pd t^S_n} & = 0,
\end{align*}
\begin{align*}
\frac{\pd U_Q}{\pd t^S_n} & = 0, &
\frac{\pd U_Q}{\pd t^R_n}
& = \biggl(\frac{U_S^{n+1}}{(n+1)!} \biggr)'.
\end{align*}
\be
\frac{\pd U_Q}{\pd t^Q_{n}}
= \biggl(\frac{U_S^n}{n!}
\sum_{g=0}^\infty \lambda^{2g} U_{S,1}^{2g} f_g''(u_R) \biggr)'.
\ee
\be
\begin{split}
\frac{\pd U_Q}{\pd t^P_1}
= & \biggl[ U_Q \frac{U_S^n}{n!}
+ \frac{U_S^{n-1}}{(n-1)!}U_R\sum_{g=0}^\infty
\lambda^{2g} U_{S,1}^{2g} f_g''(U_R) \\
+& \frac{U_S^{n-1}}{(n-1)!}\sum_{g=0}^\infty (2g-1)
\lambda^{2g} U_{S,1}^{2g} f_g'(U_R)
\biggr]'.
\end{split}
\ee

\begin{align}
\frac{\pd U_P}{\pd t^S_n}
& = \biggl(\frac{U_S^{n+1}}{(n+1)!} \biggr)', &
\frac{\pd U_P}{\pd t^R_n}
& = \biggl(\frac{U_RU_S^{n}}{n!} \biggr)'.
\end{align}
\be
\frac{\pd U_P}{\pd t^Q_{n}}
= \biggl(\frac{U_S^n}{n!} U_Q
+  \frac{U_S^{n-1}}{(n-1)!}
\sum_{g=0}^\infty \lambda^{2g} U_{S,1}^{2g} f_g'(u_R) \biggr)'.
\ee

\ben
&& \frac{\pd U_P}{\pd t^P_n}
= \biggl[\frac{U_S^n}{n!} U_P +  \frac{u_S^{n-1}}{(n-1)!}U_R U_Q  \\
& + & \frac{U_S^{n-2}}{(n-2)!} \biggl[
(U_Rf_0'(U_R)-2f_0(U_R))  +  \sum_{g=1}^\infty \lambda^{2g}
u_{S,1}^{2g} (U_R f_g'(U_R) \\
& + & (2g-2) f_g(U_R) ) \biggr]  - \frac{25\lambda^2}{3}
 \biggl(2U_{S,2} \frac{U_S^{n-1}}{(n-1)!}
+U_{S,1}^2 \frac{U_S^{n-2}}{(n-2)!}   \biggr) \\
& + & \frac{U_S^{n-1}}{(n-1)!}
\sum_{g=1}^\infty
\lambda^{2g} \biggl( 2g U_{S,1}^{2g-1} f_g'(U_R)   U_{R,1}
+ 2g (2g-2) U_{S,1}^{2g-2} u_{S,2} f_g(U_R)
\biggr) \biggr]'.
\een
\end{thm}

\section{Emergent Geometry of Gromov-Witten Theory}

\label{sec:Frobenius}

In this Section we return to the general discussions for GW theory
of all symplectic manifolds or Gromov-Witten type theories.
As pointed out by Dubrovin \cite{Dub-Int},
two types of integrable systems are hidden in such theories.
The first one is the system of WDVV equations \eqref{eqn:WDVV},
and the second one is the integrable hierarchy \eqref{eqn:IH}
and its deformation by higher genus contributions.
We will recall the geometric structures behind such integrable systems.
We will again take an emergent point of view,
i.e.,
we will try to understand how geometric structures emerge
in the GW-type theories.

\subsection{WDVV equations, Frobenius manifolds and Higgs systems}

Dubrovin \cite{Dub-Int} introduced the notion of a Frobenius
manifold to give a geometric reformulation
of the system of WDVV equations \eqref{eqn:WDVV}.
Let us recall how it naturally emerge in the setting of GW theory.
Let $\cM$ denote the small phase space with linear coordinates
$\{t^\alpha\}$,
with $t^0$ corresponding to $\cO_0 = 1 \in H^0(M)$.
The Poincar\'e paring gives $\cM$ a flat metric
\be
\eta_{\alpha\beta} =\eta(\frac{\pd}{\pd t^\alpha}, \frac{\pd}{\pd t^\beta})
= \int_M \cO_\alpha \cup \cO_\beta,
\ee
Since this metric is given by a constant matrix,
it is flat.
In particular,
its Levi-Civita connection is given by:
\be
\nabla_{\frac{\pd}{\pd t^\alpha}}(s^\beta \frac{\pd}{\pd t^\beta})
= \frac{\pd s^\beta}{\pd t^\alpha} \frac{\pd}{\pd t^\beta}.
\ee
The quantum multiplications defined in \S \ref{sec:Quant-Cohom} define a Higgs field
$A \in \Omega^1(M, \End(TM))$:
\be
A(\frac{\pd}{\pd t^\alpha}): TM \to TM, \quad
A(\frac{\pd}{\pd t^\alpha})\frac{\pd}{\pd t^\beta}:
=\frac{\pd}{\pd t^\alpha} \circ_\bt \frac{\pd}{\pd t^\beta}.
\ee
Dubrovin \cite{Dub-Int} showed that the system of the WDVV equations
is equivalent to the flatness of the new connection
\be
\tilde{\nabla}_X Y = \nabla_X Y + z \cdot X \circ Y.
\ee
Furthermore,
\be
e:=\frac{\pd}{\pd t^0}
\ee
is a flat section with respect to $\nabla$.
Such structures  make $\cM$ a Frobenius manifold.
This gives a geometric reformulation of some of the results in \S \ref{sec:TRR}
and \S \ref{sec:Quant-Cohom}.

\subsection{The deformed flat coordinates}
\label{sec:Def-flat-coor}

Let
\be \label{eqn:x-alpha}
x^\alpha =x^\alpha(t, z) = t^\alpha + \sum_{n=1}^\infty z^nv_n^\alpha(t)
\ee
be flat coordinates for the Dubrovin connection $\tilde{\nabla}$.
They are specified by the condition of
vanishing of the covariant hessian
\be
\tilde{\nabla}_{\frac{\pd}{\pd t^\alpha}}\tilde{\nabla}_{\frac{\pd}{\pd t^\beta}}
(x^\gamma\frac{\pd}{\pd t^\gamma}) = 0,
\ee
or, equivalently, by the system
\be
\frac{\pd}{\pd t^\alpha}\frac{\pd}{\pd t^\beta} x^\gamma= z c_{\alpha\beta}^\lambda
\frac{\pd}{\pd t^\lambda}x^\gamma.
\ee
After plugging in \eqref{eqn:x-alpha},
we get the following sequences of equations:
\ben
&& \frac{\pd^2 v^\gamma_1}{\pd t^\alpha\pd t^\beta}
= \frac{\pd^3F_0}{\pd t^\alpha\pd t^\beta \pd t^\mu} \eta^{\gamma\mu}, \\
&& \frac{\pd^2 v^\gamma_n}{\pd t^\alpha\pd t^\beta}
= \frac{\pd^3F_0}{\pd t^\alpha\pd t^\beta \pd t^\mu} \eta^{\lambda\mu}
\frac{\pd v_{n-1}^\gamma}{\pd t^\lambda}, \quad n \geq 2.
\een
One can check that by the TRR in genus zero
\be
v_1^\gamma = \frac{\pd F_0}{\pd t^\mu_{n-1}} \eta^{\gamma\mu}.
\ee
See Givental \cite{Giv}.
We recall that the WDVV equations can be derived from the TRR in genus zero.
So we have seen that both the WDVV equations and their solutions
are consequences of the TRR in genus zero.

\subsection{The geometric structure behind the integrable hierarchy \eqref{eqn:IH}}

As pointed out by Witten \cite[\S 3d]{Wit},
the system of equations  \eqref{eqn:IH} is a system of Hamiltonian equations
with the following Poisson bracket:
\be
\{u_\alpha(x), u_\beta(y)\} = \eta_{\alpha\beta} \delta'(x-y),
\ee
with the Hamiltonians the genus zero one-point functions which just appear
in last subsection:
\be \label{eqn:}
H_{n, \alpha} = \int dx R_{\alpha,n}(u^0, \dots, u^r).
\ee
As remarked by Witten \cite{Wit},
the KdV hierarchy is bi-Hamiltonian,
i.e., the KdV flows
are Hamiltonian with respect to two Poisson structures.
The second Poisson structure is related to the Virasoro  algebra,
the recursion for the Hamiltonians of the KdV hierarchy
can be constructed using this bi-Hamiltonian structure.
He remarked that: ``We do not know of any evidence for
a second symplectic structure playing a role for general $M$."

As pointed out by \cite{Dub-Int},
the system of equations  \eqref{eqn:IH} is a system of Hamiltonian equations
of hydrodymic type studied in Dubrovin-Novikov \cite{Dub-Nov}.
hydrodynamic type.
Let $\cM$ be any manifold and $u^0,\dots,u^r$ any local coordinates on $\cM$.
Recall that
\be \label{eqn:PB1}
\{u^a(x), u^b(y)\} =g^{ab}(u(x))\delta'(x-y) + b^{ab}_c(u(x)) \frac{\pd u^c}{\pd x}\delta(x- y)
\ee
determines a Poisson bracket on the loop space $L\cM$ consisting
of smooth functions $u^a: S^1 \to \bC$
(Poisson brackets of hydrodynamic type)
iff the tensor field $g_{ab}du^adu^B$ is a flat metric on $\cM$,
where $(g_{ab}) = (g^{ab})^{-1}$,
and the coefficients $b^{ab}_c(u)$ can be determined from $g_{ab}$ as follows:
\be
b^{ab}_c(u) = -g^{ad}(u)\Gamma^b_{dc}(u),
\ee
where $\Gamma^b_{dc}$ are the Christoffel symbols of the Levi-Civita connection
for the metric $\eta$.
Such a Poisson structure is called a Poisson structure of the hydrodynamic type.
A hamiltonian of hydrodynamic type is a functional of the form:
\be
H=\int h(u(x)) dx,
\ee
where the density $h = h(u)$ does not depend on derivatives.
Any such function $h(u)$ on $\cM$ determines a hamiltonian
system on $L\cM$ of the following form:
\be
\pd_t u^a(x)= \{u^a(x), \int h(u(y)) dy \} = w^a_b(u)\pd_x u^b,
\ee
where $w^a_b(u)= \nabla^a\nabla_a h(u)$.
This is called a hamiltonian system of hydrodynamic type.

\subsection{Comformal Frobenius manifolds and the second Poisson structures}

By the selection rules \eqref{eqn:SR-CY1} and \eqref{eqn:SR-CY2},
the Frobenius manifolds associated
with the GW theory of Calabi-Yau threefolds
are conformal invariant Frobenius manifold.
They correspond to
solutions of the WDVV equations self-similar
with respect to some scaling transformations:
\ben
&& t^\alpha \mapsto c^{1-q_\alpha} t^\alpha,
\alpha=0, 1,\dots,r, \\
&& \eta = \eta_{\alpha\beta} dt^\alpha dt^\beta
\mapsto c^{2-d} \eta, \\
&& c_{\alpha\beta\gamma} \mapsto
c^{q_\alpha+q_\beta+q_\gamma-d} c_{\alpha\beta\gamma},\\
&& F_0(\{t^\alpha\}) \mapsto c^{3-d} F_0(\{t^\alpha\}),
\een
for some $q_0 = 0, q_2, \dots, q_r$, $d$.
A pair $(\alpha,n)$ is resonant if
\be
\frac{d+1}{2}=q_\alpha-n.
\ee
If all pairs are nonresonant, then the Frobenius manifold is said to be nonresonant.

For a Calabi-Yau threefold $M$,
fix a basis $\{\cO_\alpha\}_{\alpha=0,\dots, r}$ of $H^{ev}(M)$,
\be
q_\alpha = 1- \deg \cO_\alpha,
\ee
then $d=3$.
It is easy to see that Calabi-Yau threefolds give rise to
resonant conformal invariant Frobenius manifolds.

As shown in Dubrovin \cite[Theorem 3.2]{Dub-Class},
for a conformal invariant Frobenius manifold,
the formula
\be \label{eqn:PB2}
\begin{split}
& \{t^\alpha(x), t^\beta(y)\}_1 \\
= & [(\frac{d + 1}{2}
- q_\alpha)F^{\alpha\beta}(t(x)) + (\frac{d + 1}{2}
- q_\beta)F^{\alpha\beta}(t(y))]\delta'(x-y)
\end{split}
\ee
determines a Poisson bracket compatible
with the Poisson bracket \eqref{eqn:PB1},
where
\be
F^{\alpha\beta}(t) = \eta^{\alpha\alpha'}
\eta^{\beta\beta'}\pd_{\alpha'}\pd_{\beta'}F(t)
\ee
I.e.,
any linear combination of them again is a Poisson bracket.

\subsection{Compatible Poisson structures and pencils of flat metrics}

By the theory of Dubrovin-Novikov \cite{Dub-Nov} on
Poisson brackets of hydrodynamic type,
the pencil of Poisson structures $\{,\}_1-\lambda\{,\}$ corresponds
to a pencil of flat metrics:
\be
g^{\alpha\beta} - \lambda \eta^{\alpha,\beta},
\ee
where
\be
g^{\alpha\beta}
= i_E (dt^\alpha \circ d t^\beta),
\ee
where $E$ is Euler vector field:
\be
E = \sum (1-q_\alpha) t^\alpha \frac{\pd}{\pd t^\alpha}.
\ee
Given $\lambda$, the subset of $\cM $ where $g^{\alpha\beta} - \lambda \eta^{\alpha,\beta}$
is degenerate is denoted by $\Sigma_\lambda$,
it is called the discriminant locus.

Recall that two metrics $g_1$ and $g_2$ form a flat pencil if:
\begin{itemize}
\item[1)] The linear combination $g_1 - \lambda g_2$
is a flat metric for any $\lambda$.
\item[2)] The Levi-Civita connection of this linear combination
has the form
$$\Gamma^{ij}_k = \Gamma^{ij}_{1k} - \lambda \Gamma^{ij}_{2k}.
$$
\end{itemize}
The flat pencil of metrics is quasihomogeneous of the degree $d$ if there
exists a function $\tau$  such that the vector fields
\ben
E &:=& g_1^{is}\frac{\pd \tau}{\pd t^s} \frac{\pd}{\pd t^i}, \\
e &:=& g_2^{is}\frac{\pd \tau}{\pd t^s} \frac{\pd}{\pd t^i}
\een
satisfy the following properties
\ben
&& [e, E] = e, \\
&& L_Eg_1 = (d - 1)g_1, \\
&& L_eg_1 = g_2, \\
&& L_eg_2 = 0.
\een
If the flat pencil comes from a Frobenius manifold,
then the function $\tau$ can be taken to be
\be
\tau = \eta_{0\alpha}t^\alpha.
\ee
A quasihomogeneous flat pencil is said to be regular if the $(1,1)$-tensor
\be
R^j_i = \frac{d - 1}{2}\delta^j_i + \nabla_{2i}E^j
\ee
does not degenerate.
We have seen that from a Frobenius manifold
one can associate a flat pencil of metrics.
Conversely,
given a regular quasihomogeneous flat pencil
one can reconstruct a Frobenius manifold (cf. e.g. \cite[Theorem 2.1]{Dub-Pencil}).

\subsection{The $\lambda$-periods and the dual flat coordinates}

A function $p = p(t; \lambda)$ is called $\lambda$-period of the Frobenius manifold if
it satisfies
\be
(\nabla_1- \lambda \nabla)dp = 0.
\ee
If one chooses the flat coordinates $p^0(t)$,
$\dots$, $p^r(t)$ of the intersection form,
then
\be
(dp^a,dp^b) = G^{ab}
\ee
are constants.
Furthermore,
we obtain a
reduction of the Poisson pencil to the canonical form
\be
\{p^i, p^j\}_\lambda = G^{ij}\delta'(x - y).
\ee
If $p^i(t)$ are  chosen such that
\be
L_E p^i = \frac{1 - d}{2} p^i,
\ee
then
\be
\tau = \eta_{0\alpha}t^\alpha
= \frac{1 - d}{4} G_{ab}p^ap^b,
\ee
where $ (G_{ab}) = (G^{ab})^{-1}$.
See \cite[Lemma 2.6]{Dub-Almost-Dual}.

\subsection{The dual multiplications}

The metric $(,)$ and $<,>$ are related by
multiplication by Euler vector field:
\be
(E \circ u, v) =< u, v >
\ee
See \cite[Exercise 3.3]{Dub-Frob}.
It follows that:
\be
(E\circ u \circ v, w) = <u \circ v, w >.
\ee
It turns out that $E$ is an invertible element in the quantum
cohomology ring.
\be
(  u \circ v, w) = <E^{-1} \circ u \circ v, w >.
\ee
Dubrovin \cite{Dub-Almost-Dual} defined a new multiplication
$u*v$ by
\be \label{eqn:New-Mult}
u*v:=E^{-1}\circ u \circ v.
\ee
Then one has:
\be
( u \circ v, w) = <u*v, w >.
\ee

\subsection{The dual potential function}

As proved in Dubrovin \cite{Dub-Almost-Dual},
the multiplication \eqref{eqn:New-Mult}
 together with the intersection form ( , ),
the unity = the Euler vector field = E satisfies all the axioms of Frobenius manifold
except for the unity is not parallel.
Furthermore,
let $p^1(t)$, $\dots$, $p^n(t)$ be a system of
local flat coordinates of the intersection form.
Then there exists a function $F_*(p)$ such that
\be
\frac{\pd^3F_*(p)}{\pd p^i\pd p^j\pd p^c}
= G_{ia}G_{jb} \frac{\pd t^\gamma}{\pd p^k} \frac{\pd p^a}{\pd t^\alpha}
\frac{\pd p^b}{\pd t^\beta} c^{\alpha\beta}_\gamma(t).
\ee
This can be rewritten in the following way
\be
d \biggl( \frac{\pd^2 F_*}{\pd p^a\pd p^b} \biggr)  = dp^a \circ dp^b.
\ee
The function $F_*(p)$ satisfies the following associativity equations
\be
\frac{\pd^3F_*(p)}{\pd p^i\pd p^j\pd p^a}G^{ab}
\frac{\pd^3 F_*(p)}{\pd p^b \pd p^k \pd p^l}
= \frac{\pd^3 F_*(p)}{\pd p^l \pd p^j\pd p^a}G^{ab}
\frac{\pd^3 F_*(p)}{\pd p^b\pd p^k \pd p^i}, \ee
$i, i, j, k, l = 1,\dots, n$.
For $d \neq 1$,  $F_*(p)$ satisfies the homogeneity condition
\be
\sum_a p^a\frac{\pd F_*}{\pd p^a} = 2 F_* + \frac{1}{1- d} \sum G_{ab}p^ap^b.
\ee

\subsection{The deformed dual flat coordinates}

Let
\be \label{eqn:p-alpha}
p^\alpha =p^\alpha(t, z) = p^\alpha + \sum_{n=1}^\infty z^nv_n^\alpha(t)
\ee
be dual flat coordinates for the Dubrovin connection $\tilde{\nabla}$.
They are specified by the condition of
vanishing of the covariant hessian
\be
\tilde{\nabla}_{\frac{\pd}{\pd t^\alpha}}\tilde{\nabla}_{\frac{\pd}{\pd t^\beta}}
(x^\gamma\frac{\pd}{\pd t^\gamma}) = 0,
\ee
or, equivalently, by the system
\be
\frac{\pd}{\pd t^\alpha}\frac{\pd}{\pd t^\beta} x^\gamma= z c_{\alpha\beta}^\lambda
\frac{\pd}{\pd t^\lambda}x^\gamma.
\ee
After plugging in \eqref{eqn:x-alpha},
we get the following sequences of equations:
\ben
&& \frac{\pd^2 v^\gamma_1}{\pd t^\alpha\pd t^\beta}
= \frac{\pd^3F_0}{\pd t^\alpha\pd t^\beta \pd t^\mu} \eta^{\gamma\mu}, \\
&& \frac{\pd^2 v^\gamma_n}{\pd t^\alpha\pd t^\beta}
= \frac{\pd^3F_0}{\pd t^\alpha\pd t^\beta \pd t^\mu} \eta^{\lambda\mu}
\frac{\pd v_{n-1}^\gamma}{\pd t^\lambda}, \quad n \geq 2.
\een
One can check that by the TRR in genus zero
\be
v_1^\gamma = \frac{\pd F_0}{\pd t^\mu_{n-1}} \eta^{\gamma\mu}.
\ee
See Givental \cite{Giv}.
We recall that the WDVV equations can be derived from the TRR in genus zero.
So we have seen that both the WDVV equations and their solutions
are consequences of the TRR in genus zero.

In the case of semsimple Frobenius maifolds,
Dubrovin and Zhang \cite{Dub-Zha} showed that the deformed flat coordinates
and the deformed dual flat coordinates are related
by Laplace transform.
We will verify that this also holds for the quintic in \S \ref{sec:Laplace}.

\subsection{The reconstruction program}

In \cite{Dub-Class}
Dubrovin launched a program to reconstruct a complete GW-type  theory
from a Frobenius manifold.
He focused on the case of semisimple Frobenius manifolds.
This program was developed by Dubrovin and Zhang \cite{Dub-Zha}
and many of their subsequent works,
based on the bihamiltonian systems.
For a survey of the current status of this program,
see Dubrovin \cite{Dub-GW}.

\section{Frobenius Manifold Associated with the Quintic}
\label{sec:Quintic-Frob}

In this Section we carry out the explicit computations
in the theory of Frobenius manifolds associated with the quintic.
We will rewrite the genus zero free energy on the small phase space
in terms of the dual flat coordinates.
This will reveal some similarity to the formulas in mirror symmetry computations
in \cite{CDGP}.

\subsection{Induced multiplications on the cotangent bundle
and the intersection form}

Using the quantum multiplications
computed in \S \ref{sec:Quan-Coh},
we get the following induced multiplications
on the cotangent bundle of the small phase space:
\ben
dt^P \circ \begin{pmatrix}
dt^P \\ dt^Q\\ d t^R \\ d t^S \end{pmatrix}
& = & \begin{pmatrix}
0 & 0 & 0 & 0 \\
0 & 0 & 0 & 0 \\
0 & 0 & 0 & 0 \\
1 & 0 & 0 & 0
\end{pmatrix} \cdot
\begin{pmatrix}
dt^P \\ dt^Q\\ d t^R \\ d t^S \end{pmatrix},
\een
\ben
dt^Q \circ \begin{pmatrix}
dt^P \\ dt^Q\\ d t^R \\ d t^S \end{pmatrix}
& = & \begin{pmatrix}
0 & 0 & 0 & 0 \\
0 & 0 & 0 & 0 \\
1 & 0 & 0 & 0 \\
0 & 1 & 0 & 0
\end{pmatrix} \cdot
\begin{pmatrix}
dt^P \\ dt^Q\\ d t^R \\ d t^S \end{pmatrix},
\een

\ben
dt^R \circ \begin{pmatrix}
dt^P \\ dt^Q\\ d t^R \\ d t^S \end{pmatrix}
& = & \begin{pmatrix}
0 & 0 & 0 & 0 \\
1 & 0 & 0 & 0 \\
0 & f_0'''(t^Q) & 0 & 0 \\
0 & 0 & 1 & 0
\end{pmatrix} \cdot
\begin{pmatrix}
dt^P \\ dt^Q\\ d t^R \\ d t^S \end{pmatrix},
\een

\ben
dt^S \circ \begin{pmatrix}
dt^P \\ dt^Q\\ d t^R \\ d t^S \end{pmatrix}
& = & \begin{pmatrix}
1 & 0 & 0 & 0 \\
0 & 1 & 0 & 0 \\
0 & 0 & 1 & 0 \\
0 & 0 & 0 & 1
\end{pmatrix} \cdot
\begin{pmatrix}
dt^P \\ dt^Q\\ d t^R \\ d t^S \end{pmatrix},
\een
Therefore,
after taking inner product with the Euler vector field
\be
E = t^P\frac{\pd}{\pd t^P} - t^R\frac{\pd}{\pd t^R}
- 2t^S \frac{\pd}{\pd t^S},
\ee
we get the intersection form:
\ben
(dt^\alpha, dt^\beta) = \begin{pmatrix}
0 & 0 & 0 & t^P \\
0 & 0 & t^P & 0 \\
0 & t^P & 0 & -t^R \\
t^P & 0 & -t^R & -2 t^S
\end{pmatrix}.
\een
One can check that
\be
(dt^\alpha,d t^\beta)=R\cdot (F^{\alpha\beta})
+ (F^{\alpha\beta}) \cdot R,
\ee
where the  matrix $R$ is
\be
R= \frac{d-1}{2} + \diag(d_\alpha)
=\begin{pmatrix}
2 & 0 & 0 & 0 \\
0 & 1 & 0 & 0 \\
0 & 0 & 0  & 0 \\
0 & 0 & 0 & -1
\end{pmatrix}
\ee
where $d=\dim X = 3$,
\begin{align}
d_P &= 1, & d_Q & = 0, & d_R & = -1, & d_S & = -2,
\end{align}
i.e., $d_\alpha = 1- \deg \alpha$.
A related matrix is
\be
\cV = R - \frac{1}{2}I = \begin{pmatrix}
\frac{3}{2} & 0 & 0 & 0 \\
0 & \frac{1}{2} & 0 & 0 \\
0 & 0 & -\frac{1}{2}  & 0 \\
0 & 0 & 0 & -\frac{3}{2}
\end{pmatrix}
\ee
Note
\be
\cV \cdot (\eta_{\alpha\beta}) = - (\eta_{\alpha\beta}) \cdot \cV,
\ee
it defines a symplectic structure on the small phase space.

Recall that
\be
F^{\alpha\beta} = \eta^{\alpha \alpha'}\eta^{\beta \beta'}
\frac{\pd^2F_0}{\pd t^{\alpha'}\pd t^{\beta'}}
\ee
Because we have
\ben
<dt^\alpha, dt^\beta> = \begin{pmatrix}
0 & 0 & 0 & 1 \\
0 & 0 & 1 & 0 \\
0 & 1 & 0 & 0 \\
1 & 0 & 0 & 0
\end{pmatrix}
\een
and
\be
(F_{\alpha\beta})
= \begin{pmatrix}
t^S & t^R & t^Q & t^P \\
t^R & f_0''(t^Q) & t^P & 0 \\
t^Q & t^P &  0 & 0 \\
t^P & 0 & 0 & 0
\end{pmatrix},
\ee
so explicitly we have:

\be
(F^{\alpha\beta})
= \begin{pmatrix}
0 & 0 & 0 & t^P \\
0 & 0 & t^P & t^Q \\
0 & t^P & f_0''(t^Q) & t^R \\
t^P & t^Q & t^R & t^S
\end{pmatrix}.
\ee

\subsection{Flat pencil of metrics}
Now we have:
\ben
(dt^\alpha, dt^\beta)-\lambda <dt^\alpha, dt^\beta>
= \begin{pmatrix}
0 & 0 & 0 & t^P-\lambda \\
0 & 0 & t^P-\lambda & 0 \\
0 & t^P-\lambda & 0 & -t^R \\
t^P-\lambda & 0 & -t^R & -2 t^S
\end{pmatrix}
\een
So the discriminant locus is given by:
\be
\Sigma_\lambda = \{t^P = \lambda\}.
\ee
The inverse matrix
\ben
g_{\alpha\beta}(\lambda)
= \frac{1}{(t^P-\lambda)^2}\begin{pmatrix}
2t^S & t^R & 0 & t^P-\lambda \\
t^R & 0 & t^P-\lambda & 0 \\
0 & t^P-\lambda & 0 & 0 \\
t^P-\lambda & 0 & 0 & 0
\end{pmatrix}
\een
defines a flat metric on $M\backslash \Sigma_\lambda$.
Using the formula for Christoffel symbols:
\be
\Gamma^\alpha_{\beta\gamma}(\lambda)
= \frac{1}{2} g^{\alpha\rho}(\lambda)
(\pd_\beta g_{\gamma\rho}(\lambda)
+ \pd_\gamma g_{\beta\rho}(\lambda)
- \pd_\rho g_{\beta\gamma}(\lambda)),
\ee
one can find the nonzero components of the Christoffel symbols
of $g(\lambda)$ are:
\ben
&& \Gamma^P_{PP}(\lambda) = -\frac{2}{t^P-\lambda},
\Gamma^Q_{PQ}(\lambda) =
\Gamma^Q_{QP}(\lambda) = -\frac{1}{t^P-\lambda}, \\
&& \Gamma^S_{PP}(\lambda) = \frac{2t^S}{(t^P-\lambda)^2},
\Gamma^S_{PQ}(\lambda) = \Gamma^S_{QP}(\lambda) = \frac{t^R}{(t^P-\lambda)^2}, \\
&& \Gamma^S_{PS}(\lambda) = \Gamma^S_{SP}(\lambda) = \frac{1}{t^P-\lambda},
\Gamma^S_{QR}(\lambda) = \Gamma^S_{RQ}(\lambda) = \frac{1}{t^P-\lambda}.
\een
One can verify that
\be
\Gamma^\alpha_{\beta\gamma}(\lambda)
= -g_{\beta \rho}(\lambda)\cdot
\Gamma^{\rho\alpha}_\gamma,
\ee
where
\begin{align*}
(\Gamma_P^{\alpha\beta}) &= \begin{pmatrix}
0 & 0 & 0 & -1 \\
0 & 0 & 0 & 0 \\
0 & 1 & 0 & 0 \\
2 & 0 & 0 & 0
\end{pmatrix}, &
(\Gamma_Q^{\alpha\beta}) &= \begin{pmatrix}
0 & 0 & 0 & 0 \\
0 & 0 & 0 & -1 \\
0 & 0 & 0 & 0 \\
0 & 1 & 0 & 0
\end{pmatrix}, \\
(\Gamma_R^{\alpha\beta}) &= \begin{pmatrix}
0 & 0 & 0 & 0 \\
0 & 0 & 0 & 0 \\
0 & 0 & 0 & -1 \\
0 & 0 & 0 & 0
\end{pmatrix}, &
(\Gamma_S^{\alpha\beta}) &= \begin{pmatrix}
0 & 0 & 0 & 0 \\
0 & 0 & 0 & 0 \\
0 & 0 & 0 & -1 \\
0 & 0 & 0 & 0
\end{pmatrix}.
\end{align*}
One can also check that:
\be
\Gamma^{\alpha\beta}_\gamma
= c_\gamma^{\alpha \epsilon} R_\epsilon^\beta.
\ee

\subsection{The periods}

We now solve the system of equations:
\be
\pd_\alpha \pd_\beta x
- \Gamma_{\alpha\beta}^\gamma(\lambda) \pd_\gamma x = 0.
\ee
They are explicitly given by:
\ben
&& \pd_P
\begin{pmatrix} \pd_Px \\ \pd_Qx \\
\pd_Rx \\ \pd_S x\end{pmatrix}
= \begin{pmatrix}  -\frac{2}{t^P-\lambda} & 0 &
0 &  \frac{2t^S}{(t^P-\lambda)^2} \\
0 &  -\frac{1}{t^P-\lambda} & 0 &  \frac{t^R}{(t^P-\lambda)^2} \\
0 & 0 & 0 & 0 \\
0 & 0 & 0 &  \frac{1}{t^P-\lambda}
\end{pmatrix}
\cdot \begin{pmatrix} \pd_Px \\ \pd_Qx \\
\pd_Rx \\ \pd_S x\end{pmatrix}
\een

\ben
&& \pd_Q
\begin{pmatrix} \pd_Px \\ \pd_Qx \\
\pd_Rx \\ \pd_S x\end{pmatrix}
= \begin{pmatrix}  0 & -\frac{1}{t^P-\lambda} & 0 &
\frac{t^R}{(t^P-\lambda)^2} \\
0 & 0 & 0 & 0 \\
0 & 0 & 0 & \frac{1}{t^P-\lambda} \\
0 & 0 & 0 & 0
\end{pmatrix}
\cdot \begin{pmatrix} \pd_Px \\ \pd_Qx \\
\pd_Rx \\ \pd_S x\end{pmatrix}
\een

\ben
&& \pd_R
\begin{pmatrix} \pd_Px \\ \pd_Qx \\
\pd_Rx \\ \pd_S x\end{pmatrix}
= \begin{pmatrix}
0 & 0 & 0 & 0 \\
0 & 0 & 0 & \frac{1}{t^P-\lambda} \\
0 & 0 & 0 & 0 \\
0 & 0 & 0 & 0
\end{pmatrix}
\cdot \begin{pmatrix} \pd_Px \\ \pd_Qx \\
\pd_Rx \\ \pd_S x\end{pmatrix}
\een

\ben
&& \pd_S
\begin{pmatrix} \pd_Px \\ \pd_Qx \\
\pd_Rx \\ \pd_S x\end{pmatrix}
= \begin{pmatrix}
0 & 0 & 0 & \frac{1}{t^P-\lambda} \\
0 & 0 & 0 & 0 \\
0 & 0 & 0 & 0 \\
0 & 0 & 0 & 0
\end{pmatrix}
\cdot \begin{pmatrix} \pd_Px \\ \pd_Qx \\
\pd_Rx \\ \pd_S x\end{pmatrix}
\een
The equations for $\pd_Sx$ are
\ben
&& \pd_P(\pd_Sx) = \frac{1}{t^P-\lambda}\pd_Sx, \quad
\pd_Q(\pd_Sx) = \pd_R(\pd_Sx) = \pd_S(\pd_Sx) = 0,
\een
so we get:
\ben
&& \pd_Sx= C_1\cdot (t^P-\lambda).
\een
The equations for $\pd_Rx$ are
\ben
&& \pd_Q(\pd_Rx) = \frac{1}{t^P-\lambda}\pd_Sx
= C_1, \\
&& \pd_P(\pd_Rx) = \pd_R(\pd_Rx) = \pd_S(\pd_Rx) = 0,
\een
so we get:
\be
\pd_Rx= C_1\cdot t^Q +C_2.
\ee
The equations for $\pd_Qx$ are
\ben
\pd_P(\pd_Qx) & = & -\frac{1}{t^P-\lambda}\pd_Qx
+ \frac{t^R}{(t^P-\lambda)^2} \pd_Sx
= -\frac{1}{t^P-\lambda}\pd_Qx
+ C_1 \frac{t^R}{t^P-\lambda}, \\
\pd_Q(\pd_Qx) & = & 0, \\
\pd_R(\pd_Qx) & = & \frac{1}{t^P-\lambda}\pd_Sx =
C_1, \\
\pd_S(\pd_Qx) & = & 0,
\een
so we get:
\be
\pd_Qx= C_1t^R+ \frac{C_3}{t^P-\lambda}.
\ee
The equations for $\pd_Px$ are
\ben
\pd_P(\pd_Px) & = & -\frac{2}{t^P-\lambda}\pd_Px
+ \frac{2t^S}{(t^P-\lambda)^2} \pd_Sx
= -\frac{2}{t^P-\lambda}\pd_Px
+ C_1 \frac{2t^S}{t^P-\lambda}, \\
\pd_Q(\pd_Px) & = & -\frac{1}{t^P-\lambda}\pd_Qx
+ \frac{t^R}{(t^P-\lambda)^2} \pd_Sx
= - \frac{C_3}{(t^P-\lambda)^2}, \\
\pd_R(\pd_Px) & = & 0, \\
\pd_S(\pd_Px) & = & \frac{1}{t^P-\lambda}\pd_Sx =C_1,
\een
so we get:
\be
\pd_Px= C_1t^S- C_3\frac{t^Q}{(t^P-\lambda)^2}
- C_4 \frac{1}{(t^P-\lambda)^2}.
\ee
Therefore,
\be
x= C_0 + C_1((t^P-\lambda)t^S+t^Qt^R)
+ C_2t^R+ C_3 \frac{t^Q}{t^P-\lambda}
+ C_4\frac{1}{t^P-\lambda}.
\ee
From this we get the following

\begin{prop}
For the quintic we can take the following system of dual flat coordinates for the flat pencil:
\be \label{eqn:Dual-Flat}
p^4 = (t^P-\lambda)t^S+t^Qt^R,
p^3 = t^R,
p^2 = \frac{t^Q}{t^P-\lambda},
p^1 = \frac{1}{t^P-\lambda},
\ee
\end{prop}

Note
\be
\deg p^i(\lambda) = - 1
\ee
if we set
\be
\deg t^P = \deg \lambda =1, \deg t^Q = 0, \deg t^R = -1,
\deg t^S = -2.
\ee

\subsection{The free energy in the dual flat coordinates}

From \eqref{eqn:Dual-Flat} we get:
\be
t^P-\lambda = \frac{1}{p^1},
t^Q= \frac{p^2}{p_1},
t^R =p^3,
t^S = p^1p^4-p^2p^3.
\ee
Recall that for the quintic the free energy in genus zero has the following form:
\be
\begin{split}
F_0 & = \frac{1}{2}(t^P)^2t^S + t^Pt^Qt^R
+ \frac{5}{6} (t^Q)^3
+ \sum_{m=1}^\infty N_m e^{mt^Q}q^m \\
& =  \frac{1}{2}(t^P)^2t^S + t^Pt^Qt^R
+  \frac{5}{2} \big(\frac{\omega_1\omega_2}{\omega_0^2} - \frac{\omega_3}{\omega_0}\big).
\end{split}
\ee
Therefore, taking $\lambda =0$ we have:
\ben
F_0 & = & \frac{1}{2}\frac{p^1p^4-p^2p^3}{(p^1)^2}
+ \frac{p^2p^3}{(p^1)^2}
+ \frac{5}{2} \big(\frac{\omega_1\omega_2}{\omega_0^2} - \frac{\omega_3}{\omega_0} \big) \\
& = & \frac{1}{2} \frac{1}{(p^1\omega_0)^2}
((p^1\omega_0)(p^4\omega_0-5p^1\omega_3)
+ \half (p^2\omega_0)(p^3\omega_0+5p^1\omega_2) ),
\een
where in the second equality we have used the equality
\be
p^1\omega_1 = p^2\omega_0
\ee
derived from
\be
t^Q = \frac{p^2}{p^1} = \frac{\omega_1(t)}{\omega_0(t)}.
\ee
Therefore, if we set
\ben
&& \hat{p}^1 = p^1\omega_0, \\
&& \hat{p}^2 = p^2\omega_0 = p^1\omega_1, \\
&& \hat{p}^3 = p^3\omega_0 + 5 p^1\omega_2, \\
&& \hat{p}^4 = p^4\omega_0 - 5 p^1\omega_3,
\een
then
\be
F_0 = \half \frac{\hat{p}^1\hat{p}^4+\hat{p}^2\hat{p}^3}{(\hat{p}^1)^2}.
\ee

\subsection{The metric $G^{ij}$}

A constant symmetric nondegenerate matrix $(G_{ij})$
is defined by
\be
G_{ij} = \frac{\pd p^i}{\pd t^\alpha}
g_{\alpha\beta} \frac{\pd p^j}{\pd t^\beta}.
\ee
From
\ben
(dt^\alpha, dt^\beta) = \begin{pmatrix}
0 & 0 & 0 & t^P \\
0 & 0 & t^P & 0 \\
0 & t^P & 0 & -t^R \\
t^P & 0 & -t^R & -2 t^S
\end{pmatrix}
\een
a computation gives us:
\be
(dp^i,dp^j) = \begin{pmatrix}
0 & 0 & 0 & -1 \\
0 & 0 & 1 & 0 \\
0 & 1 & 0 & 0 \\
-1 & 0 & 0 & 0
\end{pmatrix}
\ee

\subsection{The almost dual quantum multiplication}

The quantum multiplication by the Euler vector field is given by:
\be
E \circ \begin{pmatrix}
\frac{\pd}{\pd t^P} \\ \frac{\pd}{\pd t^Q} \\
\frac{\pd}{\pd t^R} \\ \frac{\pd}{\pd t^S}
\end{pmatrix}
= \begin{pmatrix}
t^P & 0 & -t^R & -2t^S \\ 0 & t^P & 0 & -t^R \\
0 & 0 & t^P & 0 \\ 0 & 0 & 0 & t^P
\end{pmatrix}
\cdot
\begin{pmatrix}
\frac{\pd}{\pd t^P} \\ \frac{\pd}{\pd t^Q} \\
\frac{\pd}{\pd t^R} \\ \frac{\pd}{\pd t^S}
\end{pmatrix}
\ee

\be
\cU = \begin{pmatrix}
t^P & 0 & -t^R & -2t^S \\ 0 & t^P & 0 & -t^R \\
0 & 0 & t^P & 0 \\ 0 & 0 & 0 & t^P
\end{pmatrix}
\ee

The inverse of $E$ can be found in the following fashion:
Let $E^{-1} = A\frac{\pd}{\pd t^P}
+ B \frac{\pd}{\pd t^Q} + C \frac{\pd}{\pd t^R} + D \frac{\pd}{\pd t^S}$.
Recall that
$E= t^P\frac{\pd}{\pd t^P}
- t^R\frac{\pd}{\pd t^R} - 2 t^S\frac{\pd}{\pd t^S}$.
\ben
&& \biggl(t^P\frac{\pd}{\pd t^P}
- t^R\frac{\pd}{\pd t^R} - 2t^S \frac{\pd}{\pd t^S} \biggr)
\cdot \biggl(A\frac{\pd}{\pd t^P}
+ B \frac{\pd}{\pd t^Q} + C \frac{\pd}{\pd t^R} + D \frac{\pd}{\pd t^S}\biggr) \\
& = & At^P\frac{\pd}{\pd t^P}
+ Bt^P \frac{\pd}{\pd t^Q} + C t^P\frac{\pd}{\pd t^R} + Dt^P \frac{\pd}{\pd t^S}
- A t^R\frac{\pd}{\pd t^R}
- B t^R\frac{\pd}{\pd t^S}  - 2 A t^S\frac{\pd}{\pd t^S},
\een
and so $A=\frac{1}{t^P}$, $B=0$, $C=\frac{t^R}{(t^P)^3}$,
$D=\frac{2t^S}{(t^P)^2}$,
i.e.,
\be
E^{-1} = \frac{1}{t^P}\frac{\pd}{\pd t^P}
+ \frac{t^R}{(t^P)^2}\frac{\pd}{\pd t^R}
+ \frac{2t^S}{(t^P)^2} \frac{\pd}{\pd t^S}.
\ee
So we get the following explicit formula for dual multiplications:
\ben
\frac{\pd}{\pd t^P} * \begin{pmatrix}
\frac{\pd}{\pd t^P} \\
\frac{\pd}{\pd t^Q} \\
\frac{\pd}{\pd t^R} \\
\frac{\pd}{\pd t^S} \end{pmatrix}
& = & E^{-1} \circ \begin{pmatrix}
1 & 0 & 0 & 0 \\
0 & 1 & 0 & 0 \\
0 & 0 & 1 & 0 \\
0 & 0 & 0 & 1
\end{pmatrix} \cdot
 \begin{pmatrix}
\frac{\pd}{\pd t^P} \\
\frac{\pd}{\pd t^Q} \\
\frac{\pd}{\pd t^R} \\
\frac{\pd}{\pd t^S} \end{pmatrix}\\
& = & \begin{pmatrix}
1 & 0 & 1 & 2 \\
0 & 1 & 0 & 1 \\
0 & 0 & 1 & 0 \\
0 & 0 & 0 & 1
\end{pmatrix} \cdot
 \begin{pmatrix}
\frac{\pd}{\pd t^P} \\
\frac{\pd}{\pd t^Q} \\
\frac{\pd}{\pd t^R} \\
\frac{\pd}{\pd t^S} \end{pmatrix},
\een

\ben
\frac{\pd}{\pd t^Q} * \begin{pmatrix}
\frac{\pd}{\pd t^P} \\
\frac{\pd}{\pd t^Q} \\
\frac{\pd}{\pd t^R} \\
\frac{\pd}{\pd t^S} \end{pmatrix}
& = &  E^{-1} \circ
\begin{pmatrix}
0 & 1 & 0 & 0 \\
0 & 0 & 5+\sum_{m=1}^\infty K_mm^3e^{mt^Q} & 0 \\
0 & 0 & 0 & 1 \\
0 & 0 & 0 & 0
\end{pmatrix} \cdot
 \begin{pmatrix}
\frac{\pd}{\pd t^P} \\
\frac{\pd}{\pd t^Q} \\
\frac{\pd}{\pd t^R} \\
\frac{\pd}{\pd t^S} \end{pmatrix} \\
& = &
\begin{pmatrix}
0 & 1 & 0 & 1 \\
0 & 0 & 5+\sum_{m=1}^\infty K_mm^3e^{mt^Q} & 0 \\
0 & 0 & 0 & 1 \\
0 & 0 & 0 & 0
\end{pmatrix} \cdot
 \begin{pmatrix}
\frac{\pd}{\pd t^P} \\
\frac{\pd}{\pd t^Q} \\
\frac{\pd}{\pd t^R} \\
\frac{\pd}{\pd t^S} \end{pmatrix}
\een

\ben
\frac{\pd}{\pd t^R} * \begin{pmatrix}
\frac{\pd}{\pd t^P} \\
\frac{\pd}{\pd t^Q} \\
\frac{\pd}{\pd t^R} \\
\frac{\pd}{\pd t^S} \end{pmatrix}
& = & \begin{pmatrix}
0 & 0 & 1 & 0 \\
0 & 0 & 0 & 1 \\
0 & 0 & 0 & 0 \\
0 & 0 & 0 & 0
\end{pmatrix} \cdot
 \begin{pmatrix}
\frac{\pd}{\pd t^P} \\
\frac{\pd}{\pd t^Q} \\
\frac{\pd}{\pd t^R} \\
\frac{\pd}{\pd t^S} \end{pmatrix},
\een

\ben
\frac{\pd}{\pd t^S} * \begin{pmatrix}
\frac{\pd}{\pd t^P} \\
\frac{\pd}{\pd t^Q} \\
\frac{\pd}{\pd t^R} \\
\frac{\pd}{\pd t^S} \end{pmatrix}
& = &   \begin{pmatrix}
0 & 0 & 0 & 1 \\
0 & 0 & 0 & 0 \\
0 & 0 & 0 & 0 \\
0 & 0 & 0 & 0
\end{pmatrix} \cdot
 \begin{pmatrix}
\frac{\pd}{\pd t^P} \\
\frac{\pd}{\pd t^Q} \\
\frac{\pd}{\pd t^R} \\
\frac{\pd}{\pd t^S} \end{pmatrix}.
\een

\subsection{The almost dual potential function}

We use
\be
d \biggl( \frac{\pd^2F_*}{\pd p^a\pd p^b} \biggr)
 = dp_a \circ dp_b
\ee
to compute $F_*$.
First note:
\begin{align*}
dp^1 & = -\frac{dt^P}{(t^P)^2}, &
dp^2 & = -\frac{t^Qdt^P}{(t^P)^2}+ \frac{dt^Q}{t^P}, \\
dp^3 & = dt^R, &
dp^4 & = t^Sdt^P+t^Rdt^Q+t^Qdt^R+t^Pdt^S.
\end{align*}
\begin{align*}
dt^P & = -\frac{dp^1}{(p^1)^2}, &
dt^Q & = -\frac{p^2dp^1}{(p^1)^2}+ \frac{dp^2}{p^1}, \\
dt^R & = dp^3, &
dt^S & = p^4dp^1-p^3dp^2-p^2dp^3+p^1dp^4.
\end{align*}
So we get:
\ben
dp^1 \circ \begin{pmatrix}
dp^1 \\ d p^2 \\ dp^3 \\ dp^4
\end{pmatrix}
& = & \begin{pmatrix}
0 & 0 & 0 & 0 \\ 0 & 0 & 0 & 0 \\
0 & 0 & 0 & 0 \\ - \frac{1}{t^P} & 0 & 0 & 0
\end{pmatrix} \cdot \begin{pmatrix}
d t^P \\ dt^Q \\ dt^R \\ dt^S \end{pmatrix} \\
& = & \begin{pmatrix}
0 & 0 & 0 & 0 \\ 0 & 0 & 0 & 0 \\
0 & 0 & 0 & 0 \\  \frac{1}{p^1} & 0 & 0 & 0
\end{pmatrix} \cdot \begin{pmatrix}
d p^1 \\ dp^2 \\ dp^3 \\ dp^4 \end{pmatrix},
\een
therefore,
\be
\frac{\pd^2F_*}{\pd p^4\pd p^4}
= \frac{\pd^2F_*}{\pd p^4\pd p^3} =
\frac{\pd^2F_*}{\pd p^4\pd p^2} = 0,
\frac{\pd^2F_*}{\pd p^4\pd p^1} = \log (p^1).
\ee
Similarly,
\ben
dp^2 \circ \begin{pmatrix}
dp^1 \\ d p^2 \\ dp^3 \\ dp^4
\end{pmatrix}
& = & \begin{pmatrix}
0 & 0 & 0 & 0 \\ 0 & 0 & 0 & 0 \\
\frac{1}{t^P} & 0 & 0 & 0 \\ 0 & 1 & 0 & 0
\end{pmatrix} \cdot \begin{pmatrix}
d t^P \\ dt^Q \\ dt^R \\ dt^S \end{pmatrix} \\
& = & \begin{pmatrix}
0 & 0 & 0 & 0 \\ 0 & 0 & 0 & 0 \\
-\frac{1}{p^1} & 0 & 0 & 0 \\
-\frac{p^2}{(p^1)^2} & \frac{1}{p^1}  & 0 & 0
\end{pmatrix} \cdot \begin{pmatrix}
d p^1 \\ dp^2 \\ dp^3 \\ dp^4 \end{pmatrix}
\een
so we get:
\be
\frac{\pd^2F_*}{\pd p^3\pd p^4} =
\frac{\pd^2F_*}{\pd p^3\pd p^3} = 0,
\frac{\pd^2F_*}{\pd p^3\pd p^2} = - \log(p^1),
\frac{\pd^2F_*}{\pd p^3\pd p^1} = -\frac{p^2}{p^1}.
\ee
From
\ben
dp^3 \circ \begin{pmatrix}
dp^1 \\ d p^2 \\ dp^3 \\ dp^4
\end{pmatrix}
& = & \begin{pmatrix}
0 & 0 & 0 & 0 \\ \frac{1}{t^P} & 0 & 0 & 0 \\
0  & f_0'''(t^Q) & 0 & 0 \\ t^R & t^Qf'''_0(t^Q) & t^P & 0
\end{pmatrix} \cdot \begin{pmatrix}
d t^P \\ dt^Q \\ dt^R \\ dt^S \end{pmatrix},
\een
we know that
\ben
dp_2 \cdot dp_4 & = & 0,  \\
dp_2 \cdot dp_3 & = & \frac{dt^P}{t^P}
= d \log t^P = - d\log p^1, \\
dp_2 \cdot dp_2 & = & f'''_0(t^Q)dt^Q = d f''(t^Q)
= d f''(\frac{p^2}{p^1}), \\
dp_2 \cdot dp_1 & = & t^Rdt^P+t^Qf'''_0(t^Q) dt^Q+ t^Pdt^R
= d(t^Pt^R+t^Qf_0''(t^Q)-f_0'(t^Q)) \\
& = & d \biggl(\frac{p^3}{p^1} + \frac{p^2}{p^1} f_0''(\frac{p^2}{p^1})
- f'_0(\frac{p^2}{p^1}) \biggr).
\een
so we get:
\be
\begin{split}
\frac{\pd^2F_*}{\pd p^2\pd p^4} & = 0, \\
\frac{\pd^2F_*}{\pd p^2\pd p^3} & = -\log p^1, \\
\frac{\pd^2F_*}{\pd p^2\pd p^2} & = f''(\frac{p^2}{p^1}), \\
\frac{\pd^2F_*}{\pd p^2\pd p^1}
& = - \biggl(\frac{p^3}{p^1} + \frac{p^2}{p^1} f_0''(\frac{p^2}{p^1})
- f'_0(\frac{p^2}{p^1}) \biggr).
\end{split}
\ee

\ben
dp^4 \cdot \begin{pmatrix}
dp^1 \\ d p^2 \\ dp^3 \\ dp^4
\end{pmatrix}
& = & \begin{pmatrix}
-\frac{1}{t^P} & 0 & 0 & 0 \\ 0 & 1 & 0 & 0 \\
t^R  & t^Qf_0'''(t^Q) & t^P & 0 \\
2t^Pt^S+2t^Qt^R & 2t^Pt^R+(t^Q)^2f'''_0(t^Q) & 2t^Pt^Q & (t^P)^2
\end{pmatrix} \cdot \begin{pmatrix}
d t^P \\ dt^Q \\ dt^R \\ dt^S \end{pmatrix},
\een
From this we know that
\ben
dp_1 \cdot dp_4 & = & -\frac{dt^P}{t^P} = - d\log t^P= d\log p^1,  \\
dp_1 \cdot dp_3 & = & - dt^Q =- d \frac{p^2}{p^1}, \\
dp_1 \cdot dp_2 & = & -(t^Rdt^P+t^Qf'''_0(t^Q)dt^Q+t^Pdt^R) \\
& = & -d (t^Pt^R+ t^Qf''(t^Q) - f'(t^Q)) \\
& = & -d \biggl( \frac{p^3}{p^1}+ \frac{p^2}{p^1}
f''(\frac{p^2}{p^1})- f_0'(\frac{p^2}{p^1}) \biggr), \\
dp_1 \cdot dp_1 & = & (2t^Pt^S+2t^Qt^R)dt^P
+ (2t^Pt^R+(t^Q)^2f'''_0(t^Q)) dt^Q \\
& + & 2t^Pt^Qdt^R + (t^P)^2dt^S \\
& = & d \biggl((t^P)^2t^S + 2t^Pt^Qt^R
+ 2f_0(t^Q)- 2t^Qf_0'(t^Q) + (t^Q)^2f_0''(t^Q) \biggr).
\een
Note $2F_0$ appears on the right-hand side of the last equality.
So we have:
\begin{equation*}
\begin{split}
\frac{\pd^2F_*}{\pd p^1\pd p^4} & = \log p^1, \\
\frac{\pd^2F_*}{\pd p^1\pd p^3} & = \frac{p^2}{p^1}, \\
\frac{\pd^2F_*}{\pd p^1\pd p^2} & =
\frac{p^3}{p^1} + \frac{p^2}{p^1} f_0''(\frac{p^2}{p^1})
- f'_0(\frac{p^2}{p^1}), \\
\frac{\pd^2F_*}{\pd p^1\pd p^1}
& =  (t^P)^2t^S + 2t^Pt^Qt^R
+ 2f_0(t^Q)- 2t^Qf_0'(t^Q) + (t^Q)^2f_0''(t^Q) \\
& = \frac{p^1p^4+p^2p^3}{(p^1)^2}
+ 2f(\frac{p^2}{p^1})-2\frac{p^2}{p^1}f_0'(\frac{p^2}{p^1})
+ (\frac{p^2}{p^1})^2f_0''(\frac{p^2}{p^1}).
\end{split}
\end{equation*}

\begin{prop}
For the quintic we have
\be
F_* = (p^1p^4-p^2p^3)\log p^1 - p^1p^4
+ (p^1)^2f_0(\frac{p^2}{p^1}).
\ee
\end{prop}

Note we can verify that
\be
\sum_a
p^a \frac{\pd F_*}{\pd p^a} = 2 F_* +\frac{1}{1- d}
\sum_{a,b} G_{ab}p^ap^b.
\ee
Also note:
\bea
&& E = \frac{\pd}{\pd t^P} - \frac{\pd}{\pd t^R}
- 2\frac{\pd}{\pd t^S}
=-p^1\frac{\pd}{\pd p^1} - p^2 \frac{\pd}{\pd p^2}
- p^3 \frac{\pd}{\pd p^3} - p^4 \frac{\pd}{\pd p^4}, \\
&& e = \frac{\pd}{\pd t^P} = -(p^1)^2 \frac{\pd}{\pd p^1}
- p^1p^2 \frac{\pd}{\pd p^2}.
\eea

\subsection{The almost dual flat coordinates}

The equations for the almost dual flat coordinates
(also called the {\em twisted periods}) are:
\be
\pd_\alpha \xi \cdot \cU
= \xi ¡¤\cdot (\cV + \nu -\frac{1}{2} ) C_\alpha.
\ee
Explicitly we get:
\ben
&& \pd_P(\pd_Px, \pd_Qx, \pd_Rx, \pd_Sx) \\
& = & (\pd_Px, \pd_Qx, \pd_Rx, \pd_Sx) \cdot
\begin{pmatrix}
-2+\nu & 0 & 0 & 0 \\
0 & -1+\nu & 0 & 0 \\
0 & 0 & 0+\nu  & 0 \\
0 & 0 & 0 & 1+\nu
\end{pmatrix} \\
&& \cdot
\frac{1}{(t^P)^2}\begin{pmatrix}
t^P & 0 & 0 & 0 \\ 0 & t^P & 0 & 0 \\
t^R & 0 & t^P & 0 \\ 2t^S & t^R & 0 & t^P
\end{pmatrix} \\
& = & (\pd_Px, \pd_Qx, \pd_Rx, \pd_Sx) \\
&& \cdot
\frac{1}{(t^P)^2}\begin{pmatrix}
(-2+\nu)t^P & 0 & 0 & 0 \\ 0 & (-1+\nu) t^P & 0 & 0 \\
\nu t^R & 0 & \nu t^P & 0 \\ 2(1+\nu)t^S & (1+\nu)t^R & 0 & (1+\nu)t^P
\end{pmatrix}
\een

\ben
&& \pd_Q(\pd_Px, \pd_Qx, \pd_Rx, \pd_Sx) \\
& = & (\pd_Px, \pd_Qx, \pd_Rx, \pd_Sx) \cdot
\begin{pmatrix}
-2+\nu & 0 & 0 & 0 \\
0 & -1+\nu & 0 & 0 \\
0 & 0 & 0+\nu  & 0 \\
0 & 0 & 0 & 1+\nu
\end{pmatrix} \\
&& \cdot \begin{pmatrix}
0 & 0 & 0 & 0 \\
1 & 0 & 0& 0 \\
0 & f_0'''(t^Q)   & 0 & 0 \\
0 & 0 & 1 & 0
\end{pmatrix}
\cdot
\frac{1}{(t^P)^2}\begin{pmatrix}
t^P & 0 & 0 & 0 \\ 0 & t^P & 0 & 0 \\
t^R & 0 & t^P & 0 \\ 2t^S & t^R & 0 & t^P
\end{pmatrix} \\
& = & (\pd_Px, \pd_Qx, \pd_Rx, \pd_Sx)
 \cdot
\frac{1}{(t^P)^2}\begin{pmatrix}
0 & 0 & 0 & 0 \\ (-1+\nu) t^P & 0 & 0 & 0 \\
0 & \nu t^P f_0'''(t^Q) & 0 & 0 \\ (1+\nu)t^R & 0 & (1+\nu)t^P & 0
\end{pmatrix}
\een
\ben
&& \pd_R(\pd_Px, \pd_Qx, \pd_Rx, \pd_Sx) \\
& = & (\pd_Px, \pd_Qx, \pd_Rx, \pd_Sx) \cdot
\begin{pmatrix}
-2+\nu & 0 & 0 & 0 \\
0 & -1+\nu & 0 & 0 \\
0 & 0 & 0+\nu  & 0 \\
0 & 0 & 0 & 1+\nu
\end{pmatrix} \\
&& \cdot \begin{pmatrix}
0 & 0 & 0 & 0 \\
0 & 0 & 0 & 0 \\
1 & 0 & 0 & 0 \\
0 & 1 & 0 & 0
\end{pmatrix}
\cdot
\frac{1}{(t^P)^2}\begin{pmatrix}
t^P & 0 & 0 & 0 \\ 0 & t^P & 0 & 0 \\
t^R & 0 & t^P & 0 \\ 2t^S & t^R & 0 & t^P
\end{pmatrix} \\
& = & (\pd_Px, \pd_Qx, \pd_Rx, \pd_Sx)
 \cdot
\frac{1}{(t^P)^2}\begin{pmatrix}
0 & 0 & 0 & 0 \\ 0 & 0 & 0 & 0 \\
\nu t^P & 0 & 0 & 0 \\ 0 & (1+\nu)t^P & 0 & 0
\end{pmatrix}
\een
\ben
&& \pd_S(\pd_Px, \pd_Qx, \pd_Rx, \pd_Sx) \\
& = & (\pd_Px, \pd_Qx, \pd_Rx, \pd_Sx) \cdot
\begin{pmatrix}
-2+\nu & 0 & 0 & 0 \\
0 & -1+\nu & 0 & 0 \\
0 & 0 & 0+\nu  & 0 \\
0 & 0 & 0 & 1+\nu
\end{pmatrix} \\
&& \begin{pmatrix}
0 & 0 & 0 & 0 \\
0 & 0 & 0 & 0 \\
0 & 0 & 0 & 0 \\
1 & 0 & 0 & 0
\end{pmatrix}
\cdot
\frac{1}{(t^P)^2}\begin{pmatrix}
t^P & 0 & 0 & 0 \\ 0 & t^P & 0 & 0 \\
t^R & 0 & t^P & 0 \\ 2t^S & t^R & 0 & t^P
\end{pmatrix} \\
& = & (\pd_Px, \pd_Qx, \pd_Rx, \pd_Sx)
 \cdot
\frac{1}{(t^P)^2}\begin{pmatrix}
0 & 0 & 0 & 0 \\ 0 & 0 & 0 & 0 \\
0 & 0 & 0 & 0 \\ (1+\nu)t^P & 0 & 0 & 0.
\end{pmatrix}
\een

The equations for $\pd_Sx$ are
\ben
&& \pd_P(\pd_Sx) = \frac{1+\nu}{t^P-\lambda}\pd_Sx, \quad
\pd_Q(\pd_Sx) = \pd_R(\pd_Sx) = \pd_S(\pd_Sx) = 0,
\een
so we get:
\ben
&& \pd_Sx= C_1\cdot (t^P-\lambda)^{1+\nu}.
\een
The equations for $\pd_Rx$ are
\ben
&& \pd_P(\pd_Rx) = \frac{\nu}{t^P-\lambda} \pd_Rx, \\
&& \pd_Q(\pd_Rx) = \frac{1+\nu}{t^P-\lambda}\pd_Sx
= C_1\cdot (1+\nu) (t^P-\lambda)^{\nu}, \\
&& \pd_R(\pd_Rx) = \pd_S(\pd_Rx) = 0,
\een
so we get:
\be
\pd_Rx= C_1\cdot t^Q \cdot (1+\nu)(t^P-\lambda)^\nu
+C_2(t^P-\lambda)^\nu.
\ee
The equations for $\pd_Qx$ are
\ben
\pd_P(\pd_Qx) & = & \frac{(-1+\nu)}{t^P-\lambda}\pd_Qx
+ \frac{(1+\nu)t^R}{(t^P-\lambda)^2} \pd_Sx \\
& = & \frac{(-1+\nu)}{t^P-\lambda}\pd_Qx
+ C_1(1+\nu)t^R (t^P-\lambda)^{\nu-1}, \\
\pd_Q(\pd_Qx) & = & \frac{\nu}{t^P-\lambda} f_0'''(t^Q)\pd_Rx \\
& = & \frac{\nu}{t^P-\lambda} f_0'''(t^Q) \cdot
(C_1\cdot t^Q \cdot (1+\nu)(t^P-\lambda)^\nu+C_2(t^P-\lambda)^\nu), \\
\pd_R(\pd_Qx) & = & \frac{1+\nu}{t^P-\lambda}\pd_Sx =
C_1\cdot (1+\nu) (t^P-\lambda)^\nu, \\
\pd_S(\pd_Qx) & = & 0,
\een
so we get:
\be
\begin{split}
\pd_Qx & = C_1\cdot (1+\nu) \biggl( t^R (t^P-\lambda)^\nu
 + \nu(t^P-\lambda)^{\nu-1}(t^Qf_0''(t^Q)-f_0'(t^Q)) \biggr) \\
& + C_2\nu (t^P-\lambda)^{\nu-1} f_0''(t^Q)
+ C_3(t^P-\lambda)^{\nu-1}.
\end{split}
\ee
The equations for $\pd_Px$ are
\ben
\pd_P(\pd_Px) & = & \frac{-2+\nu}{t^P-\lambda}\pd_Px
+ \frac{2(1+\nu)t^S}{(t^P-\lambda)^2} \pd_Sx \\
& = & \frac{(-2+\nu)}{t^P-\lambda}\pd_Px
+ C_1 \cdot 2(1+\nu)t^S (t^P-\lambda)^{\nu-1}, \\
\pd_Q(\pd_Px) & = & \frac{-1+\nu}{t^P-\lambda}\pd_Qx
+ \frac{(1+\nu)t^R}{(t^P-\lambda)^2} \pd_Sx \\
& = & \frac{\nu-1}{t^P-\lambda}\biggl(
 C_1\cdot (1+\nu) \biggl( t^R (t^P-\lambda)^\nu
 + \nu(t^P-\lambda)^{\nu-1}(t^Qf_0''(t^Q)-f_0'(t^Q)) \biggr) \\
& + & C_2\cdot \nu  (t^P-\lambda)^{\nu-1}f_0''(t^Q)
+ C_3(t^P-\lambda)^{\nu-1}
\biggr)
+ C_1\cdot (\nu+1) t^R(t^P-\lambda)^{\nu-1} \\
& = & C_1 \biggl((\nu+1)\nu t^R(t^P-\lambda)^{\nu-1}
+ (\nu+1)\nu(\nu-1)(t^P-\lambda)^{\nu-2}
(t^Qf_0''(t^Q)-f_0'(t^Q))\biggr) \\
& + & C_2 \cdot \nu(\nu-1) (t^P-\lambda)^{\nu-2}
f_0''(t^Q)
+ C_3 \cdot (\nu-1)(t^P-\lambda)^{\nu-2}, \\
\pd_R(\pd_Px) & = & \frac{\nu}{t^P-\lambda} \pd_Rx \\
& = & C_1\cdot t^Q \cdot (1+\nu)\nu (t^P-\lambda)^{\nu-1}
+C_2\cdot \nu (t^P-\lambda)^{\nu-1}, \\
\pd_S(\pd_Px) & = & \frac{1+\nu}{t^P-\lambda}\pd_Sx
=C_1\cdot (\nu+1) (t^P-\lambda)^{\nu},
\een
so we get:
\be
\begin{split}
\pd_Px= C_1& \cdot \biggl( t^S \cdot (\nu+1)(t^P-\lambda)^\nu
+(\nu+1)\nu (t^P-\lambda)^{\nu-1}  t^Qt^R\\
& + (\nu+1)\nu(\nu-1)(t^P-\lambda)^{\nu-2}
(t^Qf_0'(t^Q)-f_0(t^Q))\biggr) \\
+ C_2 \cdot &
\nu(\nu-1)(t^P-\lambda)^{\nu-2} f_0'(t^Q)   \\
+ C_3 \cdot & (\nu-1)(t^P-\lambda)^{\nu-2} t^Q
+ C_4 \cdot (\nu-1) (t^P-\lambda)^{\nu-2}.
\end{split}
\ee
The general solution is:
\be
\begin{split}
x= & C_0 + C_1((t^P-\lambda)^{\nu+1}t^S
+ (\nu+1) (t^P-\lambda)^{\nu}t^Qt^R \\
& +(\nu+1)\nu(t^P-\lambda)^{\nu-1}(t^Qf_0'(t^Q)-f_0(t^Q))) \\
+ & C_2\cdot
(t^R (t^P-\lambda)^{\nu}+ \nu(t^P-\lambda)^{\nu-1}
f_0'(t^Q)) \\
+ & C_3\cdot (t^P-\lambda)^{\nu-1}t^Q \\
+&  C_4\cdot (t^P-\lambda)^{\nu-1}.
\end{split}
\ee

\begin{prop}
For the quintic one can take
the following deformed dual flat coordinates:
\be \label{eqn:p-in-t-lambda}
\begin{split}
p^4(\nu) = & (t^P-\lambda)^{\nu+1}t^S
+(\nu+1)(t^P-\lambda)^{\nu}t^Qt^R \\
& +(\nu+1)\nu(t^P-\lambda)^{\nu-1}(t^Qf_0'(t^Q)-f_0(t^Q)), \\
p^3(\nu) = & t^R (t^P-\lambda)^{\nu}
+ \nu(t^P-\lambda)^{\nu-1} f_0'(t^Q), \\
p^2(\nu) = & (t^P-\lambda)^{\nu-1}t^Q \\
p^1(\nu) = & (t^P-\lambda)^{\nu-1}.
\end{split}
\ee
\end{prop}

Note the following property of the deformed
dual flat coordinates:
\be
\begin{split}
& \frac{\pd p^1(\nu)}{\pd t^P}
= (\nu-1) \cdot p^1(\nu-1), \\
& \frac{\pd p^2(\nu)}{\pd t^P}
= (\nu-1) \cdot p^2(\nu-1), \\
& \frac{\pd p^3(\nu)}{\pd t^P}
= \nu \cdot p^3(\nu-1), \\
& \frac{\pd p^4(\nu)}{\pd t^P}
= (\nu+1) \cdot p^1(\nu-1).
\end{split}
\ee
The inverse map of \eqref{eqn:p-in-t-lambda} is given by:
\be
\begin{split}
& t^P-\lambda  = (p^1(\nu))^{1/(\nu-1)}, \\
& t^Q= \frac{p^2(\nu)}{p^1(\nu)}, \\
& t^R = \frac{p^3(\nu)-\nu p^1(\nu)
f_0'(\frac{p_2(\nu)}{p_1(\nu)})}{(p^1(\nu))^{\nu/(\nu-1)}},
\end{split}
\ee
\ben
t^S & = & (p^1)^{-(\nu+1)/(\nu-1)}
\biggl( p^4-(\nu+1) (p^1)^{\nu/(\nu-1)} \cdot
\frac{p^2}{p^1} \cdot \frac{p^3- \nu p^1 f'(\frac{p^2}{p^1})}
{(p^1)^{\nu/(\nu-1)} } \\
& + & (\nu+1)\nu p^1 \big( \frac{p^2}{p^1}f_0'(\frac{p^2}{p^1})
-  f_0(\frac{p^2}{p^1}) \big) \biggr)  \\
 & = & (p^1)^{-(\nu+1)/(\nu-1)}
\biggl( p^4-(\nu+1)   \cdot
\frac{p^2 \cdot p^3 }{p^1} -
(\nu+1)\nu p^1 f_0(\frac{p^2}{p^1}) \biggr).
\een
According to Dubrovin \cite[Proposition 5.11]{Dub-Almost-Dual},
if one set $f=(-1-\nu) t^S$.
Then the intersection form is given by the Hessian of $f$:
\be
(\frac{\pd}{\pd p^i}, \frac{\pd}{\pd p^j})
=\frac{\pd^2f}{\pd p^i\pd p^j}.
\ee

\subsection{Deformed dual flat coordinates
as Laplace transform of the deformed flat coordinates}
\label{sec:Laplace}

As observed by Givental \cite{Giv},
the deformed flat coordinates are given by the genus zero one-point functions on
the small phase space.
For the quintic,
such functions have been computed in \S \ref{sec:One-Point-0}
and have been verify to satisfy the quantum differential equations
in \S \ref{sec:QDE}.

Using notations of \S \ref{sec:Def-flat-coor},
we can rewrite the results in \S \ref{sec:One-Point-0} as follows:
\ben
x^P(z) & = & t^P+\sum_{n\geq 0}^\infty z^{n+1}\corrr{\tau_n(S)}_0
= \frac{e^{t_{P}z}-1}{z}, \\
x^Q(z) & = & t^Q + \sum_{n\geq 0}^\infty z^{n+1}\corrr{\tau_n(R)}_0= e^{t^Pz}t^Q, \\
x^R(z) & = & t^R+ \sum_{n=0}^\infty z^{n+1}\corrr{\tau_n(Q)}_0
= e^{t^Pz} t^R + z e^{t^Pz}f'_0(t^Q), \\
x^S(z) & = & t^S+ \sum_{n\geq 0}^\infty z^{n+1} \corrr{\tau_n(P)}_0 \\
& = & e^{t^Pz}t^S
+ ze^{t^Pz}t^Qt^R
+ z^2 e^{t^Pz} (t^Qf_0'(t^Q)-f_0(t^Q)).
\een

\begin{prop}
For the quintic,
the deformed dual flat coordinates are given by
the Laplace transform of the deformed flat coordinates.
More precisely,
\bea
p^1(\nu) & = & \frac{(-1)^{1-\nu}}{\Gamma(1-\nu) }
\int_0^\infty e^{-\lambda z} \cdot z^{1-\nu} (x^P(z)+\frac{1}{z}) dz, \label{eqn:p1(nu)} \\
p^2(\nu) & = & \frac{(-1)^{1-\nu}}{\Gamma(1-\nu) }\int_0^\infty e^{-\lambda z} \cdot z^{0-\nu} x^Q(z) dz, \\
p^3(\nu) & = & \frac{(-1)^{-\nu}}{\Gamma(-\nu) }
\int_0^\infty e^{-\lambda z} \cdot z^{-1-\nu} x^R(z) dz, \\
p^4(\nu) & = & \frac{(-1)^{-1-\nu}}{\Gamma(-1-\nu)}\int_0^\infty e^{-\lambda z} \cdot z^{-2-\nu} x^S(z) dz.
\eea
\end{prop}

\begin{proof}
By the definition of the Gamma-function we have:
\ben
&& \int_0^\infty e^{-\lambda z} \cdot z^{1-\nu} (x^P(z)+\frac{1}{z}) dz
= \int_0^\infty e^{-\lambda z} \cdot z^{1-\nu} \frac{e^{t^Pz}}{z} dz \\
& = & \int_0^\infty e^{-(\lambda-t^P) z} \cdot z^{(1-\nu)-1} dz \\
& = & (\lambda-t^P)^{\mu-1} \cdot \int_0^\infty e^{-z} \cdot z^{(1-\nu)-1} dz \\
& = & \Gamma(1-\nu) \cdot (\lambda-t^P)^{\nu-1}.
\een
This proves \eqref{eqn:p1(nu)}.
The other three identities are proved in the same fashion.
\end{proof}

\section{Emergent K\"ahler Geometry on the Small Phase Space of the Quintic}

\label{sec:Quintic-Spec}

We now present the special geometry on the small phase space of the quintic.

\subsection{The flat sections for the Dubrovin connection}

By the results in \S \ref{sec:Two-Point-I} and \S \ref{sec:QDE},
we know that the following are flat sections
for the Dubrovin connection:
\ben
s_S & = & e^{-t^P} \frac{\pd}{\pd t^S}, \\
s_R & = & e^{-t^P}\biggl(\frac{\pd}{\pd t^R} - t^Q \frac{\pd}{\pd t^S}\biggr), \\
s_Q & = & e^{-t^P}\biggl[\frac{\pd}{\pd t^Q}
- \frac{\pd^2 f_0(t^Q)}{\pd (t^Q)^2} \frac{\pd}{\pd t^R}
+ \biggl(\frac{\pd f_0(t^Q)}{\pd t^Q} - t^R
\biggr)\frac{\pd}{\pd t^S}  \biggr], \\
s_P & = & e^{-t^P}\biggl[\frac{\pd}{\pd t^P} - t^Q\frac{\pd}{\pd t^Q}
+ \biggl(t^Q \frac{\pd^2f_0(t^Q)}{\pd (t^Q)^2}
-\frac{\pd f_0(t^Q)}{\pd t^Q} -t^R \biggr)  \frac{\pd}{\pd t^R} \\
& - & \biggl(t^Q\frac{\pd f_0(t^Q)}{\pd t^Q} - 2 f_0(t^Q) +t^S-t^Qt^R\biggr)
\frac{\pd}{\pd t^S}\biggr].
\een
They are given by the formulas for $S_{\alpha,\beta}(z)$.
For example,
\ben
s_P &= &S_{S,P}(-1)\frac{\pd}{\pd t^P} + S_{R,P}(-1)\frac{\pd}{\pd t^Q}
+ S_{Q,P}(-1)\frac{\pd}{\pd t^R} + S_{P,P}(-1)\frac{\pd}{\pd t^S}.
\een

\subsection{The complex symplectic structure}

We introduce a complex symplectic structure such that:
\be
\Omega(s_P, s_S) = -\Omega(s_S,s_P) = - \Omega (s_Q, s_R) =
= \Omega(s_R,s_Q) = 1,
\ee
and all other $\omega(s_\alpha, s_\beta) = 0$.
Because we have:
\ben
\frac{\pd}{\pd t^S} & = & e^{t^P} s_S, \\
\frac{\pd}{\pd t^R} & = & e^{t^P} \biggl[s_R + t^Q s_S\biggr], \\
\frac{\pd}{\pd t^Q} & = & e^{t^P} \biggl[s_Q  + \frac{\pd^2 f_0(t^Q)}{\pd (t^Q)^2}s_R
+ \biggl(t^Q \frac{\pd^2 f_0(t^Q)}{\pd (t^Q)^2}
- \frac{\pd f_0(t^Q)}{\pd t^Q} + t^R\biggr) s_S \biggr], \\
\frac{\pd}{\pd t^P} & = & e^{t^P} \biggl[s_P  + t^Qs_Q
+ \biggl(\frac{\pd f_0(t^Q)}{\pd t^Q} + t^R\biggr) s_R \\
& + & \biggl(t^Q\frac{\pd f_0(t^Q)}{\pd t^Q}
- 2 f_0(t^Q) +t^S + t^Qt^R\biggr)
s_S \biggr],
\een
we get:
\ben
&& \Omega( \frac{\pd}{\pd t^P}, \frac{\pd}{\pd t^Q})
= 2e^{2t^P} t^R, \\
&& \Omega( \frac{\pd}{\pd t^P}, \frac{\pd}{\pd t^R})
= 0, \\
&& \Omega( \frac{\pd}{\pd t^P}, \frac{\pd}{\pd t^S})
= e^{2t^P}, \\
&& \Omega( \frac{\pd}{\pd t^Q}, \frac{\pd}{\pd t^R})
= - e^{2t^P}, \\
&& \Omega( \frac{\pd}{\pd t^Q}, \frac{\pd}{\pd t^S}) = 0, \\
&& \Omega( \frac{\pd}{\pd t^R}, \frac{\pd}{\pd t^S}) = 0.
\een
Therefore,
in the flat coordinates $t^P, t^Q, t^R, t^S$ on the small phase space of the quintic,
we have:
\be
\Omega = 2e^{2t^P} t^R dt^P \wedge d t^Q
+ e^{2t^P} dt^P \wedge dt^S
- e^{2t^P} dt^Q \wedge dt^R.
\ee
It is clear that
\be
d\Omega =0.
\ee
I.e.,
$\Omega$ is a holomorphic symplectic structure on the small phase space.
Note we can find Darboux coordinates for this symplectic structaure:
\be
\Omega = d(e^{t^P}t^R) \wedge d (e^{t^P}t^Q)
+ d(e^{t^P}) \wedge d (e^{t^P}t^S).
\ee

\begin{rmk}
When $t^P=t^R=t^S=0$,
we have
\be
\frac{\pd}{\pd t^P} =  s_P  + t^Qs_Q
+  \frac{\pd f_0(t^Q)}{\pd t^Q}  s_R
 +  \biggl(t^Q\frac{\pd f_0(t^Q)}{\pd t^Q}
- 2 f_0(t^Q)  \biggr)
s_S.
\ee
The coefficients on the right-hand side are
$$(1, t^Q, f'_0(t^Q), t^Qf_0'(t^Q) - 2f_0(t^Q)).$$
These also appear in \eqref{eqn:4Periods}
which describe the special geometry of the four periods
of the Picard-Fuchs equation of the quintic.
\end{rmk}

\subsection{The real structure}

We define a real structure $\tau$ on the small phase of the quintic by requiring:
\be
\tau(s_\alpha) = s_\alpha, \quad \alpha = P, Q, R, S.
\ee
The action of $\tau$ on the basis $\frac{\pd}{\pd t^\alpha}$ is computed as follows:
\ben
\frac{\pd}{\pd t^S} & \mapsto & s_S = \frac{\pd}{\pd t^S}, \\
\frac{\pd}{\pd t^R} & \mapsto & s_Q
+ \overline{t^Q} s_S \\
& = & \frac{\pd}{\pd t^R}
+(\overline{t^Q}-t^Q) \frac{\pd}{\pd t^S}, \\
\frac{\pd}{\pd t^Q} & \mapsto & s_Q
+ \overline{\frac{\pd^2 f_0(t^Q)}{\pd (t^Q)^2}} s_R
+ \biggl(\overline{t^Q} \overline{\frac{\pd^2 f_0(t^Q)}{\pd (t^Q)^2}}
- \overline{\frac{\pd f_0(t^Q)}{\pd t^Q}}\biggr) s_S \\
& = & \frac{\pd}{\pd t^Q}
+ \biggl(\overline{\frac{\pd^2 f_0(t^Q)}{\pd (t^Q)^2}}
- \frac{\pd^2 f_0(t^Q)}{\pd (t^Q)^2}\biggr) \frac{\pd}{\pd t^R}
\\
& + & \biggl((\overline{t^Q}-t^Q) \overline{\frac{\pd^2 f_0(t^Q)}{\pd (t^Q)^2}}
+ \frac{\pd f_0(t^Q)}{\pd t^Q} - \overline{\frac{\pd f_0(t^Q)}{\pd t^Q}}\biggr) \frac{\pd}{\pd t^S}, \\
\frac{\pd}{\pd t^P} & \mapsto & s_P  + \overline{t^Q}s_Q
+ \overline{\frac{\pd f_0(t^Q)}{\pd t^Q}} s_R
+ \biggl(\overline{t^Q}\overline{\frac{\pd f_0(t^Q)}{\pd t^Q}}
- 2 \overline{f_0(t^Q)} \biggr) s_S \\
& = & \frac{\pd}{\pd t^P} - t^Q\frac{\pd}{\pd t^Q}
+ \biggl(t^Q \frac{\pd^2f_0(t^Q)}{\pd (t^Q)^2}
-\frac{\pd f_0(t^Q)}{\pd t^Q} \biggr)  \frac{\pd}{\pd t^R} \\
& - & \biggl(t^Q\frac{\pd f_0(t^Q)}{\pd t^Q} - 2 f_0(t^Q) \biggr) \frac{\pd}{\pd t^S}.
\\
& + & \overline{t^Q} \biggl( \frac{\pd}{\pd t^Q} - \frac{\pd^2 f_0(t^Q)}{\pd (t^Q)^2}
\frac{\pd}{\pd t^R} + \frac{\pd f_0(t^Q)}{\pd t^Q} \frac{\pd}{\pd t^S} \biggr) \\
& + &  \overline{\frac{\pd f_0(t^Q)}{\pd t^Q}}
\biggl(\frac{\pd}{\pd t^R} - t^Q \frac{\pd}{\pd t^S} \biggr)
+ \biggl(\overline{t^Q}\overline{\frac{\pd f_0(t^Q)}{\pd t^Q}}
- 2 \overline{f_0(t^Q)} \biggr) \frac{\pd}{\pd t^S} \\
& = & \frac{\pd}{\pd t^P} - (t^Q-\overline{t^Q})\frac{\pd}{\pd t^Q} \\
& + & \biggl( (t^Q-\overline{t^Q}) \frac{\pd^2f_0(t^Q)}{\pd (t^Q)^2}
-\frac{\pd f_0(t^Q)}{\pd t^Q}
+ \overline{\frac{\pd f_0(t^Q)}{\pd t^Q}} \biggr)  \frac{\pd}{\pd t^R} \\
& - & \biggl[(t^Q-\overline{t^Q})
\biggl(\frac{\pd f_0(t^Q)}{\pd t^Q}
+ \overline{\frac{\pd f_0(t^Q)}{\pd t^Q}} \biggr)
 - 2 (f_0(t^Q)-\overline{f(t^Q)}) \biggr] \frac{\pd}{\pd t^S}.
\een

\subsection{The pseudo-Hermitian metric}

Using the holomorphic symplectic form
and the real structure defined the preceding two Subsections we now define:
\be
h(\frac{\pd}{\pd t^\alpha}, \frac{\pd}{\pd t^\beta})
= i \Omega(\frac{\pd}{\pd t^\alpha}, \tau(\frac{\pd}{\pd t^\beta})).
\ee
Explicitly, we have:
\ben
h(\frac{\pd}{\pd t^P}, \frac{\pd}{\pd t^P})
& = & i \Omega(\frac{\pd}{\pd t^P}, \tau(\frac{\pd}{\pd t^P})) \\
& = & i e^{t^P+\overline{t^P}}
\biggl[-\biggl(t^Q\frac{\pd f_0(t^Q)}{\pd t^Q}
- 2 f_0(t^Q) +t^S + t^Qt^R\biggr) \\
& + & \biggl(\overline{t^Q}\overline{\frac{\pd f_0(t^Q)}{\pd t^Q}}
- 2 \overline{f_0(t^Q)} + \overline{t^S} + \overline{t^Qt^R} \biggr) \\
& + &\overline{t^Q} \biggl(\frac{\pd f_0(t^Q)}{\pd t^Q} + t^R\biggr)
- t^Q\biggl(\overline{\frac{\pd f_0(t^Q)}{\pd t^Q}}
+ \overline{t^R}\biggr)  \biggr] \\
& = & i e^{t^P+\overline{t^P}}
\biggl[-(t^Q-\overline{t^Q})
\biggl(\frac{\pd f_0(t^Q)}{\pd t^Q} + \overline{\frac{\pd f_0(t^Q)}{\pd t^Q}}\biggr) \\
& + & 2 \biggl(f_0(t^Q) - \overline{f_0(t^Q)} \biggr)
- (t^S - \overline{t^S})
- ( t^Q - \overline{t^Q}) (t^R+\overline{t^R}) \biggr].
\een

\ben
&& h(\frac{\pd}{\pd t^P}, \frac{\pd}{\pd t^Q})
= i\Omega(\frac{\pd}{\pd t^P}, \tau(\frac{\pd}{\pd t^Q})) \\
& = & i e^{t^P+\overline{t^P}}
\biggl[\biggl(\overline{t^Q}
\overline{\frac{\pd^2 f_0(t^Q)}{\pd (t^Q)^2}}
- \overline{\frac{\pd f_0(t^Q)}{\pd t^Q}} + \overline{t^R}\biggr)
- t^Q\overline{\frac{\pd^2 f_0(t^Q)}{\pd (t^Q)^2}}
+ \biggl(\frac{\pd f_0(t^Q)}{\pd t^Q} + t^R\biggr) \biggr] \\
& = & i e^{t^P+\overline{t^P}}
\biggl[-\biggl(t^Q-\overline{t^Q}\biggr)
\overline{\frac{\pd^2 f_0(t^Q)}{\pd (t^Q)^2}}
+ \biggl(\frac{\pd f_0(t^Q)}{\pd t^Q}
- \overline{\frac{\pd f_0(t^Q)}{\pd t^Q}} \biggr) +
\biggl( t^R + \overline{t^R}\biggr)  \biggr].
\een

\ben
&& h(\frac{\pd}{\pd t^P}, \frac{\pd}{\pd t^R})
= i \Omega(\frac{\pd}{\pd t^P}, \tau(\frac{\pd}{\pd t^R}))
= -i e^{t^P+\overline{t^P}}
\biggl[t^Q-\overline{t^Q}  \biggr].
\een

\ben
&& h(\frac{\pd}{\pd t^P}, \frac{\pd}{\pd t^S})
= i\Omega(\frac{\pd}{\pd t^P}, \tau(\frac{\pd}{\pd t^S}))
= ie^{t^P+\overline{t^P}}.
\een

\ben
h(\frac{\pd}{\pd t^Q}, \frac{\pd}{\pd t^Q})
& = & i\Omega(\frac{\pd}{\pd t^Q}, \tau(\frac{\pd}{\pd t^Q}))
= i e^{t^P+\overline{t^P}}
\biggl[ \frac{\pd^2 f_0(t^Q)}{\pd (t^Q)^2}
- \overline{\frac{\pd^2 f_0(t^Q)}{\pd (t^Q)^2}} \biggr].
\een

\ben
h(\frac{\pd}{\pd t^Q}, \frac{\pd}{\pd t^R})
& = & i\Omega(\frac{\pd}{\pd t^Q}, \tau(\frac{\pd}{\pd t^R}))
=- ie^{t^P+\overline{t^P}}.
\een

\ben
h(\frac{\pd}{\pd t^Q}, \frac{\pd}{\pd t^S})
& = & i\Omega(\frac{\pd}{\pd t^Q}, \tau(\frac{\pd}{\pd t^S}))
= 0.
\een

\ben
h(\frac{\pd}{\pd t^R}, \frac{\pd}{\pd t^R})
& = & i\Omega(\frac{\pd}{\pd t^R}, \tau(\frac{\pd}{\pd t^R}))
= 0.
\een

\ben
h(\frac{\pd}{\pd t^R}, \frac{\pd}{\pd t^S})
& = & i\Omega(\frac{\pd}{\pd t^R}, \tau(\frac{\pd}{\pd t^S}))
= 0.
\een
\ben
h(\frac{\pd}{\pd t^S}, \frac{\pd}{\pd t^S})
& = & i\Omega(\frac{\pd}{\pd t^S}, \tau(\frac{\pd}{\pd t^S}))
= 0.
\een
With these explicit expressions,
we can directly check the following:

\begin{prop}
The K\"ahler potential $K$ for $h$ can be taken to be:
\ben
h(\frac{\pd}{\pd t^P}, \frac{\pd}{\pd t^P})
& = & i e^{t^P+\overline{t^P}}
\biggl[-(t^Q-\overline{t^Q})
\biggl(\frac{\pd f_0(t^Q)}{\pd t^Q} + \overline{\frac{\pd f_0(t^Q)}{\pd t^Q}}\biggr) \\
& + & 2 \biggl(f_0(t^Q) - \overline{f_0(t^Q)} \biggr)
- (t^S - \overline{t^S})
- ( t^Q - \overline{t^Q}) (t^R+\overline{t^R}) \biggr] \\
& = & i
\biggl[ e^{t^P} \cdot
\overline{e^{t^P}\biggl( t^Q\frac{\pd f_0(t^Q)}{\pd t^Q} - 2f_0(t^Q) + t^S + t^Qt^R\biggr)} \\
& - & e^{t^P} \biggl( t^Q \frac{\pd f_0(t^Q)}{\pd t^Q} - 2 f_0(t^Q) + t^S + t^Qt^R
\biggr)
\cdot e^{\overline{t^P}} \\
& - & \biggl( e^{t^P} t^Q \cdot
\overline{e^{t^P} \biggl( \frac{\pd f_0(t^Q)}{\pd t^Q} + t^R \biggr)} \\
&& - \overline{e^{t^P}t^Q} \cdot e^{t^P} \biggl( \frac{\pd f_0(t^Q)}{\pd t^Q} +t^R
\biggr) \biggr) \biggr].
\een
\end{prop}

\subsection{Hyperk\"aher potential}

From the explicit formula for the K\"ahler potential $K$,
we introduce:
\begin{align}
u & = e^{t^P}, & v & = e^{t^P}t^Q, & w & = e^{t^P}t^R, &
x & = e^{t^P}(t^S+t^Qt^R).
\end{align}
We also introduce
\be
\cF = - u^2 f_0(\frac{v}{u}) + ux - vw.
\ee
Note we have
\ben
\frac{\pd \cF}{\pd u} & = & e^{t^P}t^Q \frac{\pd f_0(t^Q)}{\pd t^Q}
- 2e^{t^P}f_0(t^Q) + e^{t^P} (t^S + t^Q t^R) \\
& = & vf_0'(\frac{v}{u}) - 2u f_0(\frac{v}{u}) + x, \\
\frac{\pd \cF}{\pd v} & = & -(e^{t^P}\frac{\pd f_0(t^Q)}{\pd t^Q} + e^{t^P}t^R)
= - u f_0'(\frac{v}{u}) - w,
\een
and so
\ben
d\frac{\pd \cF}{\pd u} & = &
\biggl(-\frac{v^2}{u^2} f_0''(\frac{v}{u})
+ \frac{2v}{u} f_0'(\frac{v}{u})
- 2 f(\frac{v}{u})\biggr) du\\
& + &\biggl( \frac{v}{u} f_0''(\frac{v}{u})
- f_0'(\frac{v}{u}) \biggr) dv + dx, \\
d \frac{\pd \cF}{\pd v}
& = & \biggl( \frac{v}{u} f_0''(\frac{v}{u})
- f_0'(\frac{v}{u})\biggr) du
- f_0''(\frac{v}{u}) dv
-dw.
\een
Now we easily check the following:

\begin{prop}
In the new coordinates $u,v,w$ ad n$x$,
the holomorphic symplectic form and the K\"ahler potential can be written as:
\be
\Omega = d u \wedge d\frac{\pd \cF}{\pd u}
+ dv \wedge d \frac{\pd \cF}{\pd v}
= du \wedge dx - dv \wedge dw.
\ee
\be
K =
i \biggl[ u \overline{\frac{\pd \cF}{\pd u}}
+ v \overline{\frac{\pd \cF}{\pd v}}
- \bar{u} \frac{\pd \cF}{\pd u}
- \bar{v} \frac{\pd \cF}{\pd v} \biggr]
\ee
\end{prop}

\begin{prop}
The holomorphic symplectic form $\Omega$ is parallel with respect to
the Levi-Civita connection of $h$.
\end{prop}

\begin{proof}
Let us first compute the pseudo-Hermitian metric
in the new coordinates $u,v,w,x$:
\ben
\frac{\pd K}{\pd u} & = &
i \biggl[ \overline{\frac{\pd \cF}{\pd u}}
- \bar{u} \frac{\pd^2 \cF}{\pd u^2}
- \bar{v} \frac{\pd^2 \cF}{\pd u\pd v} \biggr], \\
\frac{\pd K}{\pd v} & = & i \biggl[
\overline{\frac{\pd \cF}{\pd v}}
- \bar{u} \frac{\pd^2 \cF}{\pd u\pd v}
- \bar{v} \frac{\pd^2 \cF}{\pd v^2} \biggr], \\
\frac{\pd K}{\pd w}
& = & i \bar{v}, \\
\frac{\pd K}{\pd x}
& = & - i  \bar{u},
\een

\begin{align*}
\frac{\pd^2 K}{\pd u \pd \bar{u}} & =
i \biggl[ \overline{\frac{\pd^2 \cF}{\pd u^2}}
- \frac{\pd^2 \cF}{\pd u^2}  \biggr], &
\frac{\pd^2 K}{\pd u \pd \bar{v}} & =
i \biggl[ \overline{\frac{\pd^2 \cF}{\pd u\pd v}}
- \frac{\pd^2 \cF}{\pd u\pd v}  \biggr], \\
\frac{\pd^2 K}{\pd u \pd \bar{w}} & =  0, &
\frac{\pd^2 K}{\pd u \pd \bar{x}} & =  i, \\
\frac{\pd^2 K}{\pd v\pd \bar{v}} & = i \biggl[
\overline{\frac{\pd^2 \cF}{\pd v^2}}
- \frac{\pd^2 \cF}{\pd v^2} \biggr], &
\frac{\pd^2 K}{\pd v\pd \bar{w}} & =  -i, \\
\frac{\pd^2 K}{\pd v\pd \bar{x}} & = 0, &
\frac{\pd^2 K}{\pd w\pd \bar{w}} & =  0, \\
\frac{\pd^2 K}{\pd w\pd \bar{x}} & = 0, &
\frac{\pd^2 K}{\pd x\pd \bar{x}} & = 0.
\end{align*}
So we get the matrix for the metric:
\ben
h=\begin{pmatrix}
i \biggl[ \overline{\frac{\pd^2 \cF}{\pd u^2}}
- \frac{\pd^2 \cF}{\pd u^2}  \biggr] &
i \biggl[ \overline{\frac{\pd^2 \cF}{\pd u\pd v}}
- \frac{\pd^2 \cF}{\pd u\pd v}  \biggr] &
0 & i \\
i \biggl[ \overline{\frac{\pd^2 \cF}{\pd u\pd v}}
- \frac{\pd^2 \cF}{\pd u\pd v}  \biggr] &
i \biggl[ \overline{\frac{\pd^2 \cF}{\pd v^2}}
- \frac{\pd^2 \cF}{\pd v^2} \biggr]  & -i & 0 \\
0 & i &  0 & 0 \\
-i & 0 & 0 & 0
\end{pmatrix}
\een
and its inverse matrix:
\ben
h^{-1}=\begin{pmatrix}
0 & 0 & 0 & i \\
0 & 0 & -i & 0 \\
0 & i &  -i \biggl[ \overline{\frac{\pd^2 \cF}{\pd v^2}}
- \frac{\pd^2 \cF}{\pd v^2} \biggr] & i \biggl[ \overline{\frac{\pd^2 \cF}{\pd u\pd v}}
- \frac{\pd^2 \cF}{\pd u\pd v}  \biggr] \\
-i & 0 & i \biggl[ \overline{\frac{\pd^2 \cF}{\pd u\pd v}}
- \frac{\pd^2 \cF}{\pd u\pd v}  \biggr] &
-i \biggl[ \overline{\frac{\pd^2 \cF}{\pd u^2}}
- \frac{\pd^2 \cF}{\pd u^2}  \biggr]
\end{pmatrix}.
\een

Next we compute the Levi-Civita connection.
First we have
\ben
\pd_u h=\begin{pmatrix}
- i \frac{\pd^3 \cF}{\pd u^3}  &
- i \frac{\pd^3 \cF}{\pd u^2\pd v}  &
0 & 0 \\
- i \frac{\pd^3 \cF}{\pd u^2\pd v}  &
- i \frac{\pd^3 \cF}{\pd u \pd v^2} &
0 & 0 \\
0 & 0 & 0 & 0 \\
0 & 0 & 0 & 0
\end{pmatrix},
\een
then we get:
\be
\pd_u h \cdot h^{-1}
= \begin{pmatrix}
0 & 0 & - \frac{\pd^3 \cF}{\pd u^2\pd v}
& \frac{\pd^3 \cF}{\pd u^3}   \\
0 & 0 & - \frac{\pd^3 \cF}{\pd u \pd v^2} &
\frac{\pd^3 \cF}{\pd u^2\pd v}   \\
0 & 0 & 0 & 0 \\
0 & 0 & 0 & 0
\end{pmatrix}.
\ee
This implies
\be
\nabla_{\frac{\pd}{\pd u}}
\begin{pmatrix}
\frac{\pd}{\pd u}  \\ \frac{\pd}{\pd v} \\
\frac{\pd}{\pd w} \\ \frac{\pd}{\pd x} \end{pmatrix}
= \begin{pmatrix}
0 & 0 & - \frac{\pd^3 \cF}{\pd u^2\pd v}
& \frac{\pd^3 \cF}{\pd u^3}   \\
0 & 0 & - \frac{\pd^3 \cF}{\pd u \pd v^2} &
\frac{\pd^3 \cF}{\pd u^2\pd v}   \\
0 & 0 & 0 & 0 \\
0 & 0 & 0 & 0
\end{pmatrix}
\cdot
\begin{pmatrix}
\frac{\pd}{\pd u}  \\ \frac{\pd}{\pd v} \\
\frac{\pd}{\pd w} \\ \frac{\pd}{\pd x} \end{pmatrix}.
\ee
and
\be
\nabla_{\frac{\pd}{\pd u}}
\begin{pmatrix}
du  \\ dv \\ dw \\ dx \end{pmatrix}
= \begin{pmatrix}
0 & 0 & 0 & 0 \\
0 & 0 & 0 & 0 \\
\frac{\pd^3 \cF}{\pd u^2\pd v}
& \frac{\pd^3 \cF}{\pd u \pd v^2} & 0 & 0 \\
-\frac{\pd^3 \cF}{\pd u^3} &
-\frac{\pd^3 \cF}{\pd u^2\pd v}
& 0 & 0
\end{pmatrix}
\cdot
\begin{pmatrix}
du  \\ dv \\ dw \\ dx \end{pmatrix}.
\ee
Therefore,
\ben
\nabla_{\frac{\pd}{\pd u}} \Omega
& = & \nabla_{\frac{\pd}{\pd u}} (d u \wedge dx - dv \wedge dw) \\
& = & \nabla_{\frac{\pd}{\pd u}} d u \wedge dx
+ du \wedge \nabla_{\frac{\pd}{\pd u}} dx
- \nabla_{\frac{\pd}{\pd u}} dv \wedge dw
- dv \wedge \nabla_{\frac{\pd}{\pd u}} dw \\
& = & 0 \wedge dx + du \wedge
\biggl(-\frac{\pd^3 \cF}{\pd u^3} du
-\frac{\pd^3 \cF}{\pd u^2\pd v} dv \biggr) \\
& - & 0 \wedge dw
- dv \wedge \biggl(\frac{\pd^3 \cF}{\pd u^2\pd v} du +
\frac{\pd^3 \cF}{\pd u \pd v^2} dv\biggr)   \\
& = & 0.
\een
The covariant derivative in $v$ can be computed similarly:
\ben
\pd_v h=\begin{pmatrix}
-i \frac{\pd^3 \cF}{\pd u^2\pd v}  &
- i \frac{\pd^3 \cF}{\pd u\pd v^2} &
0 & 0 \\
- i \frac{\pd^3 \cF}{\pd u\pd v^2}  &
- i \frac{\pd^3 \cF}{\pd v^3}  & 0 & 0 \\
0 & 0 & 0 & 0 \\
0 & 0 & 0 & 0
\end{pmatrix}
\een

\ben
\pd_v h \cdot h^{-1}
=\begin{pmatrix}
0 & 0  &
- \frac{\pd^3 \cF}{\pd u\pd v^2}
 & \frac{\pd^3 \cF}{\pd u^2\pd v}   \\
0 & 0 &
- \frac{\pd^3 \cF}{\pd v^3}  &
\frac{\pd^3 \cF}{\pd u\pd v^2}   \\
0 & 0 & 0 & 0 \\
0 & 0 & 0 & 0
\end{pmatrix}
\een

\be
\nabla_{\frac{\pd}{\pd v}}
\begin{pmatrix}
\frac{\pd}{\pd u}  \\ \frac{\pd}{\pd v} \\
\frac{\pd}{\pd w} \\ \frac{\pd}{\pd x} \end{pmatrix}
= \begin{pmatrix}
0 & 0  &
- \frac{\pd^3 \cF}{\pd u\pd v^2}
 & \frac{\pd^3 \cF}{\pd u^2\pd v}   \\
0 & 0 &
- \frac{\pd^3 \cF}{\pd v^3}  &
\frac{\pd^3 \cF}{\pd u\pd v^2}   \\
0 & 0 & 0 & 0 \\
0 & 0 & 0 & 0
\end{pmatrix}
\cdot
\begin{pmatrix}
\frac{\pd}{\pd u}  \\ \frac{\pd}{\pd v} \\
\frac{\pd}{\pd w} \\ \frac{\pd}{\pd x} \end{pmatrix}.
\ee

\be
\nabla_{\frac{\pd}{\pd v}}
\begin{pmatrix}
du  \\ dv \\ dw \\ dx \end{pmatrix}
= \begin{pmatrix}
0 & 0 & 0 & 0 \\
0 & 0 & 0 & 0 \\
\frac{\pd^3 \cF}{\pd u\pd v^2} & \frac{\pd^3 \cF}{\pd v^3}
& 0 & 0 \\
-\frac{\pd^3 \cF}{\pd u^2\pd v}  &
-  \frac{\pd^3 \cF}{\pd u\pd v^2} &
0 & 0
\end{pmatrix}
\cdot
\begin{pmatrix}
du  \\ dv \\ d w \\ d x \end{pmatrix}.
\ee

\ben
\nabla_{\frac{\pd}{\pd v}} \Omega
& = & \nabla_{\frac{\pd}{\pd v}} (d u \wedge dx
- dv \wedge dw) \\
& = & \nabla_{\frac{\pd}{\pd v}} d u \wedge dx
+ du \wedge \nabla_{\frac{\pd}{\pd v}} dx
- \nabla_{\frac{\pd}{\pd v}} dv \wedge dw
- dv \wedge \nabla_{\frac{\pd}{\pd v}} dw \\
& = & 0 \wedge dx + du \wedge
\biggl(-\frac{\pd^3 \cF}{\pd u^2\pd v} du
-\frac{\pd^3 \cF}{\pd u\pd v^2} dv \biggr) \\
& - & 0 \wedge dw
- dv \wedge \biggl(\frac{\pd^3 \cF}{\pd u\pd v^2} du +
\frac{\pd^3 \cF}{\pd v^3} dv\biggr)   \\
& = & 0.
\een
Finally, from
\ben
\pd_w h=\pd_x h=\begin{pmatrix}
0 & 0 & 0 & 0 \\
0 & 0 & 0 & 0 \\
0 & 0 & 0 & 0 \\
0 & 0 & 0 & 0
\end{pmatrix}
\een
we can compute the covariant derivatives in $w$ and $x$.
\end{proof}

The above two Propositions imply that there is a structure of pseodo-hyperk\"ahler
manifold on the small phase space
of the quintic. (We call $\cF$ in the hyperk\"ahler potential of this structure.)
This gives us the emergent special K\"ahler structure on the small phase
space of the quintic.
For mathematical reference on special K\"hler geometry,
see Freed \cite{Fre}.

\subsection{Compatibility with the Frobenius manifold structure}

Since we have computed the covariant derivative in the coordinates $u,v,w,x$,
we will also compute the quantum multiplications in these coordinates.
Note
\begin{align*}
t^P & = \log u, & t^Q & = \frac{v}{u}, &
t^R & = \frac{w}{u}, & t^S & = \frac{x}{u} - \frac{vw}{u^2}.
\end{align*}
So we have
\ben
&& \frac{\pd}{\pd t^P}= u \frac{\pd}{\pd u}
+ v\frac{\pd}{\pd v} + w \frac{\pd}{\pd w}
+ x\frac{\pd}{\pd x}, \\
&& \frac{\pd}{\pd t^Q}= u \frac{\pd}{\pd v}
+ w\frac{\pd}{\pd x}, \\
&& \frac{\pd}{\pd t^R}= u \frac{\pd}{\pd w}
+ v\frac{\pd}{\pd x}, \\
&& \frac{\pd}{\pd t^S}= u\frac{\pd}{\pd x}.
\een

\ben
&& \frac{\pd}{\pd u}
= \frac{1}{u} \frac{\pd}{\pd t^P}
- \frac{v}{u^2}\frac{\pd}{\pd t^Q}
- \frac{w}{u^2}\frac{\pd}{\pd t^R}
- (\frac{x}{u^2} - \frac{2vw}{u^3})
\frac{\pd}{\pd t^S}, \\
&& \frac{\pd}{\pd v} = \frac{1}{u} \frac{\pd}{\pd t^Q}
- \frac{w}{u^2} \frac{\pd}{\pd t^S}, \\
&& \frac{\pd}{\pd w} = \frac{1}{u} \frac{\pd}{\pd t^R}
- \frac{v}{u^2} \frac{\pd}{\pd t^S}, \\
&& \frac{\pd}{\pd x} = \frac{1}{u} \frac{\pd}{\pd t^S}.
\een
We get
\ben
&& \eta=\begin{pmatrix}
-\frac{2x}{u^3} + \frac{6vw}{u^4} & - \frac{2w}{u^3}
& - \frac{2v}{u^3} & \frac{1}{u^2} \\
- \frac{2w}{u^3} & 0 & \frac{1}{u^2} & 0 \\
- \frac{2v}{u^3} & \frac{1}{u^2} & 0 & 0 \\
\frac{1}{u^2} & 0 & 0 & 0
\end{pmatrix}
\een
and
\ben
&& \eta^{-1} = \begin{pmatrix}
0 & 0 & 0 & u^2 \\
0 & 0 & u^2 & 2uv \\
0 & u^2 & 0 & 2uw \\
u^2 & 2uv & 2uw & 2(ux-vw)
\end{pmatrix}
\een
Using the formulas for quantum multiplications in \S \ref{sec:Quan-Coh},
we get:
\be
\begin{split}
& \frac{\pd}{\pd u} \circ
\begin{pmatrix}
\frac{\pd}{\pd u}  \\ \frac{\pd}{\pd v} \\
\frac{\pd}{\pd w} \\ \frac{\pd}{\pd x} \end{pmatrix} \\
= & \begin{pmatrix}
\frac{1}{u} & -\frac{v}{u^2}  & - \frac{w}{u^2} +\frac{v^2}{u^3}f_0'''(\frac{v}{u})
&  \frac{v^3}{u^4}f_0'''(\frac{v}{u})-\frac{x}{u^2}+ \frac{2vw}{u^3} \\
0 & \frac{1}{u} & -\frac{v}{u^2}f_0'''(\frac{v}{u}) &
-\frac{w}{u^2} -\frac{v^2}{u^3}f_0'''(\frac{v}{u})   \\
0 & 0 & \frac{1}{u} & - \frac{v}{u^2} \\
0 & 0 & 0 & \frac{1}{u}
\end{pmatrix}
\cdot
\begin{pmatrix}
\frac{\pd}{\pd u}  \\ \frac{\pd}{\pd v} \\
\frac{\pd}{\pd w} \\ \frac{\pd}{\pd x} \end{pmatrix}.
\end{split}
\ee

\be \frac{\pd}{\pd v} \circ
\begin{pmatrix}
\frac{\pd}{\pd u}  \\ \frac{\pd}{\pd v} \\
\frac{\pd}{\pd w} \\ \frac{\pd}{\pd x} \end{pmatrix}
= \begin{pmatrix}
0 & \frac{1}{u} & -\frac{v}{u^2}f_0'''(\frac{v}{u}) &
-\frac{w}{u^2} -\frac{v^2}{u^3}f_0'''(\frac{v}{u})  \\
0 & 0 & \frac{1}{u}f_0'''(\frac{v}{u}) &
\frac{v}{u^2}f_0'''(\frac{v}{u}) \\
0 & 0 & 0 & \frac{1}{u} \\
0 & 0 & 0 & 0
\end{pmatrix}
\cdot
\begin{pmatrix}
\frac{\pd}{\pd u}  \\ \frac{\pd}{\pd v} \\
\frac{\pd}{\pd w} \\ \frac{\pd}{\pd x} \end{pmatrix}.
\ee

\be \frac{\pd}{\pd w} \circ
\begin{pmatrix}
\frac{\pd}{\pd u}  \\ \frac{\pd}{\pd v} \\
\frac{\pd}{\pd w} \\ \frac{\pd}{\pd x} \end{pmatrix}
= \begin{pmatrix}
0 & 0 & \frac{1}{u} & -\frac{v}{u^2}   \\
0 & 0 & 0 & \frac{1}{u}   \\
0 & 0 & 0 & 0 \\
0 & 0 & 0 & 0
\end{pmatrix}
\cdot
\begin{pmatrix}
\frac{\pd}{\pd u}  \\ \frac{\pd}{\pd v} \\
\frac{\pd}{\pd w} \\ \frac{\pd}{\pd x} \end{pmatrix}.
\ee

\be \frac{\pd}{\pd x} \circ
\begin{pmatrix}
\frac{\pd}{\pd u}  \\ \frac{\pd}{\pd v} \\
\frac{\pd}{\pd w} \\ \frac{\pd}{\pd x} \end{pmatrix}
= \begin{pmatrix}
0 & 0 & 0 & \frac{1}{u}   \\
0 & 0 & 0 & 0 \\
0 & 0 & 0 & 0 \\
0 & 0 & 0 & 0
\end{pmatrix}
\cdot
\begin{pmatrix}
\frac{\pd}{\pd u}  \\ \frac{\pd}{\pd v} \\
\frac{\pd}{\pd w} \\ \frac{\pd}{\pd x} \end{pmatrix}.
\ee
Denote by $C_u$, $C_v$, $C_w$ and $C_x$ the $4 \times 4$-matrices
on the right-hand sides of these equalities respectively.

\begin{thm}
For the quintic we have
\be
[\nabla_{\frac{\pd}{\pd \alpha}}, C_\beta] = [\nabla_{\frac{\pd}{\pd \beta}}, C_\alpha]
\ee
for $\alpha, \beta = u,v,w,x$.
\end{thm}

\begin{proof}
These can be verified by explicit computations below:
\ben
&& [\nabla_{\frac{\pd}{\pd w}}, C_x] =\frac{\pd}{\pd w}\begin{pmatrix}
0 & 0 & 0 & \frac{1}{u}   \\
0 & 0 & 0 & 0 \\
0 & 0 & 0 & 0 \\
0 & 0 & 0 & 0
\end{pmatrix}
=\begin{pmatrix}
0 & 0 & 0 & 0 \\
0 & 0 & 0 & 0 \\
0 & 0 & 0 & 0 \\
0 & 0 & 0 & 0
\end{pmatrix}, \\
&& [\nabla_{\frac{\pd}{\pd x}}, C_w]
= \frac{\pd}{\pd x}\begin{pmatrix}
0 & 0 & \frac{1}{u} & -\frac{v}{u^2}   \\
0 & 0 & 0 & \frac{1}{u}   \\
0 & 0 & 0 & 0 \\
0 & 0 & 0 & 0
\end{pmatrix}
= \begin{pmatrix}
0 & 0 & 0 & 0 \\
0 & 0 & 0 & 0 \\
0 & 0 & 0 & 0 \\
0 & 0 & 0 & 0
\end{pmatrix}, \een

\ben
[\nabla_{\frac{\pd}{\pd v}}, C_x]
& = & \frac{\pd}{\pd v} \begin{pmatrix}
0 & 0 & 0 & \frac{1}{u}   \\
0 & 0 & 0 & 0 \\
0 & 0 & 0 & 0 \\
0 & 0 & 0 & 0
\end{pmatrix}
=\begin{pmatrix}
0 & 0 & 0 & 0 \\
0 & 0 & 0 & 0 \\
0 & 0 & 0 & 0 \\
0 & 0 & 0 & 0
\end{pmatrix}, \een
\ben
[\nabla_{\frac{\pd}{\pd x}}, C_v]
& = & \frac{\pd}{\pd x} \begin{pmatrix}
0 & \frac{1}{u} & -\frac{v}{u^2}f_0'''(\frac{v}{u}) &
-\frac{w}{u^2} -\frac{v^2}{u^3}f_0'''(\frac{v}{u})  \\
0 & 0 & \frac{1}{u}f_0'''(\frac{v}{u}) &
\frac{v}{u^2}f_0'''(\frac{v}{u}) \\
0 & 0 & 0 & \frac{1}{u} \\
0 & 0 & 0 & 0
\end{pmatrix} \\
& = & \begin{pmatrix}
0 & 0 & 0 & 0 \\
0 & 0 & 0 & 0 \\
0 & 0 & 0 & 0 \\
0 & 0 & 0 & 0
\end{pmatrix}.
\een

\ben
&& [\nabla_{\frac{\pd}{\pd u}}, C_x]
= \frac{\pd}{\pd u}
\begin{pmatrix}
0 & 0 & 0 & \frac{1}{u}   \\
0 & 0 & 0 & 0 \\
0 & 0 & 0 & 0 \\
0 & 0 & 0 & 0
\end{pmatrix}
= \begin{pmatrix}
0 & 0 & 0 & -\frac{1}{u^2}   \\
0 & 0 & 0 & 0 \\
0 & 0 & 0 & 0 \\
0 & 0 & 0 & 0
\end{pmatrix},
\een
\ben
[\nabla_{\frac{\pd}{\pd x}}, C_u]
& = & \frac{\pd}{\pd x}\begin{pmatrix}
\frac{1}{u} & -\frac{v}{u^2}  & - \frac{w}{u^2} +\frac{v^2}{u^3}f_0'''(\frac{v}{u})
&  \frac{v^3}{u^4}f_0'''(\frac{v}{u})-\frac{x}{u^2}+ \frac{2vw}{u^3} \\
0 & \frac{1}{u} & -\frac{v}{u^2}f_0'''(\frac{v}{u}) &
-\frac{w}{u^2} -\frac{v^2}{u^3}f_0'''(\frac{v}{u})   \\
0 & 0 & \frac{1}{u} & - \frac{v}{u^2} \\
0 & 0 & 0 & \frac{1}{u}
\end{pmatrix} \\
&= &\begin{pmatrix}
0 & 0 & 0 & -\frac{1}{u^2}  \\
0 & 0 & 0 & 0 \\
0 & 0 & 0 & 0 \\
0 & 0 & 0 & 0
\end{pmatrix}
\een

\ben
[\nabla_{\frac{\pd}{\pd v}}, C_w]
& = & \frac{\pd}{\pd v} \begin{pmatrix}
0 & 0 & \frac{1}{u} & -\frac{v}{u^2}   \\
0 & 0 & 0 & \frac{1}{u}   \\
0 & 0 & 0 & 0 \\
0 & 0 & 0 & 0
\end{pmatrix}
= \begin{pmatrix}
0 & 0 & 0 & -\frac{1}{u^2}   \\
0 & 0 & 0 & 0   \\
0 & 0 & 0 & 0 \\
0 & 0 & 0 & 0
\end{pmatrix}, \een
\ben
[\nabla_{\frac{\pd}{\pd w}}, C_v]
& = & \frac{\pd}{\pd w} \begin{pmatrix}
0 & \frac{1}{u} & -\frac{v}{u^2}f_0'''(\frac{v}{u}) &
-\frac{w}{u^2} -\frac{v^2}{u^3}f_0'''(\frac{v}{u})  \\
0 & 0 & \frac{1}{u}f_0'''(\frac{v}{u}) &
\frac{v}{u^2}f_0'''(\frac{v}{u}) \\
0 & 0 & 0 & \frac{1}{u} \\
0 & 0 & 0 & 0
\end{pmatrix} \\
& = & \begin{pmatrix}
0 & 0 & 0 & -\frac{1}{u^2}  \\
0 & 0 & 0 & 0 \\
0 & 0 & 0 & 0 \\
0 & 0 & 0 & 0
\end{pmatrix}.
\een

\ben
\nabla_{\frac{\pd}{\pd u}} C_w
& = & \frac{\pd}{\pd u} \begin{pmatrix}
0 & 0 & \frac{1}{u} & -\frac{v}{u^2}   \\
0 & 0 & 0 & \frac{1}{u}   \\
0 & 0 & 0 & 0 \\
0 & 0 & 0 & 0
\end{pmatrix}
= \begin{pmatrix}
0 & 0 & -\frac{1}{u^2} & \frac{2v}{u^3}   \\
0 & 0 & 0 & -\frac{1}{u^2} \\
0 & 0 & 0 & 0 \\
0 & 0 & 0 & 0
\end{pmatrix},
\een

\ben
\nabla_{\frac{\pd}{\pd w}} C_u
& = & \frac{\pd}{\pd w} \begin{pmatrix}
\frac{1}{u} & -\frac{v}{u^2}  & - \frac{w}{u^2} +\frac{v^2}{u^3}f_0'''(\frac{v}{u})
&  \frac{v^3}{u^4}f_0'''(\frac{v}{u})-\frac{x}{u^2}+ \frac{2vw}{u^3} \\
0 & \frac{1}{u} & -\frac{v}{u^2}f_0'''(\frac{v}{u}) &
-\frac{w}{u^2} -\frac{v^2}{u^3}f_0'''(\frac{v}{u})   \\
0 & 0 & \frac{1}{u} & - \frac{v}{u^2} \\
0 & 0 & 0 & \frac{1}{u}
\end{pmatrix} \\
& = & \begin{pmatrix}
0 & 0 &  -\frac{1}{u^2} & \frac{2v}{u^3}  \\
0 & 0 & 0 & -\frac{1}{u^2} \\
0 & 0 & 0 & 0 \\
0 & 0 & 0 & 0
\end{pmatrix}.
\een

\ben
&& \nabla_{\frac{\pd}{\pd u}} C_v
=\frac{\pd}{\pd u} \begin{pmatrix}
0 & \frac{1}{u} & -\frac{v}{u^2}f_0'''(\frac{v}{u}) &
-\frac{w}{u^2} -\frac{v^2}{u^3}f_0'''(\frac{v}{u})  \\
0 & 0 & \frac{1}{u}f_0'''(\frac{v}{u}) &
\frac{v}{u^2}f_0'''(\frac{v}{u}) \\
0 & 0 & 0 & \frac{1}{u} \\
0 & 0 & 0 & 0
\end{pmatrix} \\
& = & \begin{pmatrix}
0 & -\frac{1}{u^2} & \frac{2v}{u^3}f_0'''(\frac{v}{u})+\frac{v^2}{u^4}f_0^{(4)}(\frac{v}{u}) &
\frac{2w}{u^3}+\frac{3v^2}{u^4}f_0'''(\frac{v}{u}) + \frac{v^3}{u^5}f_0^{(4)}(\frac{v}{u}) \\
0 & 0 & -\frac{1}{u^2}f_0'''(\frac{v}{u})- \frac{v}{u^3}f_0^{(4)}(\frac{v}{u})
& -\frac{2v}{u^3}f_0'''(\frac{v}{u})-\frac{v^2}{u^4}f_0^{(4)}(\frac{v}{u}) \\
0 & 0 & 0 & -\frac{1}{u^2}  \\
0 & 0 & 0 & 0
\end{pmatrix},
\een

\ben
&& \nabla_{\frac{\pd}{\pd v}} C_u
= \frac{\pd}{\pd v} \begin{pmatrix}
\frac{1}{u} & -\frac{v}{u^2}  & - \frac{w}{u^2} +\frac{v^2}{u^3}f_0'''(\frac{v}{u})
&  \frac{v^3}{u^4}f_0'''(\frac{v}{u})-\frac{x}{u^2}+ \frac{2vw}{u^3} \\
0 & \frac{1}{u} & -\frac{v}{u^2}f_0'''(\frac{v}{u}) &
-\frac{w}{u^2} -\frac{v^2}{u^3}f_0'''(\frac{v}{u})   \\
0 & 0 & \frac{1}{u} & - \frac{v}{u^2} \\
0 & 0 & 0 & \frac{1}{u}
\end{pmatrix} \\
& = &  \begin{pmatrix}
0 & -\frac{1}{u^2}  &  \frac{2v}{u^3}f_0'''(\frac{v}{u}) + \frac{v^2}{u^4} f_0^{(4)}(\frac{v}{u})
&  \frac{3v^2}{u^4}f_0'''(\frac{v}{u}) + \frac{v^3}{u^5}f_0^{(4)}(\frac{v}{u}) + \frac{2w}{u^3} \\
0 & 0 & -\frac{1}{u^2}f_0'''(\frac{v}{u})-\frac{v}{u^3}f_0^{(4)}(\frac{v}{u}) &
-\frac{2v}{u^3}f_0'''(\frac{v}{u}) -\frac{v^2}{u^4}f_0^{(4)}(\frac{v}{u})  \\
0 & 0 & 0 & - \frac{1}{u^2} \\
0 & 0 & 0 & 0
\end{pmatrix} \\
\een
\end{proof}

\subsection{The $\bar{C}$'s}

To understand the emergent $tt^*$-geometry,
let us first compute the action of the real structure on
the basis $\frac{\pd}{\pd u}$, $\frac{\pd}{\pd v}$, $\frac{\pd}{\pd w}$ and $\frac{\pd}{\pd x}$.
Note
\ben
&& \frac{\pd}{\pd x} = \frac{1}{u} \frac{\pd}{\pd t^S} = \frac{1}{u} s_S,
\een
and so
\ben
s_S = u \frac{\pd}{\pd x}.
\een
Now we get:
\be
\tau(\frac{\pd}{\pd x}) = \frac{1}{\bar{u}} s_S = \frac{u}{\bar{u}} \frac{\pd}{\pd x}.
\ee
Similarly,

\ben
&& \frac{\pd}{\pd w} = \frac{1}{u} \frac{\pd}{\pd t^R}
- \frac{v}{u^2} \frac{\pd}{\pd t^S}
= \frac{1}{u} (s_R+ \frac{v}{u} s_S)  - \frac{v}{u^2} s_S
= \frac{1}{u} s_R,
\een
so we get
\be
s_R = u \frac{\pd}{\pd w}
\ee
and
\be
\tau(\frac{\pd}{\pd w}) = \frac{1}{\bar{u}}s_R = \frac{u}{\bar{u}} \frac{\pd}{\pd w}.
\ee
The formula $\frac{\pd}{\pd v}$ involves more terms:
\ben
\frac{\pd}{\pd v} & = & \frac{1}{u} \frac{\pd}{\pd t^Q}
- \frac{w}{u^2} \frac{\pd}{\pd t^S} \\
& = & \frac{1}{u} \biggl( s_Q  + \frac{\pd^2 f_0(t^Q)}{\pd (t^Q)^2}s_R
+ \biggl(t^Q \frac{\pd^2 f_0(t^Q)}{\pd (t^Q)^2}- \frac{\pd f_0(t^Q)}{\pd t^Q}\biggr) s_S\biggr)
- \frac{w}{u^2} s_S \\
& = & \frac{1}{u} s_Q + \frac{1}{u} f_0''(\frac{v}{u}) s_R
+ \biggl(\frac{v}{u^2} f_0''(\frac{v}{u})- \frac{1}{u}f_0'(\frac{v}{u})
- \frac{w}{u^2} \biggr) s_S,
\een
and so
\ben
s_Q & = & u \frac{\pd}{\pd v} - f_0''(\frac{v}{u}) s_R
- \biggl(\frac{v}{u} f_0''(\frac{v}{u})- f_0'(\frac{v}{u})- \frac{w}{u} \biggr) s_S \\
& = & u \frac{\pd}{\pd v} -u f_0''(\frac{v}{u}) \frac{\pd}{\pd w}
- \biggl(v f_0''(\frac{v}{u})- u f_0'(\frac{v}{u})- w \biggr) \frac{\pd}{\pd x}.
\een

\ben
\tau(\frac{\pd}{\pd v})
& = &\frac{1}{\bar{u}} s_Q + \frac{1}{\bar{u}} \overline{f_0''(\frac{v}{u})} s_R
+ \biggl(\frac{\bar{v}}{\bar{u}^2} \overline{f_0''(\frac{v}{u})}
- \frac{1}{\bar{u}}\overline{f_0'(\frac{v}{u})}- \frac{\bar{w}}{\bar{u}^2} \biggr) s_S \\
& = & \frac{1}{\bar{u}} \biggl(u \frac{\pd}{\pd v} -u f_0''(\frac{v}{u}) \frac{\pd}{\pd w}
- \biggl(v f_0''(\frac{v}{u})- u f_0'(\frac{v}{u})- w \biggr) \frac{\pd}{\pd x}\biggr) \\
& + & \frac{1}{\bar{u}} \overline{f_0''(\frac{v}{u})} \cdot u \frac{\pd}{\pd w}
+ \biggl(\frac{\bar{v}}{\bar{u}^2} \overline{f_0''(\frac{v}{u})}
- \frac{1}{\bar{u}}\overline{f_0'(\frac{v}{u})}- \frac{\bar{w}}{\bar{u}^2} \biggr) u\frac{\pd}{\pd x} \\
& = & \frac{u}{\bar{u}} \frac{\pd}{\pd v}
- \frac{u}{\bar{u}} \biggl(f_0''(\frac{v}{u})- \overline{f_0''(\frac{v}{u})}
\biggr)\frac{\pd}{\pd w} \\
& - & \frac{u}{\bar{u}}
\biggl(\frac{v}{u} f_0''(\frac{v}{u})- f_0'(\frac{v}{u})- \frac{w}{u}
- \overline{\frac{v}{u} f_0''(\frac{v}{u})- f_0'(\frac{v}{u})- \frac{w}{u} }\biggr)
\frac{\pd}{\pd x}.
\een
Finally,
\ben
\frac{\pd}{\pd u}
& = & \frac{1}{u} \frac{\pd}{\pd t^P}
- \frac{v}{u^2}\frac{\pd}{\pd t^Q}
- \frac{w}{u^2}\frac{\pd}{\pd t^R}
- (\frac{x}{u^2} - \frac{2vw}{u^3})
\frac{\pd}{\pd t^S} \\
& = & \frac{1}{u} \biggl(s_P  + t^Qs_Q
+ \frac{\pd f_0(t^Q)}{\pd t^Q} s_R
+ \biggl(t^Q\frac{\pd f_0(t^Q)}{\pd t^Q} - 2 f_0(t^Q) \biggr)
s_S \biggr) \\
& - & \frac{v}{u^2} \biggl( s_Q  + \frac{\pd^2 f_0(t^Q)}{\pd (t^Q)^2}s_R
+ \biggl(t^Q \frac{\pd^2 f_0(t^Q)}{\pd (t^Q)^2}- \frac{\pd f_0(t^Q)}{\pd t^Q}\biggr) s_S
\biggr) \\
& - & \frac{w}{u^2} \biggl( s_R + t^Q s_S \biggr)
- (\frac{x}{u^2} - \frac{2vw}{u^3})
\frac{\pd}{\pd t^S} \\
& = & \frac{1}{u} s_P
- \biggl(\frac{v}{u^2} f_0''(\frac{v}{u})- \frac{1}{u} f_0'(\frac{v}{u})
+ \frac{w}{u^2} \biggr) s_R \\
& - & \biggl( \frac{v^2}{u^3} f_0''(\frac{v}{u})
- \frac{2v}{u^2}f_0'(\frac{v}{u})
+ \frac{2}{u} f_0(\frac{v}{u})
+\frac{x}{u^2} - \frac{v w}{u^3} \biggr) s_S.
\een

\ben
s_P
& = & u \biggl[\frac{\pd}{\pd u}
+ \biggl(\frac{v}{u^2} f_0''(\frac{v}{u})- \frac{1}{u} f_0'(\frac{v}{u})
+ \frac{w}{u^2} \biggr) s_R \\
& + & \biggl( \frac{v^2}{u^3} f_0''(\frac{v}{u})
- \frac{2v}{u^2}f_0'(\frac{v}{u})
+ \frac{2}{u} f_0(\frac{v}{u})
+\frac{x}{u^2} - \frac{v w}{u^3} \biggr) s_S \biggr] \\
& = & u \biggl[\frac{\pd}{\pd u}
+ \biggl(\frac{v}{u^2} f_0''(\frac{v}{u})- \frac{1}{u} f_0'(\frac{v}{u})
+ \frac{w}{u^2} \biggr) u\frac{\pd}{\pd w} \\
& + & \biggl( \frac{v^2}{u^3} f_0''(\frac{v}{u})
- \frac{2v}{u^2}f_0'(\frac{v}{u})
+ \frac{2}{u} f_0(\frac{v}{u})
+\frac{x}{u^2} - \frac{v w}{u^3} \biggr) u\frac{\pd}{\pd x} \biggr].
\een

\ben
\tau(\frac{\pd}{\pd u})
& = & \frac{1}{\bar{u}} s_P
- \overline{\biggl(\frac{v}{u^2} f_0''(\frac{v}{u})- \frac{1}{u} f_0'(\frac{v}{u})
+ \frac{w}{u^2} \biggr)} s_R \\
& - & \overline{\biggl( \frac{v^2}{u^3} f_0''(\frac{v}{u})
- \frac{2v}{u^2}f_0'(\frac{v}{u})
+ \frac{2}{u} f_0(\frac{v}{u})
+\frac{x}{u^2} - \frac{v w}{u^3} \biggr)} s_S \\
& = & \frac{u}{\bar{u}} \biggl[\frac{\pd}{\pd u}
+ \biggl(\frac{v}{u^2} f_0''(\frac{v}{u})- \frac{1}{u} f_0'(\frac{v}{u})
+ \frac{w}{u^2} \biggr) u\frac{\pd}{\pd w} \\
& + & \biggl( \frac{v^2}{u^3} f_0''(\frac{v}{u})
- \frac{2v}{u^2}f_0'(\frac{v}{u})
+ \frac{2}{u} f_0(\frac{v}{u})
+\frac{x}{u^2} - \frac{v w}{u^3} \biggr) u\frac{\pd}{\pd x} \biggr] \\
& - & \overline{\biggl(\frac{v}{u^2} f_0''(\frac{v}{u})- \frac{1}{u} f_0'(\frac{v}{u})
+ \frac{w}{u^2} \biggr)} u \frac{\pd}{\pd w} \\
& - & \overline{\biggl( \frac{v^2}{u^3} f_0''(\frac{v}{u})
- \frac{2v}{u^2}f_0'(\frac{v}{u})
+ \frac{2}{u} f_0(\frac{v}{u})
+\frac{x}{u^2} - \frac{v w}{u^3} \biggr)} u \frac{\pd}{\pd x}.
\een
To summarize,
we get:
\ben
\tau\begin{pmatrix}
\frac{\pd}{\pd u} \\ \frac{\pd}{\pd v} \\
\frac{\pd}{\pd w} \\ \frac{\pd}{\pd x}
\end{pmatrix}
=M \begin{pmatrix}
\frac{\pd}{\pd u} \\ \frac{\pd}{\pd v} \\
\frac{\pd}{\pd w} \\ \frac{\pd}{\pd x}
\end{pmatrix}
=  \frac{u}{\bar{u}} \begin{pmatrix}
1 & 0 & a - \bar{a} & b - \bar{b} \\
0 & 1 & c - \bar{c} & d - \bar{d} \\
0 & 0 & 1 & 0 \\
0 & 0 & 0 & 1
\end{pmatrix} \cdot
\begin{pmatrix}
\frac{\pd}{\pd u} \\ \frac{\pd}{\pd v} \\
\frac{\pd}{\pd w} \\ \frac{\pd}{\pd x}
\end{pmatrix},
\een
where
\ben
&& a = \frac{v}{u} f_0''(\frac{v}{u})- f_0'(\frac{v}{u})
+ \frac{w}{u}, \\
&& b = \frac{v^2}{u^2} f_0''(\frac{v}{u})
- \frac{2v}{u}f_0'(\frac{v}{u})
+ 2f_0(\frac{v}{u})
+\frac{x}{u} - \frac{v w}{u^2}, \\
&& c = - f_0''(\frac{v}{u}), \\
&& d = - \biggl(\frac{v}{u} f_0''(\frac{v}{u})-f_0'(\frac{v}{u})-\frac{w}{u}\biggr).
\een

It is easy to check that:
\be
\overline{M}M = I.
\ee

\ben
\tau C_x \tau \begin{pmatrix}
\frac{\pd}{\pd u} \\ \frac{\pd}{\pd v} \\
\frac{\pd}{\pd w} \\ \frac{\pd}{\pd x}
\end{pmatrix}
& = & M\overline{C_x} \cdot \overline{M} \begin{pmatrix}
\frac{\pd}{\pd u} \\ \frac{\pd}{\pd v} \\
\frac{\pd}{\pd w} \\ \frac{\pd}{\pd x}
\end{pmatrix}
=\tau
\begin{pmatrix}
\frac{1}{\bar{u}} \frac{\pd}{\pd x} \\ 0 \\ 0 \\ 0
\end{pmatrix}
= \begin{pmatrix}
\frac{1}{u} \cdot \frac{u}{\bar{u}} \frac{\pd}{\pd x} \\ 0 \\ 0 \\ 0
\end{pmatrix} \\
& = & \begin{pmatrix}
0 & 0 & 0 & \frac{1}{\bar{u}}   \\
0 & 0 & 0 & 0 \\
0 & 0 & 0 & 0 \\
0 & 0 & 0 & 0
\end{pmatrix}
\cdot \begin{pmatrix}
\frac{\pd}{\pd u} \\ \frac{\pd}{\pd v} \\
\frac{\pd}{\pd w} \\ \frac{\pd}{\pd x}
\end{pmatrix},
\een

\be
\bar{C}_{\bar{x}}
= \begin{pmatrix}
0 & 0 & 0 & \frac{1}{\bar{u}}   \\
0 & 0 & 0 & 0 \\
0 & 0 & 0 & 0 \\
0 & 0 & 0 & 0
\end{pmatrix}
\ee

\be
\bar{C}_{\bar{w}} %% =M\overline{C_w}\bar{M}
= \begin{pmatrix}
0 & 0 & \frac{1}{\bar{u}} & -\frac{\bar{v}}{\bar{u}^2}   \\
0 & 0 & 0 & \frac{1}{\bar{u}}   \\
0 & 0 & 0 & 0 \\
0 & 0 & 0 & 0
\end{pmatrix}
\ee

\be
\bar{C}_{\bar{v}} %% =M\overline{C_v}\bar{M}
= \begin{pmatrix}
0 & \frac{1}{\bar{u}} & -\frac{\bar{v}}{\bar{u}^2}\overline{f_0'''(\frac{v}{u})}
-\frac{C}{\bar{u}} &
-\frac{\bar{w}}{\bar{u}^2} -\frac{\bar{v}^2}{\bar{u}^3}\overline{f_0'''(\frac{v}{u})}
+ \frac{A-D}{\bar{u}}  \\
0 & 0 & \frac{1}{\bar{u}}\overline{f_0'''(\frac{v}{u})} &
\frac{\bar{v}}{\bar{u}^2}\overline{f_0'''(\frac{v}{u})}
+\frac{C}{\bar{u}}   \\
0 & 0 & 0 & \frac{1}{\bar{u}}  \\
0 & 0 & 0 & 0
\end{pmatrix}
\ee
\be
\bar{C}_{\bar{u}} %% =M\overline{C_u}\bar{M}
= \begin{pmatrix}
\frac{1}{\bar{u}} & -\frac{\bar{v}}{\bar{u}^2}&
-\frac{\bar{w}}{\bar{u}^2}+\frac{\bar{v}^2}{\bar{u}^3}\overline{f_0'''(\frac{v}{u})}
+\frac{C\bar{v}}{\bar{u}^2} &
-\frac{\bar{x}}{\bar{u}^2} + \frac{\bar{v}^3}{\bar{u}^4}\overline{f_0'''(\frac{v}{u})}
+ \frac{2\bar{v}\bar{w}}{\bar{u}^3}-\frac{A-D}{\bar{u}^2}  \\
0 & \frac{1}{\bar{u}} & -\frac{\bar{v}}{\bar{u}}\overline{f_0'''(\frac{v}{u})} &
-\frac{\bar{w}}{\bar{u}^2}-\frac{\bar{v}^2}{\bar{u}^3}\overline{f_0'''(\frac{v}{u})}
-\frac{C\bar{v}}{\bar{u}^2}  \\
0 & 0 & \frac{1}{\bar{u}}  & -\frac{\bar{v}}{\bar{u}^2} \\
0 & 0 & 0 & \frac{1}{\bar{u}}
\end{pmatrix}
\ee
where $A=a-\bar{a}$, etc.
One can check that
\be
[C_\alpha, \bar{C}_{\bar{\beta}} ] = 0
\ee
except for
\ben
&& [C_v, \bar{C}_{\bar{v}}]
= \frac{1}{|u|^2} \begin{pmatrix}
0 & 0 & \overline{f_0'''(\frac{v}{u})}-f_0'''(\frac{v}{u}) &
 \overline{\frac{v}{u}f_0'''(\frac{v}{u})}-\frac{v}{u} f_0'''(\frac{v}{u})  \\
0 & 0 & 0 &  f_0'''(\frac{v}{u})-\overline{f_0'''(\frac{v}{u})}  \\
0 & 0 & 0 & 0 \\
0 & 0 & 0 & 0
\end{pmatrix}
\een
and three other cases:
$[C_v, \bar{C}_{\bar{u}}]$,  $[C_u, \bar{C}_{\bar{v}}]$ and $[C_u, \bar{C}_{\bar{u}}]$.

\subsection{The change of coordinate on the complexified K\"ahler moduli space}

The variable $t^Q$ is understood as the coordinate
on the complexified K\"ahler moduli space.
Introduce a new coordinate
\be
x=\exp (2\pi i t^Q).
\ee
Then we have
\be
\frac{\pd}{\pd t^Q} = x \frac{\pd}{\pd x}.
\ee
When one circles around $x=0$ in the $x$-plane,
$t^Q$ is changed to $t^Q+1$.
The flat sections undergo the following changes:
\ben
\tilde{s}_S & = & \frac{\pd}{\pd t^S} = s_S, \\
\tilde{s}_R & = & \frac{\pd}{\pd t^R} - (t^Q+1) \frac{\pd}{\pd t^S}
= s_R -  s_S, \\
\tilde{s}_Q & = & \frac{\pd}{\pd t^Q} - \frac{\pd^2 f_0(t^Q+1)}{\pd (t^Q)^2}
\frac{\pd}{\pd t^R} + \frac{\pd f_0(t^Q+1)}{\pd t^Q} \frac{\pd}{\pd t^S} \\
& = & s_Q - 5\frac{\pd}{\pd t^R} + \frac{5}{2} (2t^Q+1) \frac{\pd}{\pd t^S} \\
& = & s_Q - 5 s_R + \frac{5}{2} s_S, \\
\tilde{s}_P & = & \frac{\pd}{\pd t^P} - (t^Q+1)\frac{\pd}{\pd t^Q}
+ \biggl( (t^Q+1) \frac{\pd^2f_0(t^Q+1)}{\pd (t^Q)^2}
-\frac{\pd f_0(t^Q+1)}{\pd t^Q} \biggr)  \frac{\pd}{\pd t^R} \\
& - & \biggl((t^Q+1)\frac{\pd f_0(t^Q+1)}{\pd t^Q} - 2 f_0(t^Q+1) \biggr)
\frac{\pd}{\pd t^S} \\
& = & s_P - \frac{\pd}{\pd t^Q}
+ \biggl(5t^Q+\frac{5}{2}\biggr) \frac{\pd}{\pd t^R}
- \biggl(\frac{5}{2}(t^Q)^2+\frac{5}{2}t^Q+\frac{5}{6}\biggr)\frac{\pd}{\pd t^S} \\
& = & s_P-s_Q+ \frac{5}{2}s_R - \frac{5}{6}s_S.
\een
Written in matrix form:
\be
\begin{pmatrix}
\tilde{s}_S \\ \tilde{s}_R \\ \tilde{s}_Q \\ \tilde{s}_P
\end{pmatrix}
= \begin{pmatrix}
1 & 0 & 0 & 0 \\
-1 & 1 & 0 & 0 \\
\frac{5}{2} & -5 & 1 & 0 \\
-\frac{5}{6} & \frac{5}{2} & -1 & 1
\end{pmatrix}
\cdot
\begin{pmatrix} s_S \\ s_R \\ s_Q \\ s_P \end{pmatrix}.
\ee
The monodromy matrix is the exponential of the following nilpotent matrix:
\be
\begin{pmatrix}
0 & 0 & 0 & 0 \\
-1 & 0 & 0 & 0 \\
0 & -5 & 0 & 0 \\
0 & 0 & -1 & 0
\end{pmatrix}.
\ee
This matrix is the matrix of the cup product with $-Q$:
\be
Q \cdot \begin{pmatrix} S \\ R \\ Q \\ P \end{pmatrix}
=\begin{pmatrix}
0 & 0 & 0 & 0 \\
1 & 0 & 0 & 0 \\
0 & 5 & 0 & 0 \\
0 & 0 & 1 & 0
\end{pmatrix}
\cdot \begin{pmatrix} S \\ R \\ Q \\ P \end{pmatrix}.
\ee

\section{Conclusions and Speculations}

\label{sec:Conclusions}

In our attempt to unify Witten Conjecture/Kontsevich Theorem with mirror symmetry
of the quintic,
we have found that methods of statistical physics to be very useful.
By the mean field theory developed by Dijkgraaf,  Eguchi and his collaborators,
we have written down the integrable hierarchy associated with the GW theory of the quintic.
This leads to the structure of Frobenius manifold developed by Dubrovin and his collaborators
on the small phase space of the quintic.
We also present the special K\"ahler geometry on the small phase.
In carrying out the explicit computations,
we have reduced to the computations of the free energy on the small phase space.
Note we are working on the $A$-side side,
but many geometric structures originally discovered on he $B$-theory side naturally emerge.
So we feel it is natural  to refer to them as emergent geometry.

Such results reveal the similarities and differences
between the GW theory of a point and the GW theory of the quintic.
And the similarities and differences
suggest some directions for future investigations
which we now present some speculations.
In both cases
we obtain integrable hierarchies.
In case of a point
the integrable hierarchy is the KdV hierarchy,
which is a reduction of the KP hierarchy.
This suggests that the integrable hierarchy
for the quintic might be a reduction
of the $4$-component KP hierarchy.
It will be interesting to check whether this is the case,
because in case of a point
one can identify the partition function as an
element in the fermionic Fock space
via the boson-fermion correspondence
via the Kac-Schwarz operator \cite{Kac-Sch}.
An explicit formula for the affine coordinates
of the corresponding element in the Sato Grassmannian
has been obtained by the author \cite{Zhou-Virasoro}
using the Virasoro constraints.
An alternative derivation using the KP hierarchy
was later given by Balogh-Yang \cite{Bal-Yan}.
A formula for the $n$-point function in  all genera
based on the affine coordinates of
an element in the Sato Grassmannian has been derived
by the author \cite{Zhou-Emergent}. In particular
it can be applied to the Witten-Kontsevich tau-function.
In a work in progress we will generalize this formula
to $n$-component KP hierarchy.

Another direction of search for unification
is the Eynard-Orantin topological  recursions \cite{EO}.
In \cite{Zhou-DVV} the author proved that the GW theory
of a point satisfies the EO topological recursion,
and in this case the recursion is equivalent to
the DVV Virasoro constraints.
Furthermore,
the spectral curve in this case lies in a family of
curves obtained by computing the genus zero
one-point function together with
a  Laplace transform.
The starting point of most of the computations
in this paper is the genus zero one-point function
and in the end we also take a Laplace transform
to get the deformed dual flat coordinates.
We speculate that this is not just a coincidence,
it means that the EO topological recursion
should also hold for the quintic.

\vspace{.2in}
{\bf Acknowledgements}.
The author is partly supported by NSFC grants 11661131005
and 11890662.

\end{document}